\documentclass[draftcls,onecolumn]{IEEEtran}
\pdfoutput=1
\usepackage[dvipsnames]{xcolor}%
\usepackage{everysel}
\usepackage{tikz}
\usetikzlibrary{shapes,arrows,fit,calc,positioning,automata}
\usepackage{pgfplots}
\usepackage{epstopdf}
\usepackage{graphicx}
\usepackage{epsf}
\usepackage[ruled]{algorithm2e}
\usepackage{caption}
\usetikzlibrary{fit}
\pgfplotsset{compat=newest} 
\pgfplotsset{plot coordinates/math parser=false} 
\usetikzlibrary{shapes,arrows}
\usetikzlibrary{shapes.geometric}
\usetikzlibrary{calc}

 \tikzset{%
  blocks/.style    = {draw, thick, rectangle, minimum height = 3em,
    minimum width = 3em},
  sum/.style      = {draw, circle, node distance = 2cm}, %
  input/.style    = {coordinate}, %
  output/.style   = {coordinate} %
  buffer/.style={
        draw,
        shape border rotate=180,
        regular polygon,
        regular polygon sides=3,
        fill=white,
        node distance=2cm,
        minimum height=4em
    }

}
\newif\ifNightMode
\NightModetrue
\NightModefalse
\ifNightMode
\EverySelectfont{\color{CornflowerBlue}}
\everymath{\color{LimeGreen}}
\color[rgb]{1,1,1}
\color{LimeGreen}
\else 
\fi
\usepackage{fancyhdr}
\usepackage{amsmath}
\usepackage{dsfont}
\usepackage{amssymb}
\usepackage{color}
\usepackage{graphicx}
\usepackage{amsfonts}
\usepackage{float}
\usepackage{url,graphicx,tabularx,array,geometry,amsmath,fullpage}

\usepackage{capt-of}%
\usepackage{booktabs}
\usepackage{varwidth}
\usepackage{enumerate}

\usepackage{amsmath, amsthm, amssymb}
\newtheorem{theorem}{Theorem}
\newtheorem{lemma}{Lemma}
\newtheorem{prop}{Proposition}
\newtheorem{remark}{Remark}
\newtheorem{example}{Example}

\newcommand{\bcar}{\begin{carlisthack}}
\newcommand{\ecar}{\end{carlisthack}}
\newcommand{\prob}[1]{\mathbb{P}\left[#1\right]}

\newcommand{\diag}{\ensuremath{\operatorname{diag}}}

\DeclareMathOperator{\range}{range}

\newcommand{\Exs}{\ensuremath{\mathbb{E}}}

\usepackage[symbol]{footmisc}

\usepackage[dvips,all]{xy}

\long\def\comment#1{}

\newcommand{\reals}{\mathbb{R}}

\SetKwComment{Comment}{$\triangleright$\ }{}

\usepackage{mathtools}

\DeclarePairedDelimiter{\floor}{\lfloor}{\rfloor}
\newcommand{\mean}{\mathrm{mean}}
\newcommand{\median}{\mathrm{median}}
\usepackage{todonotes}

\newcommand{\powerp}{\beta}

\newcommand{\margin}[1]{\marginpar{\color{red}\tiny\ttfamily#1}}
\newcommand{\note}[1]{**\margin{#1}}
\newcommand{\real}{\mathbb{R}}
\newcommand\numberthis{\addtocounter{equation}{1}\tag{\theequation}}
\newcommand{\eqindist}{\,{\buildrel d \over =}\,}
\newcommand{\nat}{\mathbb{N}}

\ifCLASSINFOpdf 
\else
\fi

\begin{document}
\title{~\\Computational Polarization:\\ An Information-theoretic Method for Resilient Computing}

\author{Mert~Pilanci,~\IEEEmembership{Member,~IEEE}
\thanks{Mert Pilanci is with the Department
of Electrical Engineering, Stanford University, Stanford,
CA, 94305 USA e-mail: \texttt{pilanci@stanfort.edu} (see \texttt{http://www.stanford.edu/$
\sim$pilanci}).}%
}

\maketitle

\begin{abstract}
We introduce an error resilient distributed computing method based on an extension of the channel polarization phenomenon to distributed algorithms. The method leverages an algorithmic split operation that transforms two identical compute nodes to slow and fast workers, which parallels the channel split operation in Polar Codes. This operation preserves the  average runtime, analogous to the conservation of Shannon capacity in channel polarization. By leveraging a recursive construction in a similar spirit to the Fast Fourier Transform, this method synthesizes virtual compute nodes with dispersed return time distributions, which we call computational polarization. We show that the runtime distributions form a functional martingale processes, identify their limiting distributions in closed-form expressions together with non-asymptotic convergence rates, and prove strong convergence results in Banach spaces. We provide an information-theoretic lower bound on the overall runtime of any coded computation method and show that the computational polarization approach asymptotically achieves the optimal runtime for computing linear functions. An important advantage is the near linear time decoding procedure, which is significantly cheaper than Maximum Distance Separable codes.
\end{abstract}

\IEEEpeerreviewmaketitle

\section{Introduction}
As a result of the recent growth of data, the computing paradigm has shifted into massively distributed computing systems. Several distributed architectures and software frameworks have been developed for large scale computational problems. Notable examples include the open source distributed computing framework Apache Spark \cite{zaharia2010spark} and the parallel programming framework MapReduce \cite{dean2008mapreduce}. However, as the scale of a computational cluster increases, failing nodes, heterogeneity, and unpredictable delays pose significant challenges. In particular, iterative optimization algorithms in machine learning such as gradient descent suffer from slower workers, since each iteration typically requires synchronization among the worker nodes. Such slow workers are referred to as stragglers, which are especially problematic in cheaper virtual machines running in the cloud at very large scales. Another problem arises in the security of data when the nodes are subject to adversarial interference. Data encoding mechanisms can provide a layer of security for sensitive datasets. 
Recently, concepts and tools from the coding theory were applied in distributed computation systems. This topic rapidly gained interest in the recent years (see, e.g., \cite{Lee2018,li2017fundamental,tandon2016gradient,li2020coded} and the references therein). Particularly, in \cite{Lee2018} the authors proposed applying erasure codes to matrix multiplication and data shuffling.

In this paper, we describe an information theoretic framework for error resilient computation based on the polarization phenomenon. We present a general mechanism to obtain reliable and faster computational nodes from unreliable and slow worker nodes.  A major and distinguishing advantage of the is the low complexity encoding and decoding operations due to the special structure of our construction. Our encoding and decoding operations consist of only addition and subtraction operations over the field of real numbers. Remarkably, solving linear systems over the reals or finite fields is not required in contrast to existing approaches. The encoding and decoding can be performed in $O(N \log N)$ time where $N$ is the number of worker nodes with serial computation. Moreover, the depth of the decoding complexity is $\log(N)$ in a straightforward parallel implementation.

\subsection{Overview of Our Results}
We consider a distributed computing framework with unreliable and occasionally failing worker nodes, where we model the worker job completion times as real valued random variables. In a similar fashion to the channel splitting operation in polar codes, we propose a computational split mechanism that transforms two identical workers into a fast and a slower worker. More precisely, the slower worker stochastically dominates the original worker, and the original worker stochastically dominates the faster worker. The computational split can be applied recursively to obtain \emph{virtual workers} which obtain progressively better (and worse) runtimes. We show that the distribution of computation times follow a functional martingale random process and establish its convergence properties in Banach spaces. We prove that the computation times polarize: the distribution of the run times approach a Dirac delta measure in a functional sense characterized by $L_p$ norms, which we fully characterize in closed form. In particular, we prove almost sure convergence in Banach spaces, which is an improvement to the existing analysis of Polar Codes, and identify non-asymptotic rates of convergence. We introduce several measures to order the virtual workers, which can be computed ahead of the time. Moreover, we show that several slow worker nodes can be \emph{frozen} according to any given order.  Consequently, straggler nodes can be eliminated by a simple freezing operation in order to achieve computational resilience with desired deadline considerations. We show that the proposed scheme achieves optimal overall runtime, which can be characterized in terms of the order statistics of the runtime distribution. We also prove the information-theoretic optimality of the proposed scheme, which can be viewed as the analogue of achieving channel capacity for unreliable distributed computing systems. %

\subsection{Prior Work}
Several works investigated fundamental trade-offs of redundancy and recovery and optimality of proposed constructions \cite{li2017fundamental}. The ideas were extended to matrix multiplication in higher dimensions \cite{yu2017polynomial}.  The authors in \cite{tandon2016gradient} proposed a coding scheme for distributed gradient descent. We refer the reader to \cite{li2020coded} for a recent overview. Polar Codes, invented by Arikan \cite{arikan2009channel} were a major breakthrough in coding and information theory. Their construction provided the first capacity achieving codes for binary input symmetric memoryless channels with an explicit construction and efficient encoding and decoding algorithms. Our work generalizes the polarization mechanism underlying Polar Codes in finite fields to real valued random variables. We extend the existing martingale analysis in Polar Codes to Banach spaces to obtain stronger convergence results, which can also be of interest to traditional Polar Codes. The application of traditional Polar Codes to distributed computation was first proposed in our recent work \cite{bartan2019straggler}.

\subsection{Notation}

 We use the notation $1[x\le y]$ to denote a zero-one valued indicator function which equals one whenever $x\le y$ and equals zero otherwise. The vector $1_n\in \real^{n}$ is a length $n$ vector of all ones. We use the notation $A(i,j)$ to denote entry $(i,j)$ of a matrix $A$. We extend scalar functions to vector valued functions entrywise. The symbol $A\otimes B$ is used for the Kronecker product of two matrices $A$ and $B$. We use $F(t)$ to denote cumulative distribution (or density) functions (CDFs), and $p(t)$ to denote probability density functions when they exist. We denote the continuous uniform distribution on $[a,b]$ by Uniform$[a,b]$. We use Uniform$\{i_1,\cdots i_K\}$ to denote the discrete uniform distribution on the discrete set $\{i_1,\cdots i_K\}$. Similarly, Exponential$(\lambda)$ represents the continuous exponential distribution with mean parameter $\lambda$. We use $X\eqindist Y$ when two random variables are equal in distribution, i.e., $\prob{X\le t}=\prob{Y\le t},\forall t\in\real$. We use the symbol $X\xrightarrow{d} P$ to denote that the random variable $X$ converges in distribution to $P$, where $P$ is a probability distribution. 

\section{Preliminaries}

\subsection{Problem Setting}
We now describe the stochastic setting we primarily employ in our framework. Suppose that we have $N$ nodes at our disposal which are workers operating on tasks in parallel. In the physical domain, these workers might correspond to different units of computation such as threads and cores in a single processor, multiple processors, graphical processing units, or multiple servers in a cluster. We model the computation times of worker nodes as nonnegative real valued random variables $T^{(1)},\cdots, T^{(N)} \in \mathbb{R}_{\ge 0}$. The goal of computational coding is assigning coded data blocks to computational nodes such that the result of the computation can be recovered in time preferably shorter than the maximum of the variables $T^{(1)},\cdots, T^{(N)}$. We restrict our attention to computing linear functions, which enables the use of linear error correcting codes in assigning data blocks to workers, which was first proposed in \cite{Lee2018}. The design of efficient and reliable coded computation schemes and their analysis have become an important research direction in the intersection of distributed computing, error correcting codes and information theory \cite{li2020coded}.

Our goal in this manuscript is to introduce a novel computational polarization phenomenon that creates synthetic runtime distributions, and leverage this property in computational coding. Our framework differs from the scheme investigated in \cite{Lee2018} in the encoding and decoding of the error correcting code, as well as in the analysis. Importantly, our framework is based on a mathematically rich functional generalization of channel polarization that arise in Polar Codes \cite{arikan2009channel} and complements existing results. Moreover, our encoding and decoding algorithms are considerably faster compared to other proposals in coded computing. We next provide a short review of simple uncoded and coded computing schemes in the sequel.

\subsection{Uncoded Computation}
Suppose that a data matrix $A$ is partitioned to $N$ local data blocks $A_1,\cdots, A_N \in \real^{m\times d}$. Consider a linear function $f(\cdot):\, \real^{m\times d}\rightarrow \real^r$ applied to the data blocks as $f(A_1),\cdots, f(A_N)$, whose evaluations are assigned as computational tasks for the $N$ distinct worker nodes. Then, the runtime of the conventional uncoded computing scheme is given by the maximum of $N$ random variables
\begin{align*}
 T_{\mathrm{\scriptsize uncoded}}:=\max_{k=1,\cdots N}\, T^{(k)}\,,
 \end{align*}
 which is the minimum time required to collect all the responses $f(A_1),\cdots,f(A_N)$.

We initially assume for simplicity that these random variables are independent and identically distributed. We postpone the discussion of heterogeneous workers to Section \ref{sec:non-identical} in which we discuss non-identical distributions. Suppose that the computation times are i.i.d. random variables distributed according to a cumulative distribution function (CDF) $F(t)$ such that
\begin{align*}
\prob{ T^{(k)} \le t} = F(t)\, \mbox{ for $k=1,\cdots,N$}.
\end{align*}
Under the i.i.d. assumption, the CDF of the runtime of the uncoded computation scheme is found to be
\begin{align*}
 \prob{T_{\mathrm{\scriptsize uncoded}}\le t} &= \prod_{k=1,\cdots,N} \prob{T^{(k)}\le t}\\ &= F(t)^N\,.
\end{align*}

We assume that the CDF $F(t)$ is known. When $F(t)$ is not known, it can be estimated from observations via parametric models. Also, one can also use the empirical distribution function $\hat F_n$ given by
\begin{align*}
\hat F_N(t) : = \frac{1}{N} \sum_{i=1}^N 1[T^{(N)}\le t]\,,
\end{align*}
instead of the CDF $F(t)$. In distributed computing applications, there is often a vast number of observations to model the CDF $F(t)$. In Section \ref{sec:empirical}, we illustrate both approaches using parametric models, as well as empirical distribution functions.

\subsection{Repetition Coding}
Repetition coding is commonly used in distributed storage and computation systems as a simple method to incorporate redundancy by replicating the tasks. An $\frac{N}{K}$ repetition code simply replicates each task $\frac{N}{K}$ times. Therefore, one only needs to obtain at least one replica of each task to accomplish the computation.   
A straightforward calculation shows that the runtime distribution of the repetition coded scheme, denoted by $T_{\mathrm{\scriptsize repetition}}$,  can be described by the CDF (see e.g., \cite{Lee2018})
\begin{align*}
\prob{T_{\mathrm{\scriptsize repetition}}\le t}=\left( 1-(1-F(t))^{N/K} \right)^K\,.
\end{align*}
The distribution of the maximum time is quite unsatisfactory in general. This is usually due to the heavy-tailed distributions of the return times in cloud systems. As the scale of computation exceeds several hundred worker nodes, the maximum return time can be impractical.

\subsection{Maximum-Distance Separable (MDS) Codes and Coded Computation}
MDS Codes are an important class of linear block codes that achieve equality in the Singleton bound \cite{macwilliams1977theory}. Examples of MDS codes include codes with a single parity symbol which are of distance 2, and codes comprised of only two codewords; the all-zero and the all-one sequences. These are often called trivial MDS codes. In the case of binary alphabets, it is well-known that only trivial MDS codes exist. In other alphabets, examples of non-trivial MDS codes include Reed-Solomon codes and their extensions. 

A linear code of length $N$ and rank $K$ is a linear subspace of the vector space $\mathbb{F}_q^N$ of rank $K$, where $\mathbb{F}_q$ is the finite field of $q$ elements. An $(N,K)$ linear code is an MDS code if and only if any $K$ columns of its generator matrix are linearly independent.

In MDS coded linear computation, the data matrix $A$ is first divided into $K$ equal-sized submatrices. An $(N,K)$ MDS code is applied to each element of the submatrices to obtain $N$ encoded submatrices $A_1^\prime,\cdots,A_N^\prime$. We define the rate of a coding scheme as $R:=\frac{K}{N}$. The worker compute nodes run coded tasks $f(A_1^\prime),\cdots,f(A_N^\prime)$. Due to the linear independence of the generator matrix and the linearity of the map $f(\cdot)$, one can recover $f(A)$ from any $K$ task results. Therefore, the runtime of the MDS coded computation is determined by the $K^{th}$ fastest response, in contrast to the uncoded scheme which is determined by the slowest response.

Now, we present probabilistic analysis of the runtime distribution. For i.i.d. runtime random variables $T^{(1)}$, $T^{(2)}$, ..., $T^{(N)}$ sampled from the CDF $F(t)$.
Then the runtime of the MDS coded scheme denoted by $T_{\mathrm{MDS}}$ is the  $K^{th}$ smallest value, i.e., $K^{th}$ order statistics, which follows the distribution
\begin{align*}
\prob{T_{\mathrm{MDS}}\le t} &= F_{X_{(K)}}(t) \\
&= \sum_{j=K}^N {N\choose j} [F(t)]^j [1-F(t)]^{N-j}\,.
\end{align*}
It is easy to see that $T_{\mathrm{MDS}}\le T_{\mathrm{repetition}}\le T_{\mathrm{uncoded}}$ with probability one. Note that in \cite{Lee2018}, it was assumed that the runtime distributions are scaled distributions $F_0(\ell t)$ of a mother distribution $F_0(t)$, where $\ell$ is the number of subtasks allocated to the workers. In all of our results, it is more natural to not use this assumption. However, one can substitute $F(t)=F_0(\ell t)$ to obtain the corresponding results.
\subsection{Polar Codes and Channel Polarization}
Polar Coding is an error-correcting code construction method that achieves the capacity of symmetric binary-input discrete memoryless channels, such as the binary symmetric channel and binary erasure channel \cite{arikan2009channel}. In this section, we briefly overview polar codes and their analysis techniques. A central operation in Polar Codes is the channel combination operation described by the $2\times 2$ linear transformation $P_2: \{0,1\}^2\rightarrow \{0,1\}^2$ over binary vectors of length two given by
\usetikzlibrary{shapes,arrows,fit,calc,positioning,automata}
\tikzstyle{int}     = [draw, minimum size=4em]
\tikzstyle{init}    = [pin edge={to-,thin,black}]
\tikzstyle{ADD} = [draw,circle]
\tikzstyle{line}    = [draw, -latex']

\newcommand*\nodedistance{4cm}
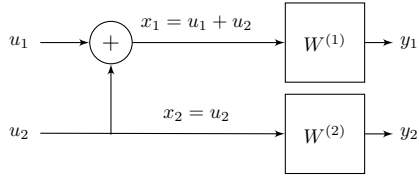
\begin{figure}[t!]
\begin{center}
\begin{tikzpicture}[node distance=\nodedistance,auto,>=latex', scale = 0.75, transform shape]
    \tikzstyle{line}=[draw, -latex']
    \node [int] (M1) { $W^{(1)}$ };
    \node [int, below=0.2 of M1] (M2) {$W^{(2)}$};
    \node [ADD,left=2.7 of M1] (ADD1) {\large $+$};
    \node [below=1.1 of ADD1] (ADD2) {};
    \node [below=-0.1 of ADD2] (ADD3) {};
    \node [above=0.1 of $(ADD1)!0.4!(M1)$] (MINUS) {$x_1=u_1 + u_2$};    
    \node [above=0.1 of $(ADD2)!0.4!(M2)$] (MINUS) {$x_2=u_2$};
    \node [right=0.5 of M1] (f1) {$y_1$};      
    \node [right=0.5 of M2] (f2) {$y_2$};      
    \node [left of=ADD1,node distance=2.0cm,text width=0.5cm,anchor=west,align=center] (A1) {$u_1$};
    \node [left of=ADD2,node distance=2.0cm,text width=0.5cm,anchor=west,align=center] (A2) {$u_2$};
    \path[line] (ADD1) edge (M1)
                (A1) edge (ADD1)
                (A2) edge (M2)
                (ADD3) edge (ADD1)
                (M1) edge (f1)
                (M2) edge (f2);
\end{tikzpicture}
\end{center}
\caption{Single-step channel transformation described by the linear map $P:\,\{0,1\}\rightarrow\{0,1\}$ described in \eqref{eq:polartransform}. Note that the arithmetic operations take place over the binary field $\mathbb{F}_2=\{0,1\}$. \label{fig:polartransform}}
\end{figure}

\begin{align}
P_2 := \left[ \begin{array}{cc} 1 & 0 \\ 1 & 1 \end{array}\right] \label{eq:polartransform}\,. 
\end{align}
Figure \ref{fig:polartransform} depicts the application of the transformation $P_2$ on the input data sequence $[u_1,u_2]$, which yields encoded sequence $[x_1,x_2]=[u_1+u_2,u_2]$. Here, $W$ is an arbitrary binary input symmetric channel described via the conditional probability distribution $W(y|x)$, where $x \in \{0,1\}$ and $y\in\mathcal{Y}$ are the input and output respectively, and $W^{(1)}$, $W^{(2)}$ are two independent, identically distributed copies of this channel. $\mathcal{Y}$ is the set of output values, e.g., $\mathcal{Y}=\{0,1,\mbox{erasure}\}$ for the Binary Erasure Channel. Note that the operations take place over the binary field $\mathbb{F}_2=\{0,1\}$, where addition is modulo $2$. The encoded sequence is presented as an input to the two identical copies of the symmetric binary-input memoryless channels $W^{(1)}$ and $W^{(2)}$. Consequently, two virtual and unequal channels $W^{-},W^{+}$ are constructed as follows. The first channel $W^{-}$ is for decoding $u_1$ from $y_1$ and $y_2$, where $u_2$ is unknown and regarded as external noise. The second channel $W^{+}$ is for decoding $u_2$ from $y_1,y_2$ and $u_1$, assuming that $u_1$ is available at the decoder. More precisely, the virtual channels are defined as
\begin{align}
W^{-}(y_1,y_2\,|\,u_1)&:= \sum_{u_2 \in \{0,1\}}\frac{1}{2} W(y_1\,|\,u_1+u_2)W(y_2\,|\,u_2) \label{eqpolarplus}\\
W^{+}(y_1,y_2,u_1\,|\,u_2)&:= \frac{1}{2} W(y_1\,|u_1+u_2)W(y_2\,|\,u_2) \label{eqpolarminus}\,.
\end{align}
It is easy to see that the total channel capacity is conserved, but redistributed into a better $I(W^{+})$ channel and worse $I(W^{-})$ channel in the following sense
\begin{align}
I(W^{-}) + I(W^{+}) &= 2 I(W) \label{eq:polarconservation}\\
 I(W^{-}) \le I(W) &\le  I(W^{+}) \label{eq:polarredist}\,,
\end{align}
where the last inequality holds with equality if and only if $I(W) \in \{0,1\}$.

A polar code of length $N=2^n$ is obtained from the linear embedding $P_n$, and encoding operation $x=P_n u$, where $P_n=P^{\otimes n}$, and the superscript $\otimes n$ denotes the $n^{th}$ Kronecker power, equivalently described by $P_n = \underbrace{P\otimes\cdots\otimes P}_{\scriptsize\mbox{n times}}$.

 Let $I(W)$ denote the symmetric Shannon capacity of the channel $W: \{0,1\}\rightarrow \mathcal{Y}$ defined as
\begin{align*}
I(W) := \sum_{x\in\{0,1\}} \sum_{y\in \mathcal{Y}} \frac{1}{2} W(y|x) \log_2 \frac{W(y|x)}{\sum_{x^\prime\in\{0,1\}}\frac{1}{2}W(y|x^\prime)}\,,
\end{align*}
which is the mutual information between the input and output of the channel $W$ when the input is uniformly distributed. As the construction size $N$ increases, the virtual channels $\{W_i\}_{i=1}^N$ polarize in the following sense. A fraction of $I(W)$ of the virtual channels approach perfect channels, i.e., $I(W_i)\approx 1$ whereas the remaining fraction of $1-I(W)$ approach pure noise channels, i.e., $I(W_i)\approx 0$. As shown in \cite{arikan2009channel}, one can send data at rate $1$ through the perfect channels while sending data at rate $0$, essentially \emph{freezing} the pure noise channels. Note that freezing a subset of the channels can be achieved by fixing certain entries of the input vector $u$ to zeros and using the remaining entries for data transmission. Polar Codes leverage this coding scheme enabled by the polarization phenomenon to achieve the symmetric capacity of any discrete memoryless channel. Moreover, the encoding and decoding complexity is $O(N\log N)$ thanks to the recursive construction of $P_n$.

\section{Computational Polarization}
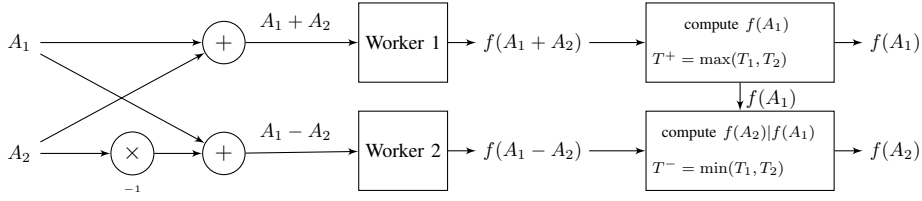
\begin{figure}[!t]
\centering
\usetikzlibrary{shapes,arrows,fit,calc,positioning,automata}
\tikzstyle{int}     = [draw, minimum size=4em]
\tikzstyle{init}    = [pin edge={to-,thin,black}]
\tikzstyle{ADD} = [draw,circle]
\tikzstyle{line}    = [draw, -latex']

\newcommand*\nodedistancepolar{4cm}
\begin{tikzpicture}[node distance=\nodedistancepolar,auto,>=latex', scale = 0.75, transform shape]
    \tikzstyle{line}=[draw, -latex']
    \node [int] (M1) { Worker 1 };
    \node [int, below=0.5 of M1] (M2) { Worker 2};
    \node [ADD,left=2 of M1] (ADD1) {\large $+$};
    \node [ADD,below=1.2 of ADD1] (ADD2) {\large $+$};
    \node [ADD,left=0.85 of ADD2] (NEG) {\large $\times$};
    \node [below=0.05 of NEG] (MINUS) {\tiny $-1$};
    \node [above=0.1 of $(ADD1)!0.4!(M1)$] (MINUS) {$A_1 + A_2$};    
    \node [above=0.1 of $(ADD2)!0.4!(M2)$] (MINUS) {$A_1 - A_2$};
    \node [right=0.5 of M1] (f1) {$f(A_1+A_2)$};      
    \node [right=0.5 of M2] (f2) {$f(A_1-A_2)$};      
    \node [left of=ADD1,node distance=4.0cm,text width=0.5cm,anchor=west,align=center] (A1) {$A_1$};
    \node [left of=ADD2,node distance=4.0cm,text width=0.5cm,anchor=west,align=center] (A2) {$A_2$};
    \node [int, right=1 of f1,text width=3.1cm] (DEC1) {\footnotesize \qquad compute $f(A_1)$\\  $T^+=\max(T_1,T_2)$};
    \node [int, right=1 of f2,text width=3.1cm] (DEC2) {\footnotesize \, compute $f(A_2)\vert f(A_1)$  \\$T^{-}=\min(T_1,T_2)$};
    \node [right=0.5 of DEC1] (DECO1) {$f(A_1)$};      
    \node [right=0.5 of DEC2] (DECO2) {$f(A_2)$}; 
    \node [below=0 of DEC1] (DEC1b) {\qquad\quad\,$f(A_1)$};  
    \path[line] (ADD1) edge (M1)
                (ADD2) edge (M2)
                (A1) edge (ADD1)
                (A2) edge (NEG)
                (NEG) edge (ADD2)
                (A2) edge (ADD1)
                (A1) edge (ADD2)
                (M1) edge (f1)
                (M2) edge (f2)
                (f1) edge (DEC1)
                (f2) edge (DEC2)
                (DEC1) edge (DECO1)
                (DEC2) edge (DECO2)
                (DEC1) edge (DEC2);
\end{tikzpicture}
\caption{The basic building block of computational polarization\label{fig:buildingblock}}
\end{figure}

In this section we present our main results. We begin by describing our setup and introducing the basic computational polarization operation that forms the basis of our method.

Let us consider a distributed computing task where a large data matrix $A$ and a function $f(\cdot)$ is given. Suppose that the data matrix is partitioned to submatrices over the rows as 
\begin{align*}
A = \left[ \begin{array}{c c c c} A_1^T & A_2^T & \hdots & A_K^T \end{array} \right]^T\,.
\end{align*}
This setting is very common in machine learning where each row corresponds to an individual sample. The splitting operation can also be performed over the columns to yield similar results.
Suppose that the function decomposes to element-wise evaluations on  submatrices as \begin{align*}
f(A) = \left[ \begin{array}{c c c c} f(A_1)^T & f(A_2)^T & \hdots & f(A_K)^T \end{array} \right]^T\,.
\end{align*}
Important examples of such functions include matrix multiplication with a given matrix, and linear filtering operations which can be stated as $f(A)= AB$ where $B\in\real^{d\times m}$. 
These decomposable functions can be evaluated in parallel by computing $f(A_1), f(A_2),...,f(A_K)$ as local tasks in a distributed system of worker nodes.

We now introduce the basic building block of computational polarization, which is an analogue to the operation depicted in Figure \ref{fig:polartransform} that creates the virtual channels in \eqref{eqpolarplus} and \eqref{eqpolarminus}. In contrast to the runtime characterizations described earlier, the described operation creates virtual worker nodes with unequal runtime distributions.

\subsection{One Step Computational Polarization}
\label{sec:onesteppolarization}
\begin{figure}[t!]
\begin{center}
\begin{minipage}{0.5\textwidth}
\centering
\begin{tikzpicture}[scale=0.6]
\begin{axis}[xmin=-0.1,xmax=1.1,ymin=0,ymax=1.1, samples=50]
  \addplot[blue, ultra thick, domain=0:1] (x,1);
  \draw[dashed, ultra thick] (0,0) -- (0,1);
  \draw[dashed, ultra thick] (1,0) -- (1,1);
\end{axis}
\end{tikzpicture}
\end{minipage}
\end{center}
\begin{minipage}{0.5\textwidth}
\centering
\begin{tikzpicture}[scale=0.6]
\begin{axis}[xmin=-0.1,xmax=1.1,ymin=0,ymax=1.1, samples=50]
  \addplot[blue, ultra thick, domain=0:1] {x};
  \draw[dashed, ultra thick] (1,0) -- (1,1);
\end{axis}
\end{tikzpicture}
\end{minipage}
\begin{minipage}{0.5\textwidth}
\centering
\begin{tikzpicture}[scale=0.6]
\begin{axis}[xmin=-0.1,xmax=1.1,ymin=0,ymax=1.1, samples=50]
  \addplot[blue, ultra thick, domain=0:1] {1-x};
  \draw[dashed, ultra thick] (0,0) -- (0,1);
\end{axis}
\end{tikzpicture}
\end{minipage}
\begin{minipage}{0.24\textwidth}
\centering
\begin{tikzpicture}[scale=0.5]
\begin{axis}[xmin=-0.1,xmax=1.1,ymin=0,ymax=4.1, samples=50]
  \addplot[blue, ultra thick, domain=0:1] {4*x^3};
  \draw[dashed, ultra thick] (1,0) -- (1,4);
\end{axis}
\end{tikzpicture}
\end{minipage}
\begin{minipage}{0.24\textwidth}
\centering
\begin{tikzpicture}[scale=0.5]
\begin{axis}[xmin=-0.1,xmax=1.1,ymin=0,ymax=4.1, samples=50]
  \addplot[blue, ultra thick, domain=0:1] {4*x-4*x^3};
\end{axis}

\end{tikzpicture}
\end{minipage}
\begin{minipage}{0.24\textwidth}
\centering
\begin{tikzpicture}[scale=0.5]
\begin{axis}[xmin=-0.1,xmax=1.1,ymin=0,ymax=4.1, samples=50]
  \addplot[blue, ultra thick, domain=0:1] {4*x*(2-x)*(1-x)};
\end{axis}
\end{tikzpicture}
\end{minipage}
\begin{minipage}{0.24\textwidth}
\centering
\begin{tikzpicture}[scale=0.5]
\begin{axis}[xmin=-0.1,xmax=1.1,ymin=0,ymax=4.1, samples=50]
  \addplot[blue, ultra thick, domain=0:1] {4*(1-x)^3};
  \draw[dashed, ultra thick] (0,0) -- (0,4);
\end{axis}
\end{tikzpicture}
\end{minipage}
\caption{Tree representation of the polarization of computation times generated from the Uniform$[0,1]$ base distribution. The base distribution is at the root node (top). The average function values of the left and right children is equal to the parent node's function value for each point in the domain. \label{fig:treeuniform}}
\end{figure}
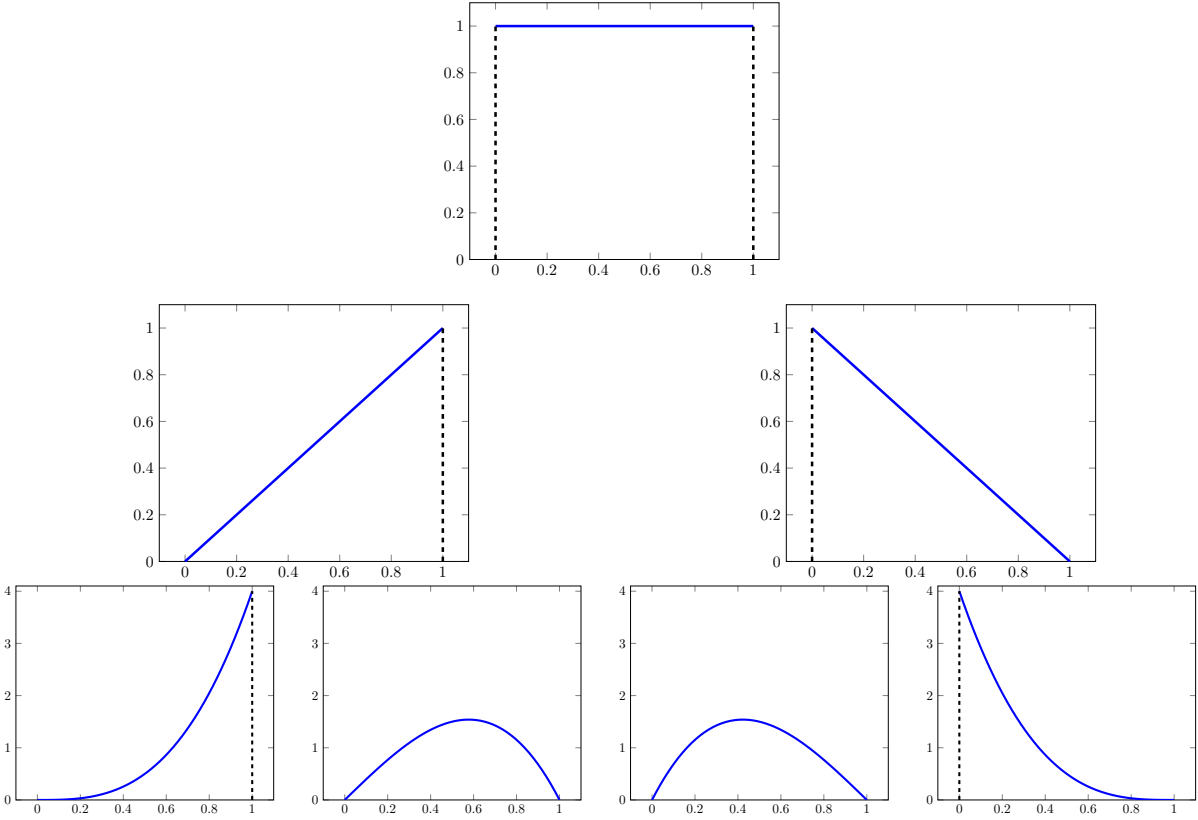

In the basic building block, we will synthesize two virtual workers from two physical workers. This operation is depicted in Figure \ref{fig:buildingblock}. Suppose that the data matrix is partitioned into two as $A = \left[ \begin{array}{c} A_1 \\ A_2 \end{array} \right]$.

Suppose that the task is to compute the values $f( A_{1})$ and $f( A_{2})$. Let us consider the order two Hadamard matrix 
\begin{align*}
H_2 :=\left[ \begin{array}{cc} 1 & 1 \\ 1 & -1 \end{array}\right]\,,
\end{align*}
and form the transformed matrices $A_1+A_2$ and $A_1-A_2$ and compute the function on these matrices
\begin{align*}
Y^{(1)} = f( A_{1} + A_{2}) \mbox{\quad and \quad} Y^{(2)} = f( A_{1} - A_{2})\,.
\end{align*}
Now, we illustrate decoding in the case of linear functions, which satisfy
\begin{align*}
f(A_1+A_2) &= f(A_1)+f(A_2) \\
f(A_1-A_2) &= f(A_1)-f(A_2)\,. 
\end{align*}
To finish the overall computation, we proceed recovering the values $f(A_1)$ and $f(A_2)$ successively as follows:
\begin{align}%
&\mbox{(i)\, {\bf slow worker:}}\quad \, \mbox{computes $f( A_{1})$ using $Y^{(1)}$ and $Y^{(2)}$} \hspace{6.1cm}\quad \label{eq:decode1}\\
&\mbox{(ii) {\bf fast worker:}}\quad\,\, \mbox{computes $f(A_2)$ using $Y^{(1)}$ and $Y^{(2)}$ assuming $f(A_{1})$ has already been computed}\label{eq:decode2}
\end{align}
In particular, we may specify the completion times of these virtual workers as follows 
\begin{align*}%
&\mbox{(i)\,\,  {\bf slow worker  }}\mbox{completes its computation when $Y^{(1)}$ and $Y^{(2)}$ are both available, i.e., at time $\max( T^{(1)}, T^{(2)})$.}& \\
&\mbox{(ii) {\bf fast worker    }}\mbox{completes its computation when either $Y^{(1)}$ or $Y^{(2)}$ is available, i.e., at time $\min( T^{(1)}, T^{(2)})$.}&
\end{align*}
Therefore, the slow worker finishes the reconstruction of the computation $f( A_{1})$ and $f( A_{2})$ when both quantities are available, therefore in time $\max( T^{(1)}, T^{(2)})$. In contrast, the fast worker finishes the reconstruction of computation in time $\min( T^{(1)}, T^{(2)} )$ due to the availability of $f(A_{1})$ beforehand.

At this point, observe that we may \emph{freeze} the slow worker by setting the input matrix $A_1$ to a matrix or zeros, or any fixed matrix and use the input matrix $A_2$ for data. This step avoids the slow runtime $\max( T^{(1)}, T^{(2)})$ via redundancy. It can be verified that this scheme is identical to a simple rate $1/2$ repetition scheme in this special case.

 The above one step polarization operation creates two virtual workers whose runtimes are $T^-$ and $T^+$ are given by
\begin{align*}
T^{+}&= \max( T^{(1)}, T^{(2)} ) \\
T^{-}&= \min( T^{(1)}, T^{(2)} )\,.
\end{align*}

Note that this is analogous to the virtual channels created by the two-step polarization transform in \eqref{eqpolarplus} and \eqref{eqpolarminus}.
It can be easily verified that the average runtime is preserved under this transformation as a result of the $\min$ and $\max$ operations, which establishes a martingale property.
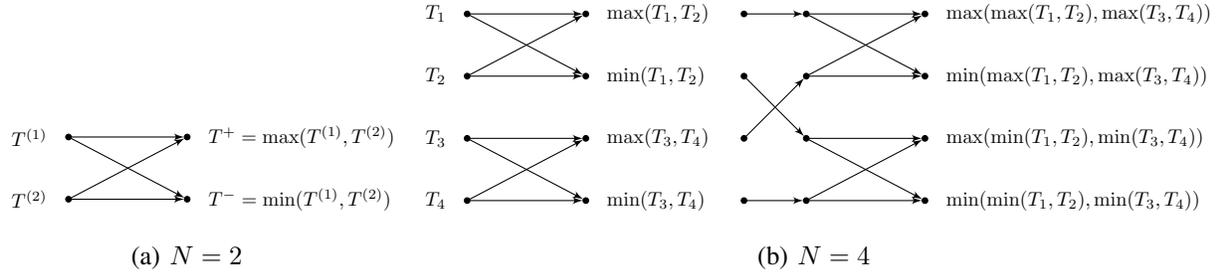
\begin{figure}[!t]
\begin{minipage}[b]{0.30\linewidth}
\centering
\tikzstyle{POINT} = [draw, circle,fill,scale=0.3]
\begin{tikzpicture}[node distance=\nodedistance,auto,>=latex', scale = 0.75, transform shape]
    \tikzstyle{line}=[draw, -latex']
    \node [POINT] (P1) {};
    \node [POINT, below=1 of P1] (P2) {};
    \node [POINT,  right=1 and 2 of P1] (P21) {};
    \node [POINT,  right=1 and 2 of P2] (P22) {};
    \node [left =0.2 of P1] {$T^{(1)}$};
    \node [left =0.2 of P2] {$T^{(2)}$};
    \node [right =0.2 of P21] {$T^{+}=\max(T^{(1)},T^{(2)})$};
    \node [right =0.2 of P22] {$T^{-}=\min(T^{(1)},T^{(2)})$};
    \path[line] (P1) edge (P21)
              (P1) edge (P22)
              (P2) edge (P21)
              (P2) edge (P22);
\end{tikzpicture}
\centerline{(a) $N=2$ }\medskip
\end{minipage}
\hfill
\begin{minipage}[b]{0.69\linewidth}
 \centering
\tikzstyle{POINT} = [draw, circle,fill,scale=0.3]
\begin{tikzpicture}[node distance=\nodedistance,auto,>=latex', scale = 0.75, transform shape]
    \tikzstyle{line}=[draw, -latex']
    \node [POINT] (P1) {};
    \node [POINT, below=1 of P1] (P2) {};
    \node [POINT, below=1 of P2] (P3) {};
    \node [POINT, below=1 of P3] (P4) {};
    \node [POINT,  right=1 and 2 of P1] (P21) {};
    \node [POINT,  right=1 and 2 of P2] (P22) {};
    \node [POINT,  right=1 and 2 of P3] (P23) {};
    \node [POINT,  right=1 and 2 of P4] (P24) {};    
    \node [left =0.2 of P1] {$T_1$};
    \node [left =0.2 of P2] {$T_2$};
    \node [left =0.2 of P3] {$T_3$};
    \node [left =0.2 of P4] {$T_4$};
    \node [right =0.2 of P21] {$\max(T_1,T_2)$};
    \node [right =0.2 of P22] {$\min(T_1,T_2)$};
    \node [right =0.2 of P23] {$\max(T_3,T_4)$};
    \node [right =0.2 of P24] {$\min(T_3,T_4)$};

    \node [POINT, right=2.7 of P21] (P31) {};
    \node [POINT, right=2.7 of P22] (P32) {};
    \node [POINT, right=2.7 of P23] (P33) {};
    \node [POINT, right=2.7 of P24] (P34) {};    
    \node [POINT, right=1 of P31] (P41) {};
    \node [POINT, right=1 of P32] (P42) {};
    \node [POINT, right=1 of P33] (P43) {};
    \node [POINT, right=1 of P34] (P44) {}; 
    \node [POINT, right=2 of P41] (P51) {};
    \node [POINT, right=2 of P42] (P52) {};
    \node [POINT, right=2 of P43] (P53) {};
    \node [POINT, right=2 of P44] (P54) {};         
    \node [right =0.2 of P51] {$\max(\max(T_1,T_2),\max(T_3,T_4))$};
    \node [right =0.2 of P52] {$\min(\max(T_1,T_2),\max(T_3,T_4))$};
    \node [right =0.2 of P53] {$\max(\min(T_1,T_2),\min(T_3,T_4))$};
    \node [right =0.2 of P54] {$\min(\min(T_1,T_2),\min(T_3,T_4))$};
    \path[line] (P1) edge (P21)
                (P1) edge (P22)
                (P2) edge (P21)
                (P2) edge (P22)
                (P3) edge (P23)
                (P3) edge (P24)
                (P4) edge (P24)
                (P4) edge (P23)
                (P31) edge (P41)
                (P32) edge (P43)
                (P33) edge (P42)
                (P34) edge (P44)
                (P41) edge (P51)
                (P42) edge (P51)
                (P41) edge (P52)
                (P42) edge (P52)
                (P43) edge (P53)
                (P44) edge (P53)
                (P44) edge (P54)
                (P43) edge (P54);
\end{tikzpicture}
\centerline{(b) $N=4$ }\medskip
\end{minipage}
\caption{Evolution of the random runtimes after (a) the two-point, (b) the four-point $\min$-$\max$ transform. \label{fig:minmax4}}
\end{figure}
Specifically, the expected values of respective runtimes obey
\begin{align*}
\frac{1}{2}  \Exs \big[ T^{-} ] + \frac{1}{2} \Exs \big[ T^{+} \big] = \frac{1}{2} \Exs \big[ T^{(1)} + T^{(2)} \big] = \Exs [T^{(1)}]\,,
\end{align*}
where in the final equality we assumed that $T^{(1)}$ and $T^{(2)}$ are identically distributed. This preservation property is reminiscent of the capacity preservation observed in Polar Codes in equation \eqref{eq:polarconservation}:
the average computation time is preserved. Furthermore, other relevant quantities are preserved under this transformation, such as the expected value of the product, i.e., 
\begin{align*}
\Exs \big[ T^- T^+\big] = \Exs \big[T_1 T_2\big] = \Exs\big[ T_1\big]^2\,,
\end{align*}
since the relation $\min(T_1, T_2) \max(T_1, T_2) = T_1T_2$ always holds.
Moreover, we have 
\begin{align}
\min(T_1,T_2) = T^{-}\le T^{+} = \max(T_1,T_2)\label{eq:minmaxinequality}
\end{align}
with probability one, which immediately follows from the definition of $T^{-}$ and $T^+$ through the $\min$ and $\max$ operations respectively. The inequality in \eqref{eq:minmaxinequality} holds with equality if and only if $T_1=T_2$. This property is analogous to the redistribution of capacity shown in \eqref{eq:polarredist}. Noting that $T_1$ and $T_2$ are i.i.d. realizations of the same random variable, it is intuitively natural to expect a redistribution of the runtime, unless the distribution is deterministic. We illustrate the two-point $\min$-$\max$ transformation via the diagram in Figure \ref{fig:minmax4} (a), which is extended in Figure \ref{fig:minmax4} (b) to a four-point transformation.

The distributions of the virtual worker runtimes can be characterized in terms of their CDFs as follows. First, consider the distribution of $T^{+}=\max( T^{(1)}, T^{(2)} )$ given by
\begin{align}
\prob{\max( T^{(1)}, T^{(2)} ) \le t} &= \prob{ T^{(1)} \le t, T^{(2)}  \le t} \nonumber\\
& = \prob{ T^{(1)}\le t } \prob{ T^{(2)}\le t} \nonumber \\
& = F(t)^2 \label{eq:max} %
\end{align}
Next, one can carry out the calculation for the distribution of $T^{-}=\min( T^{(1)}, T^{(2)} )$ as follows
\begin{align}
\prob{\min( T^{(1)}, T^{(2)} ) \le t} &= 1-\prob{\min( T^{(1)}, T^{(2)} ) > t} \nonumber\\
& = 1-\prob{ T^{(1)} > t,\, T^{(2)} > t} \nonumber\\
& = 1-\prob{ T^{(1)} > t}\prob{ T^{(2)} > t} \nonumber\\
& = 1-(1-F(t))^2\,. \label{eq:min}%
\end{align}
Therefore the CDFs follow the functional tree process shown in Figure \ref{fig:CDFtreetwo}, which map a base CDF $F(t)$ into $F^{+}(t)=F(t)^2$ and $F^{-}(t) = 1-(1-F(t))^2$.
\begin{figure}
\begin{minipage}[h]{.5\textwidth}
\tikzstyle{POINT} = [draw, circle,fill,scale=0.3]
\begin{tikzpicture}[node distance=\nodedistance,auto,>=latex', scale = 0.75, transform shape]
    \tikzstyle{line}=[draw, -latex']
    \node [POINT] (P1) {};
    \node [POINT, above right=1 and 2 of P1] (P21) {};
    \node [POINT, below right=1 and 2 of P1] (P22) {};
    \node [left =0.2 of P1] {$F(t)$};
    \node [right =0.2 of P21] {$F^{+}=F(t)^2$};
    \node [right =0.2 of P22] {$F^{-}=1-(1-F(t))^2$};
    \path[line] (P1) edge (P21)
    			(P1) edge (P22);
\end{tikzpicture}
\vspace{1cm}
\caption{ Basic split operation describing the \\functional process \label{fig:CDFtreetwo}}
\end{minipage}
\begin{minipage}[h]{0.5\textwidth}
\tikzstyle{POINT} = [draw, circle,fill,scale=0.3]
\begin{tikzpicture}[node distance=\nodedistance,auto,>=latex', scale = 0.75, transform shape]
    \tikzstyle{line}=[draw, -latex']
    \node [POINT] (P1) {};
    \node [POINT, above right=1 and 2 of P1] (P21) {};
    \node [POINT, below right=1 and 2 of P1] (P22) {};
    \node [POINT, above right=0.6 and 2 of P21] (P31) {};
    \node [POINT, below right=0.6 and 2 of P21] (P32) {};
    \node [POINT, above right=0.6 and 2 of P22] (P33) {};
    \node [POINT, below right=0.6 and 2 of P22] (P34) {};
    \node [left =0.2 of P1] {$F(t)$};
    \node [above =0.2 of P21] {$F^{+}=F(t)^2$}; 
    \node [below = 0.32 of P22] {$F^{-}=2F(t) - F(t)^2$};
    \node [right=0.2 of P31] {$F^{++}(t)=(F(t)^2)^2$};
    \node [right=0.2 of P32] {$F^{+-}(t)=2(F(t)^2) - (F(t)^2)^2$};
    \node [right=0.2 of P33] {$F^{-+}(t)=(2F(t) - F(t)^2)^2$};
    \node [right=0.2 of P34] {$F^{--}(t)=2(2F(t) - F(t)^2) - (2F(t) - F(t)^2)^2$};
    \path[line] (P1) edge (P21)
    			(P1) edge (P22)
    			(P21) edge (P31)
    			(P21) edge (P32)
    			(P22) edge (P33)
    			(P22) edge (P34);
\end{tikzpicture}
\caption{Recursive tree process describing the evolution of the CDF $F(t)$ \label{fig:CDFtreefour}}
\end{minipage}
\end{figure}
It is worth noting that the virtual worker runtimes after one step of computational polarization obey a total ordering 
\begin{align}
F(t)^2 = F^{+}(t) \le F(t) \le F^{-}(t) = 1-(1-F(t))^2 \,, \mbox{ for all $t\in \real$}\,,
\end{align}
in which equality holds if and only if $F(t)\in\{0,1\},\,\forall t$. Hence, the distribution $F^{-}$ stochastically dominates $F$, which stochastically dominates $F^{+}$ (see, e.g., \cite{hadar1969rules,bawa1975optimal}), and it follows that the runtime $F^{-}$ ($F$) is preferable to $F$ ($F^+$) for every weakly decreasing utility function $\phi(\cdot)$, i.e.,
\begin{align}
\int \phi(t)dF^{-}(t) \ge \int \phi(t)dF(t) \ge \int \phi(t)dF^{+}(t)\,.\label{eqn:stochasticdominance}
\end{align}
We next describe a recursive application of the one step computational polarization that yields a recursive random process involving the $\min$ and $\max$ operators.

\begin{figure}[!t]
\centering
\usetikzlibrary{shapes,arrows,fit,calc,positioning,automata}
\tikzstyle{box}     = [draw, minimum size=4em]
\tikzstyle{init}    = [pin edge={to-,thin,black}]
\tikzstyle{ADD} = [draw,circle]
\tikzstyle{POINT} = [draw, circle,fill,scale=0.3]
\tikzstyle{line}    = [draw, -latex']

\begin{tikzpicture}[node distance=\nodedistance,auto,>=latex', scale = 0.75, transform shape]
    \tikzstyle{line}=[draw, -latex']

    \node [ADD] (ADD1) {\large $+$};
    \node [ADD,below=1.26 of ADD1] (ADD2) {\large $+$};
    \node [ADD,left=0.85 of ADD2] (NEG) {\large $\times$};
    \node [ADD,below=1.26 of ADD2] (ADD3) {\large $+$};
    \node [ADD,below=1.26 of ADD3] (ADD4) {\large $+$};
    \node [ADD,left=0.85 of ADD4] (NEG2) {\large $\times$};    
    \node [below=0.05 of NEG] (MINUS) {\tiny $-1$};
    \node [below=0.05 of NEG2] (MINUS) {\tiny $-1$};
    \node [POINT, right=2 of ADD1] (S12) {};
    \node [POINT, right=2 of ADD2] (S22) {};
    \node [POINT, right=2 of ADD3] (S32) {};
    \node [POINT, right=2 of ADD4] (S42) {};
    \node [POINT, right=2 of S12] (S1) {};
    \node [POINT, right=2 of S22] (S2) {};
    \node [POINT, right=2 of S32] (S3) {};
    \node [POINT, right=2 of S42] (S4) {};
    \node [ADD, right=3 of S1] (ADD12) {\large $+$};
    \node [ADD, right=3 of S2] (ADD22) {\large $+$};
    \node [ADD, right=3 of S3] (ADD32) {\large $+$};
    \node [ADD, right=3 of S4] (ADD42) {\large $+$};
    \node [ADD,left=0.85 of ADD22] (NEG21) {\large $\times$};
    \node [ADD,left=0.85 of ADD42] (NEG22) {\large $\times$};
    \node [below=0.05 of NEG21] (MINUS) {\tiny $-1$};
    \node [below=0.05 of NEG22] (MINUS) {\tiny $-1$};
    \node [POINT,left=2.7 of ADD1] (A1p) {};
    \node [POINT,left=2.7 of ADD2] (A2p) {};
    \node [POINT,left=2.7 of ADD3] (A3p) {};
    \node [POINT,left=2.7 of ADD4] (A4p) {};
    \node [left of=ADD1,node distance=4.2cm,text width=0.5cm,anchor=west,align=center] (A1) {$A_1$};
    \node [left of=ADD2,node distance=4.2cm,text width=0.5cm,anchor=west,align=center] (A2) {$A_2$};
    \node [left of=ADD3,node distance=4.2cm,text width=0.5cm,anchor=west,align=center] (A3) {$A_3$};
    \node [left of=ADD4,node distance=4.2cm,text width=0.5cm,anchor=west,align=center] (A4) {$A_4$};
    \node [box, right=2 of ADD12] (M1) {$M_1$};    
    \node [box, right=2 of ADD22] (M2) {$M_2$};    
    \node [box, right=2 of ADD32] (M3) {$M_3$};    
    \node [box, right=2 of ADD42] (M4) {$M_4$};    
    \node [right=1 of M1] (f1) {$f(A_1 + A_2 + A_3+A_4)$};      
    \node [right=1 of M2] (f2) {$f(A_1 + A_2 - A_3-A_4)$};
    \node [right=1 of M3] (f3) {$f(A_1 - A_2  + A_3-A_4)$};      
    \node [right=1 of M4] (f4) {$f(A_1 - A_2 - A_3+A_4)$};
    \node [above=0.05 of $(S12)!0.4!(S1)$] (OUT1) {$A_1+A_2$};    
    \node [above=0.45 of $(S22)!0.4!(S2)$] (OUT2) {$A_3+A_4$};
    \node [above=0.05 of $(S32)!0.4!(S3)$] (OUT3) {$A_1-A_2$};    
    \node [above=0.05 of $(S42)!0.4!(S4)$] (OUT4) {$A_3-A_4$};
    \draw[thick,dotted]     ($(S1.north west)+(-0.25,0.65)$) rectangle ($(ADD22.south east)+(0.5,-0.65)$);
    \draw[thick,dotted]     ($(S3.north west)+(-0.25,0.65)$) rectangle ($(ADD42.south east)+(0.5,-0.65)$);
    \draw[thick,dotted]     ($(A1p.north west)+(-0.25,0.65)$) rectangle ($(ADD2.south east)+(0.5,-0.65)$);
    \draw[thick,dotted]     ($(A3p.north west)+(-0.25,0.65)$) rectangle ($(ADD4.south east)+(0.5,-0.65)$);
    \path[line] (ADD1) edge (S1)
                (A1p) edge (ADD1)
                (A2p) edge (NEG)
                (NEG) edge (ADD2)
                (A2p) edge (ADD1)
                (A1p) edge (ADD2)
                (NEG2) edge (ADD4)
                (A3p) edge (ADD3)
                (A4p) edge (NEG2)
                (A3p) edge (ADD4)
                (A4p) edge (ADD3)
                (S1) edge (ADD12)
                (S2) edge (NEG21)
                (NEG21) edge (ADD22)
                (NEG22) edge (ADD42)
                (S3) edge (ADD32)
                (S4) edge (NEG22)
                (M1) edge (f1)
                (M2) edge (f2)
                (M3) edge (f3)
                (M4) edge (f4)
                (S1) edge (ADD22)
                (S2) edge (ADD12)
                (S3) edge (ADD42)
                (S4) edge (ADD32)
                (ADD2) edge (S32)
                (ADD3) edge (S22)
                (S22) edge (S2)
                (S32) edge (S3)
                (S32) edge (S3)
                (S42) edge (S4)
                (ADD4) edge (S4)
                (ADD12) edge (M1)
                (ADD22) edge (M2)
                (ADD32) edge (M3)
                (ADD42) edge (M4);
\end{tikzpicture}
\caption{Recursive Computational Polarization: Two-point transformations are recursively applied to data blocks before distributed computation. This specific construction depicts $H_4$ built from two copies of $H_2$. \label{fig:4by4}}
\end{figure}
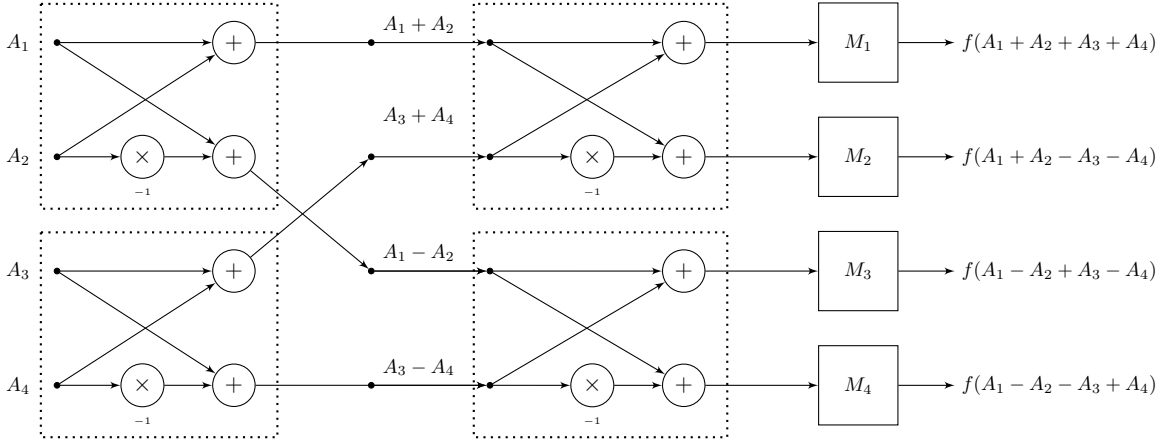

\subsection{Recursive Computational Polarization}
\label{Sec:recursive_comp_polarization}
We now recursively apply the basic construction in Figure \ref{fig:buildingblock} to $4\times 4$ as shown in Figure \ref{fig:4by4}; in this case, this operation corresponds to the discrete Hadamard transform. Using the butterfly connections, this construction can be iterated for any power of two in a similar spirit to recursive Plotkin construction of Polar codes \cite{arikan2009channel}. Given $N$ independent and identically distributed continuous random variables  $T^{(1)},T^{(2)},...,T^{(N)}\sim F$ recursively applying the $2 \times 2$ transformation. This yields a family of distributions which are polarized to either slower or faster compute times as illustrated in Figure \ref{fig:minmax4}. In the case of $N=4$, the runtimes are given by
\begin{align*}
T^{--}&= \max\big(\max( T^{(1)}, T^{(2)} ), \max( T^{(3)}, T^{(4)} )\big) \\
T^{-+}&= \min\big( \max( T^{(1)}, T^{(2)} ), \max( T^{(3)}, T^{(4)} )\big) \\
T^{+-}&= \max\big( \min( T^{(1)}, T^{(2)} ), \min( T^{(3)}, T^{(4)} )\big) \\
T^{++}&= \min\big( \min( T^{(1)}, T^{(2)} ), \min( T^{(3)}, T^{(4)} )\big)
\end{align*}
It is easy to see that the martingale property again holds
\begin{align*}
\frac{1}{4} \big( T^{--} + T^{-+} + T^{+-} + T^{++} \big) = \frac{1}{4} \big( T^{(1)} + T^{(2)} + T^{(3)} + T^{(4)} \big)\,.
\end{align*}

The corresponding CDFs are obtained by a recursive application of the mapping $F(t)\rightarrow \{ F(t)^2, 1-(1-F(t))^2 \}=\{ F^{+}(t),F^{-}(t)\}$, which are given by 
\begin{align*}
T^{++}&=(F(t)^2)^2\\
T^{+-}&=2(F(t)^2)-(F(t)^2)^2\\
T^{-+}&=(2F(t)-F(t)^2)^2\\
T^{--}&=2(2F(t)-F(t)^2)-(2F(t)-F(t)^2)^2\,.
\end{align*}
The closed-form expressions are shown in Figure \ref{fig:CDFtreefour} as a tree process with function valued nodes. In Figure \ref{fig:treeuniform}, these CDFs are plotted.  
Note that the empirical CDFs move away from each other as $N$ increases, and we obtain \textit{better} and \textit{worse} run time distributions, which is what we term the computational polarization process. We show later (see e.g., Figures \ref{fig:polarizedcdfpdfuniform} and \ref{fig:polarizedcdfpdfexp}) that a total ordering in terms of stochastic dominance similar to \eqref{eqn:stochasticdominance} will no longer be possible for larger $N$. However, we will investigate ordering according to various criteria and provide resulting theoretical guarantees.

For the general construction where $N=2^n$, we construct the $N\times N$ matrix $H_{N}=H_2^{\otimes n}$ and apply the $\mathbb{R}^N\rightarrow\mathbb{R}^N$ transformation $P_{N}H_{N}$ to the $N$-vector $\{A_1(i,j),\cdots A_N(i,j)\}$ for all matrix elements $i$ and $j$. Here, $P$ is the bit-reversal permutation matrix defined as follows. Let us index the $i^{th}$ element of a length $N=2^n$ vector by the binary representation $b_1\cdots b_n$ of the integer $i-1$, for which $i=1+\sum_{j=1}^{n} b_j 2^{j-1}$. Then, the bit-reversal permutation matrix $P_N$ maps a vector $x$ to $x^\prime$, for which $x_{b_n\cdots b_1}^\prime=x_{b_1\cdots b_n}$. In simpler terms, the bit-reversal permutation permutes the vector according to the bit reversed ordering. We refer the reader to Section 1B of \cite{arikan2009channel} for further details on the bit-reversal operation. We may use the integer index $i \in \{ 1,2,...,N \}$ and $b_1\cdots b_n$, referred to as the bit index interchangeably when it is clear from the context. Specifically, elements of the set $\{T^{(1)},\cdots T^{(N)}\}$ can be alternatively indexed by $T_{b_1\cdots b_n}$ in this fashion. 

In the general construction, $N$ virtual workers compute $f(A_k)$ when provided the values of $f(A_{1}),\cdots,f(A_{k})$ for $k=1,\cdots,N$. The runtime random variables for the virtual workers can be computed as follows. Given $N$ i.i.d. random variables $T^{(1)},T^{(2)},...,T^{(N)}\sim F$ for the runtimes, the recursive application of the $2 \times 2$ $\min$-$\max$ transformation. Let us represent the sample paths using a binary valued sequence $\{b_j\}_{j=1}^{\infty}$ with elements $b_j\in \{0,1\}$. Define
\begin{align*}
T_{b_1b_2...b_nb_{n+1}} = 
\begin{cases} \max\big(T_{b_1b_2...b_n}, T^\prime_{b_1b_2...b_n}\big) & \mbox{if } b_n=1\\
\min\big(T_{b_1b_2...b_n}, T^\prime_{b_1b_2...b_n}\big) & \mbox{if } b_n=0\,.
 \end{cases}
\end{align*}

\begin{figure}[!t]
\begin{minipage}[b]{0.23\linewidth}
  \centering
  \centerline{\includegraphics[width=4cm]{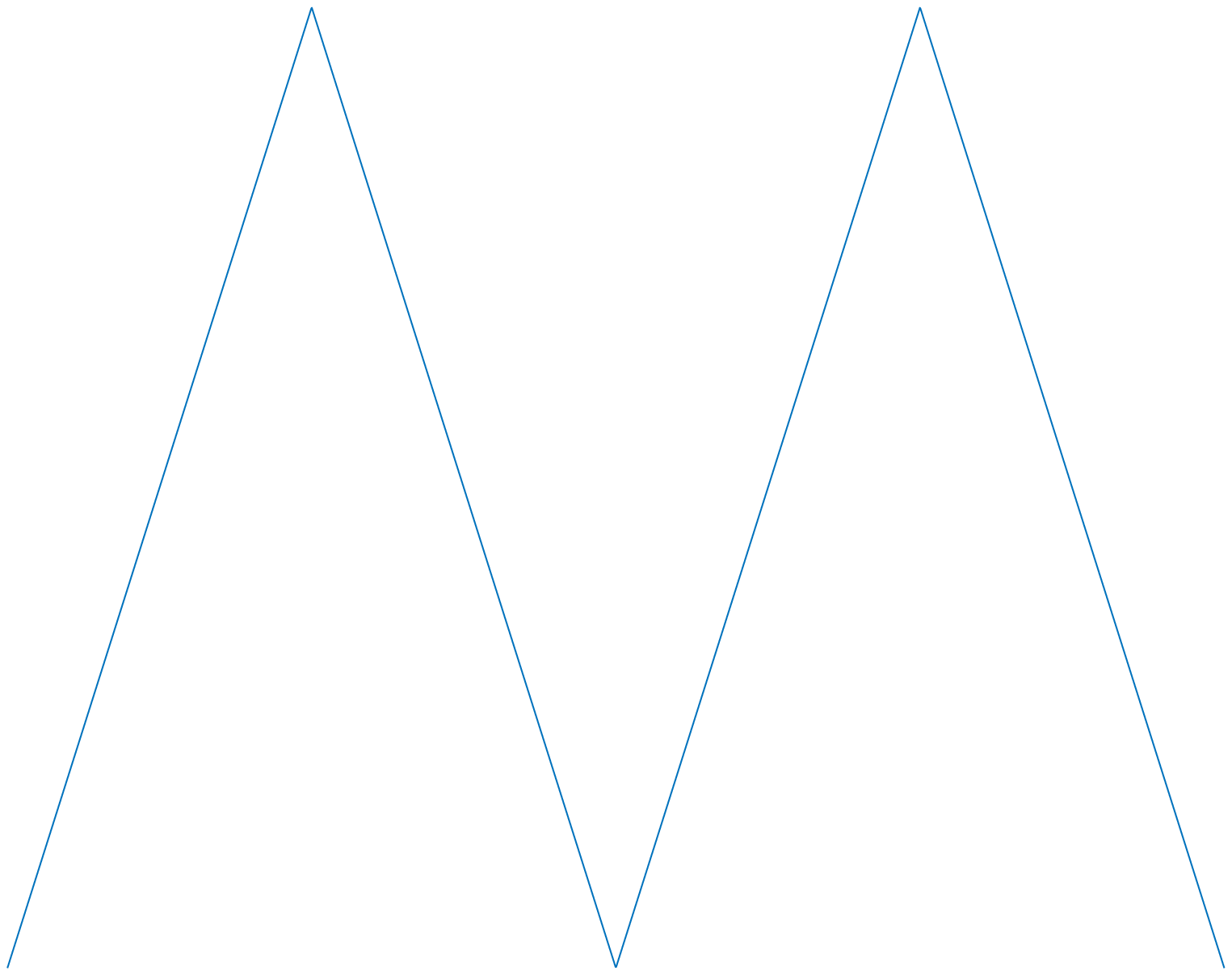}}
  \centerline{(a) $N=1$ (base distribution)}\medskip
\end{minipage}
\hfill
\begin{minipage}[b]{0.23\linewidth}
  \centering
  \centerline{\includegraphics[width=4cm]{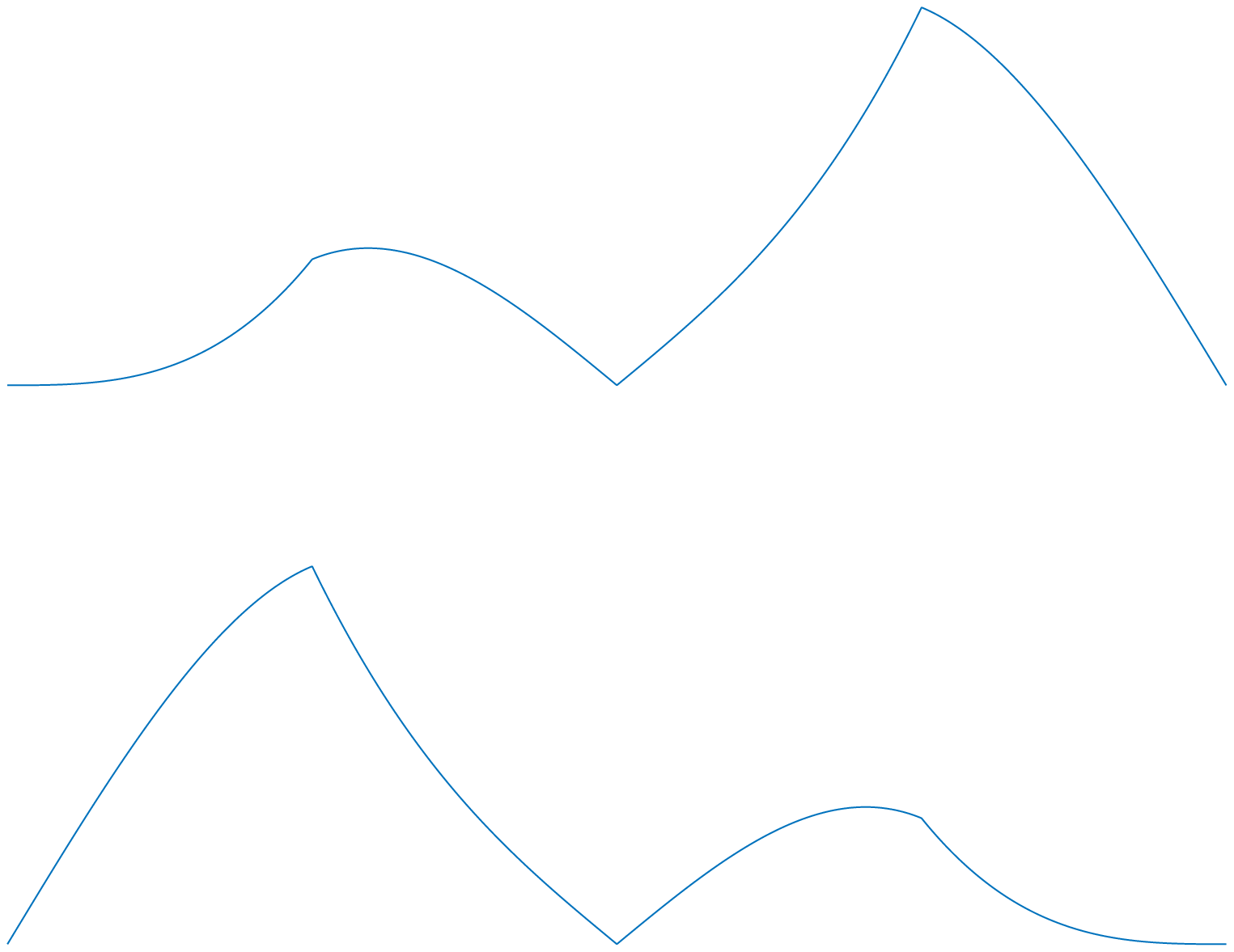}}
  \centerline{(b) $N=2$}\medskip
\end{minipage}
\hfill
\begin{minipage}[b]{0.23\linewidth}
  \centering
  \centerline{\includegraphics[width=4cm]{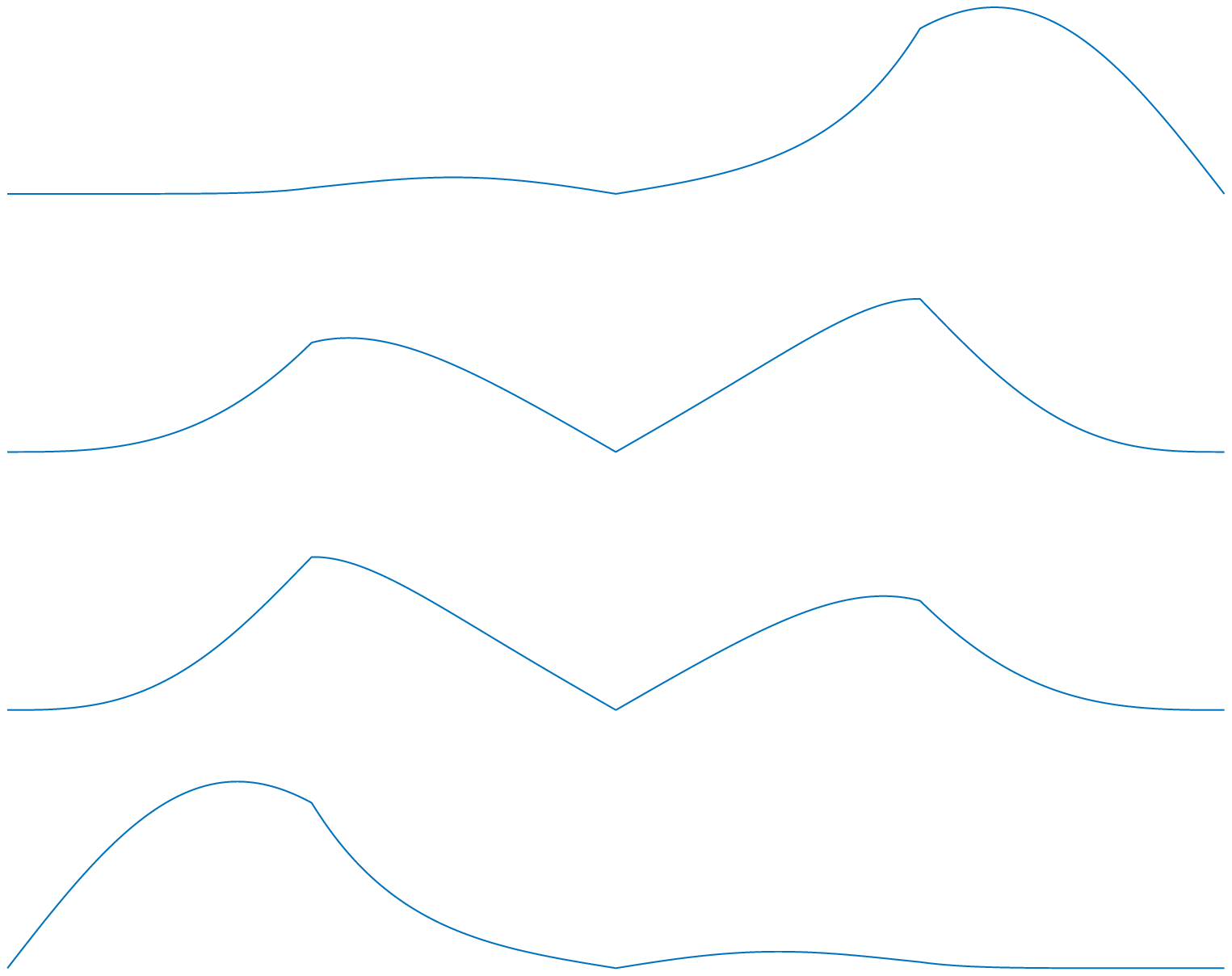}}
  \centerline{(b) $N=4$}\medskip
\end{minipage}
\hfill
\begin{minipage}[b]{0.23\linewidth}
  \centering
  \centerline{\includegraphics[width=4cm]{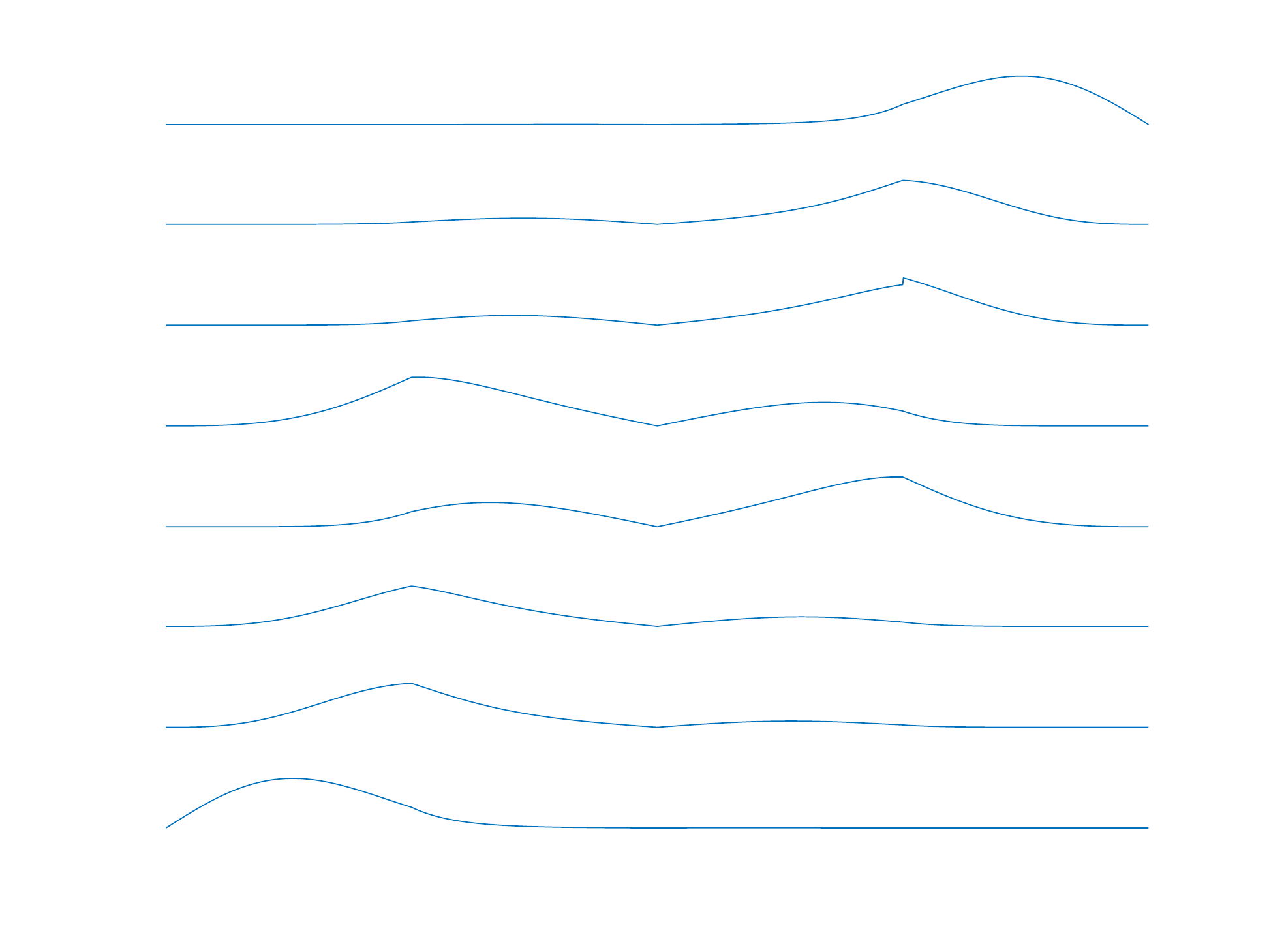}}
  \centerline{(b) $N=8$}\medskip
\end{minipage}
\caption{Computational polarization applied to the bimodal triangular distribution (a), generates the piecewise polynomial distributions in (b), (c), (d) as elements of the function valued martingale process \label{fig:triangle}}
\end{figure}

Consequently, applying the same argument as in equations \eqref{eq:max} and \eqref{eq:min}, the corresponding runtime distributions characterized by the CDFs
\begin{align*}
	F_{b_1b_2...b_n(t)} &= \prob{T_{b_1b_2...b_n} \le t}\,,
\end{align*}
follow the recursion
\begin{align}
F_{b_1b_2...b_nb_{n+1}}(t) = 
\begin{cases} 
F_{b_1b_2...b_n}(t)^2,\,\forall t & \mbox{if } b_{n+1}=1 \\ 
1-(1-F_{b_1b_2...b_n}(t))^2,\,\forall t & \mbox{if } b_{n+1}=0\,.  
\end{cases}
\label{eq:cdfgeneralrecursion}
\end{align}
We introduce the alternative integer notation $F_{n,1},F_{n,2}..., F_{n,N}$ to index the CDFs $\{ F_{b_1b_2...b_n} \}_{b_1=0,...,b_n=0}^{b_1=1,...,b_n=1}$ where $b_1,...b_n$ is the bit index and $N=2^n$, as described earlier in this section.

\begin{remark}[\textbf{Freezing workers}] As in Polar Codes, one can \emph{freeze} undesirable virtual workers: we can set $A_i=0$ for $i\in \mathcal{F}$ for some set $\mathcal{F}$ to eliminate the corresponding runtimes. As a result, we are able to achieve various runtimes by choosing a suitable freezing set $\mathcal{F}$ that controls the degree of redundancy. Alternatively, one can compute $f(A_i)$ for $i\in\mathcal{F}$ at a reliable worker node, instead of freezing the virtual workers for the same effect.
\end{remark}
Next, we illustrate the shape of the virtual runtime distributions for a variety of base CDFs.
\begin{prop}
Suppose that the base distribution specified by the CDF $F(t)$ is a piecewise polynomial function with degree at most $q$. Then the polarized CDFs after a size $N=2^n$, or equivalently depth $n$ polarizations are piecewise polynomial functions with degree at most $q N$. Furthermore, the polarized distributions admit probability densities, which are piecewise polynomial functions of degree at most $qN-1$.
\end{prop}
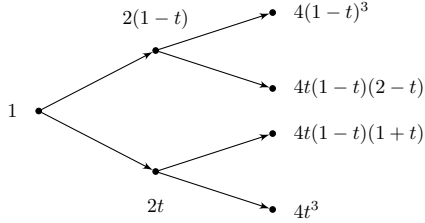
\begin{figure}
\centering
\tikzstyle{POINT} = [draw, circle,fill,scale=0.3]
\begin{tikzpicture}[node distance=\nodedistance,auto,>=latex', scale = 0.75, transform shape]
    \tikzstyle{line}=[draw, -latex']
    \node [POINT] (P1) {};
    \node [POINT, above right=1 and 2 of P1] (P21) {};
    \node [POINT, below right=1 and 2 of P1] (P22) {};
    \node [POINT, above right=0.6 and 2 of P21] (P31) {};
    \node [POINT, below right=0.6 and 2 of P21] (P32) {};
    \node [POINT, above right=0.6 and 2 of P22] (P33) {};
    \node [POINT, below right=0.6 and 2 of P22] (P34) {};
    \node [left =0.2 of P1] {$1$};
    \node [above =0.2 of P21] {$2(1-t)$};
    \node [below = 0.32 of P22] {$2t$};
    \node [right=0.2 of P31] {$4(1-t)^3$};
    \node [right=0.2 of P32] {$4t(1-t)(2-t)$};
    \node [right=0.2 of P33] {$4t(1-t)(1+t)$};
    \node [right=0.2 of P34] {$4t^3$};
    \path[line] (P1) edge (P21)
    			(P1) edge (P22)
    			(P21) edge (P31)
    			(P21) edge (P32)
    			(P22) edge (P33)
    			(P22) edge (P34);
\end{tikzpicture}
\caption{Tree process describing the polarization of the uniform distribution. Sample paths of the functional martingale process that corresponds to the polarized computation times are depicted as a tree with function valued nodes. \label{fig:tree_uniform}}
\end{figure}

The above result immediately follows from the recursive definition given in \eqref{eq:cdfgeneralrecursion} by inspection. We provide a graphical depiction of this result in Figure \ref{fig:triangle} for a bimodal triangle shaped piecewise linear distribution. Next, we provide an analytical characterization for the special case of the uniform runtime distributions.
\begin{example}[Polarization of the Uniform Distribution]
Suppose that $T^{(1)},...,T^{(n)}$ are distributed according to the uniform distribution on the unit interval $[0,1]$. Then the marginal distribution of the order statistics $T_{(1)},...,T_{(n)}$ sorted in increasing order is given by the Beta distribution family, i.e., $T_{(k)}\sim \mathrm{Beta}(k,n+1-k)$ \cite{gentle2009computational}. The probability density function of the sorted runtime variable $T_{(k)}$ is equal to
\begin{align*}
p_{T_{(k)}}(t) = \frac{n!}{(k-1)!(n-k)!}t^{k-1}(1-t)^{n-k}\,.
\end{align*}
For $n=4$, the marginal distributions are given by
\begin{align*}
p_{T_{(1)}}(t) &= 4(1-t)^3\\
p_{T_{(2)}}(t) &= 12t(1-t)^2\\
p_{T_{(3)}}(t) &= 12t^2(1-t)\\
p_{T_{(4)}}(t) &= 4t^3\,.\\
\end{align*}
In contrast, the runtime distributions for the computational polarization scheme are given by the tree in Figure \ref{fig:tree_uniform}. Interestingly, the distributions $p_{T_{(1)}}(t)$ and $p_{T_{(4)}}(t)$ coincide with the Beta distribution, since they are extremal order statistics, while the other two distributions $4t(1-t)(2-t)$ and $4t(1-t)(1+t)$ are different. Both families of probability densities are depicted in Figure \ref{fig:order_vs_polarized}. This mismatch arises as a result of the martingale nature of the polarized computation times. It can be observed that the computational polarization scheme approximates the marginal order statistics.

\begin{figure}[h]
\begin{minipage}{0.24\textwidth}
\centering
\begin{tikzpicture}[scale=0.5]
\begin{axis}[xmin=-0.1,xmax=1.1,ymin=0,ymax=4.1, samples=50]
  \addplot[blue, ultra thick, domain=0:1] {4*x^3};
  \addplot[blue, ultra thick, smooth,domain=0:1] coordinates {(1,0)(1,4)};
  \addplot[red, mark=+, mark options={scale=2}, mark repeat=5, dashed, ultra thick, domain=0:1] {4*x^3};
  \addplot[red, mark=+, mark options={scale=2}, mark repeat=5, dashed,ultra thick, smooth,domain=0:1] coordinates {(1,0)(1,4)};
\end{axis}
\end{tikzpicture}
\end{minipage}
\begin{minipage}{0.24\textwidth}
\centering
\begin{tikzpicture}[scale=0.5]
\begin{axis}[xmin=-0.1,xmax=1.1,ymin=0,ymax=4.1, samples=50]
  \addplot[blue, ultra thick, domain=0:1] {4*x-4*x^3};
  \addplot[red, mark=+, mark options={scale=2}, mark repeat=5, dashed, ultra thick, domain=0:1] {12*x^2*(1-x)};
\end{axis}
\end{tikzpicture}
\end{minipage}
\begin{minipage}{0.24\textwidth}
\centering
\begin{tikzpicture}[scale=0.5]
\begin{axis}[xmin=-0.1,xmax=1.1,ymin=0,ymax=4.1, samples=50]
  \addplot[blue, ultra thick, domain=0:1] {4*x*(2-x)*(1-x)};
  \addplot[red, mark=+, mark options={scale=2}, mark repeat=5, dashed, ultra thick, domain=0:1] {12*x*(1-x)^2};
\end{axis}
\end{tikzpicture}
\end{minipage}
\begin{minipage}{0.24\textwidth}
\centering
\begin{tikzpicture}[scale=0.5]
\begin{axis}[xmin=-0.1,xmax=1.1,ymin=0,ymax=4.1, samples=50]
  \addplot[blue, ultra thick, domain=0:1] {4*(1-x)^3};
  \addplot[blue, ultra thick, smooth,domain=0:1] coordinates {(0,0)(0,4)};
  \addplot[red, mark=+, mark options={scale=2}, mark repeat=5, dashed, ultra thick, domain=0:1] {4*(1-x)^3};
  \addplot[red, mark=+, mark options={scale=2}, mark repeat=5, dashed, ultra thick, smooth,domain=0:1] coordinates {(0,0)(0,4)};
\end{axis}
\end{tikzpicture}
\end{minipage}
\caption{Order statistics of the uniform random variables given by the Beta distribution (red $+$ markers) vs polarized uniform computation times (solid blue) \label{fig:order_vs_polarized}}
\end{figure}
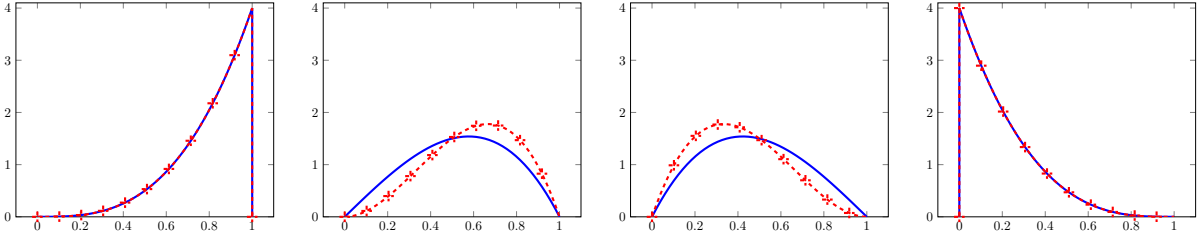
\end{example}
The preceding example also demonstrates an interesting computational trade-off in the decoding process. Decoding in the computational polarization scheme is of very low cost due to the successive cancellation operation which can be carried out over real numbers without multiplication, and using only addition and subtraction; which is simply a sign flip operation. On the other hand, although MDS codes may provide better performance, decoding MDS codes is a significantly more complex process which generally involves solving linear systems and multiple steps. These linear systems typically require cubic complexity in the dimension of the unknown variables. Moreover, certain constructions and decoding methods for MDS codes are based on finite field arithmetic and pose an additional obstacle. For instance, Fermat Number Transform (FNT) based MDS codes have sub-cubic complexity, however they suffer from additional computational overhead (see Section \ref{subsec:runtimecomp} for a comparison).

\subsection{Decoding}
 The decoding of the overall computation can be done in a similar spirit to the successive cancellation decoder of traditional Polar Codes over finite fields, where a major difference is that it operates over real-valued data. An implementation of the successive decoding strategy for real-valued polar codes was described in the earlier work \cite{bartan2019straggler}. Specifically, the principle behind the decoder parallels the successive cancellation strategy described in \cite{arikan2009channel} and is as follows. For decoding every block $f(A_i)$, it is the case that either $A_i$ is frozen, i.e., $f(A_i)$ is known beforehand, or $f(A_i)$ can be found as in $\eqref{eq:decode1}$ and $\eqref{eq:decode2}$ by addition and subtraction over the reals as a result of the recursive construction. Therefore, when a suitable freezing set $\mathcal{F}$ is determined, one can finish the overall computation in time 
\begin{align*}
\max_{b_1b_2\cdots b_n \neq \mathcal{F}} T_{b_1b_2\cdots b_n}\,,
\end{align*}
where we used the alternative bit index $b_1b_2\cdots b_n$ to index the virtual workers.

It is easy to see that the decoding can be performed in $O(N \log N)$ time in serial computation as in successive cancellation decoding of Polar Codes. With a parallel implementation using $N$ workers, the decoding complexity is $O(\log(N))$ per worker yielding a $O(\log(N))$ depth algorithm..

\subsection{Larger Construction Sizes and Asymptotics of Computational Polarization}
In this section, we explore larger construction sizes to gain intuition in the asymptotics of computational polarization. First, we note that it is possible to employ the recursive formula given in \eqref{eq:cdfgeneralrecursion} to analytically obtain the polarized distributions at any construction size, of a power of two. As an illustrative example, we show the polarized probability densities for $N=8,16,32,64$ in Figures \ref{fig:wave8and16} and \ref{fig:wave32and64}.
\begin{figure}[!t]
\begin{minipage}[b]{0.49\linewidth}
  \centering
  \centerline{\includegraphics[width=8cm]{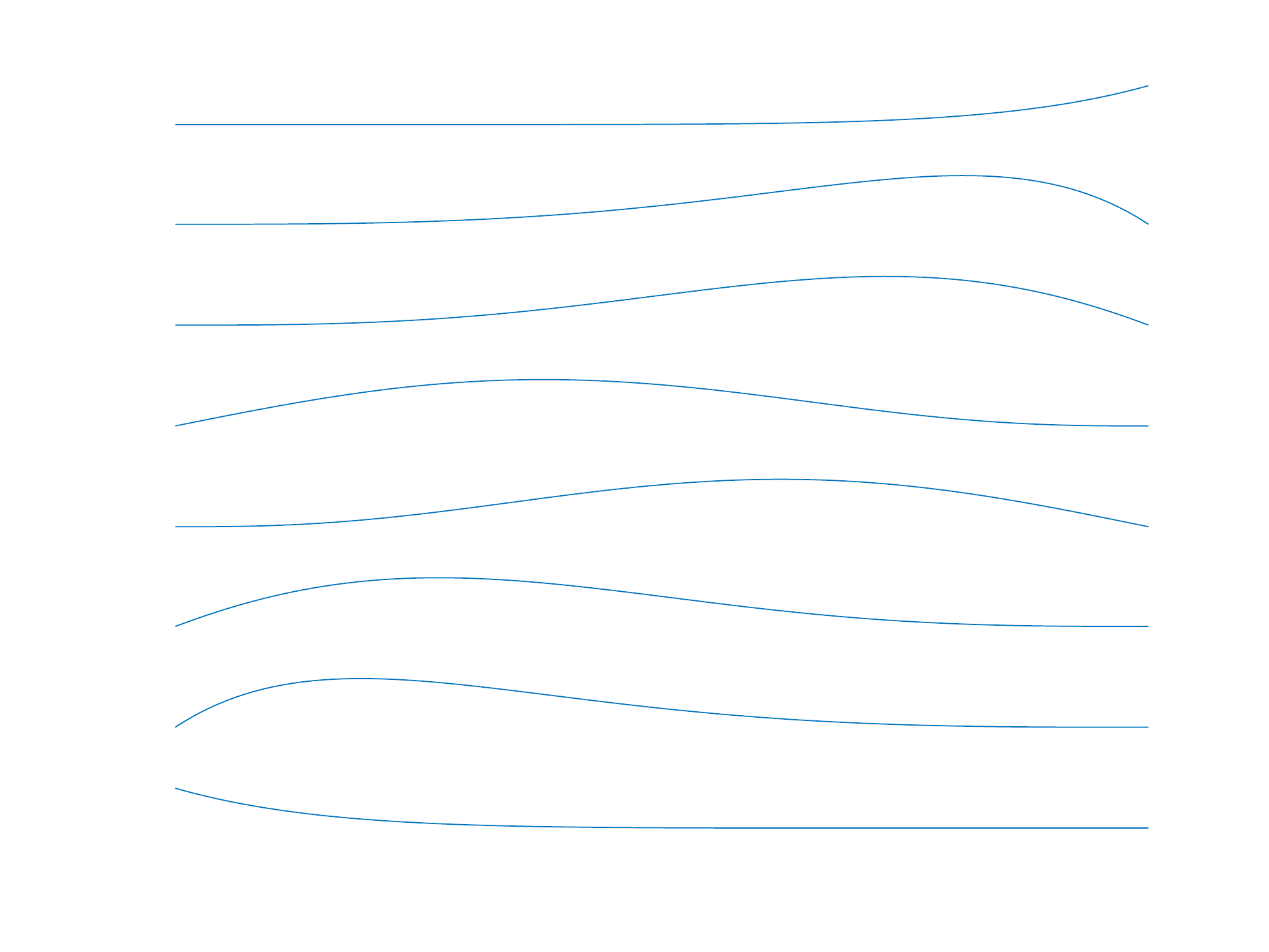}}
  \centerline{(a) $N=8$}\medskip
\end{minipage}
\hfill
\begin{minipage}[b]{0.49\linewidth}
  \centering
  \centerline{\includegraphics[width=8cm]{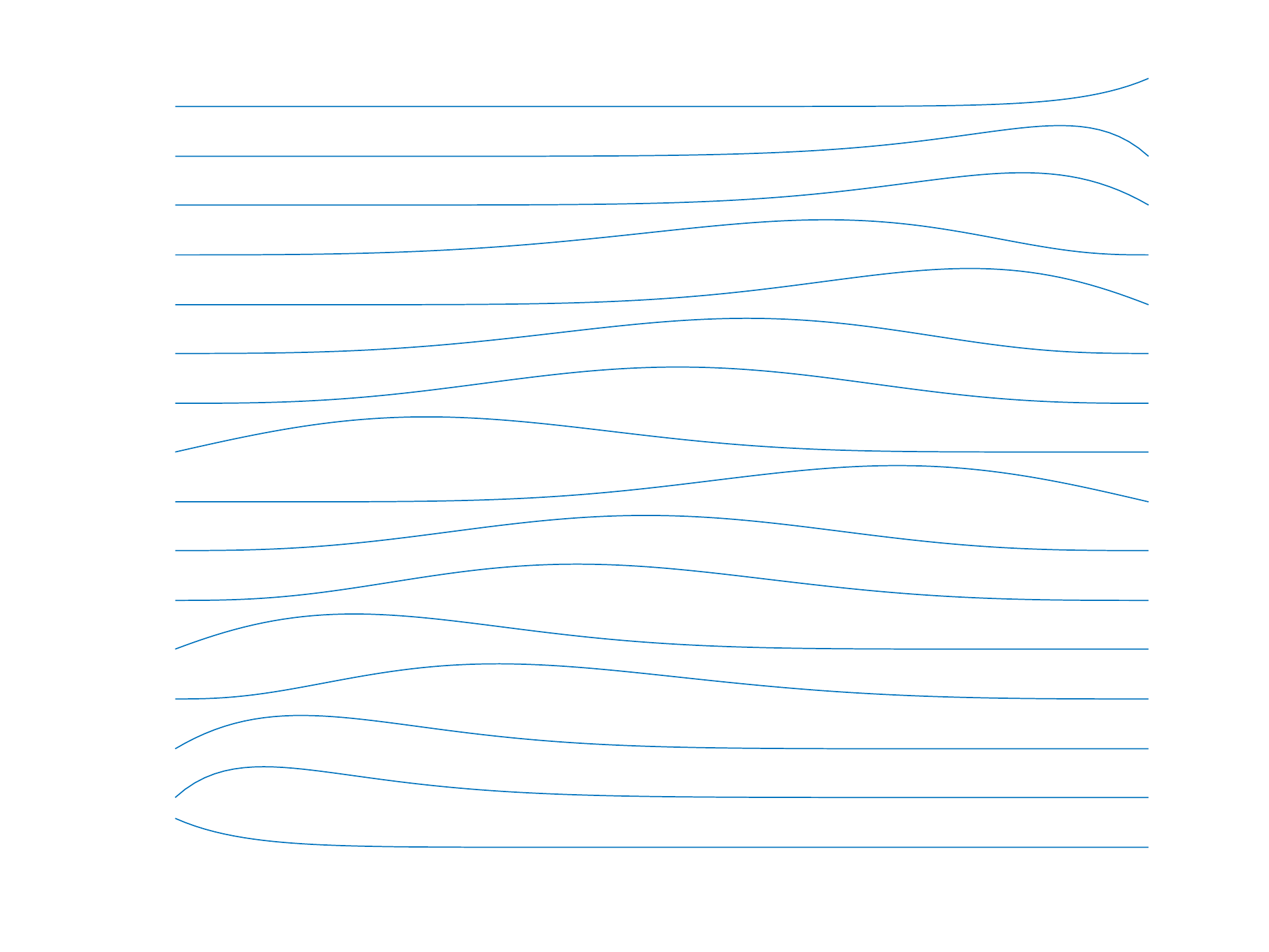}}
  \centerline{(b) $N=16$}\medskip
\end{minipage}
\caption{Computational polarization of Uniform$[0,1]$ at different construction sizes.\label{fig:wave8and16}}
\end{figure}

\begin{figure}[!t]
\begin{minipage}[b]{0.49\linewidth}
  \centering
  \centerline{\includegraphics[width=8cm]{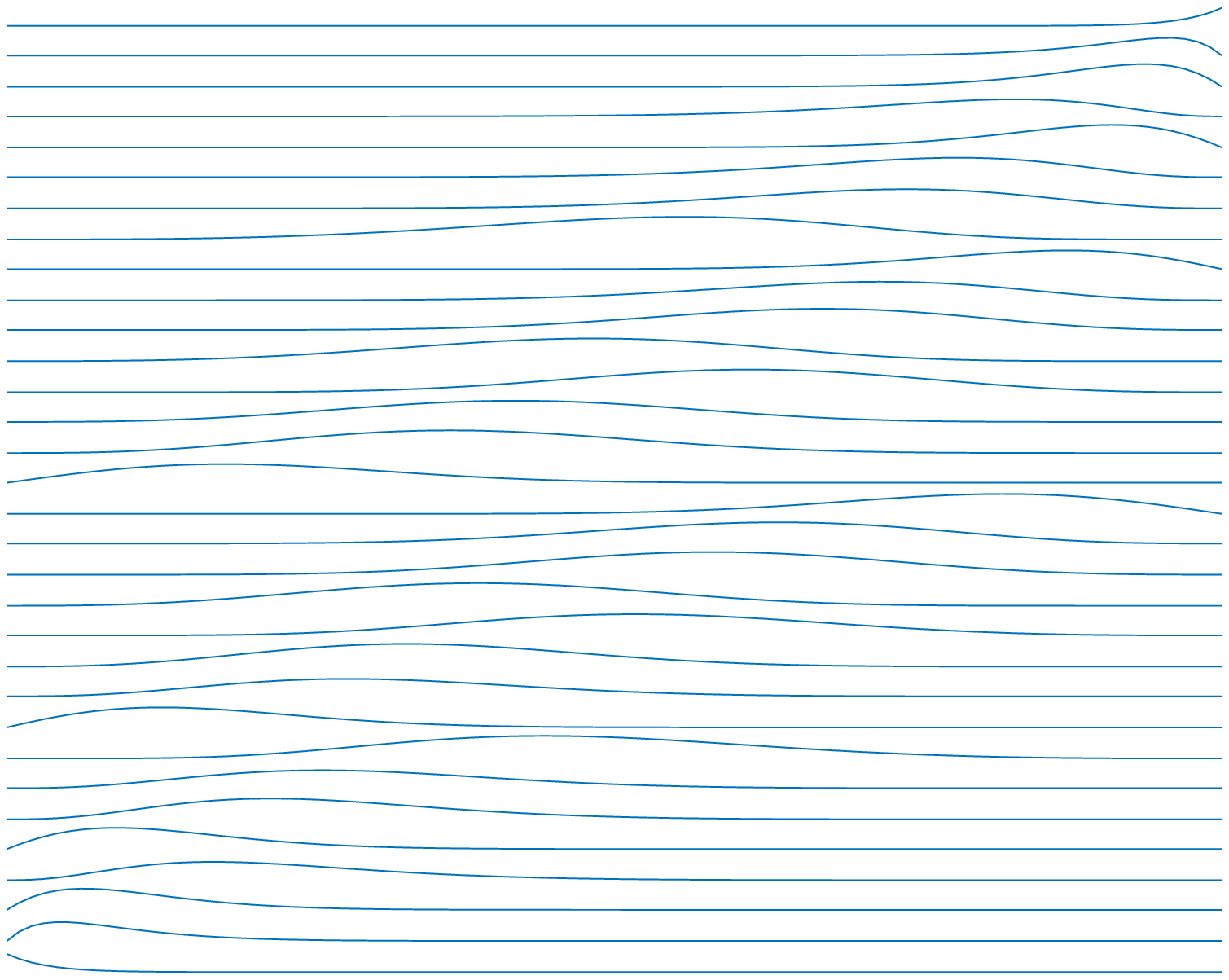}}
  \centerline{(a) $N=32$}\medskip
\end{minipage}
\hfill
\begin{minipage}[b]{0.49\linewidth}
  \centering
  \centerline{\includegraphics[width=8cm]{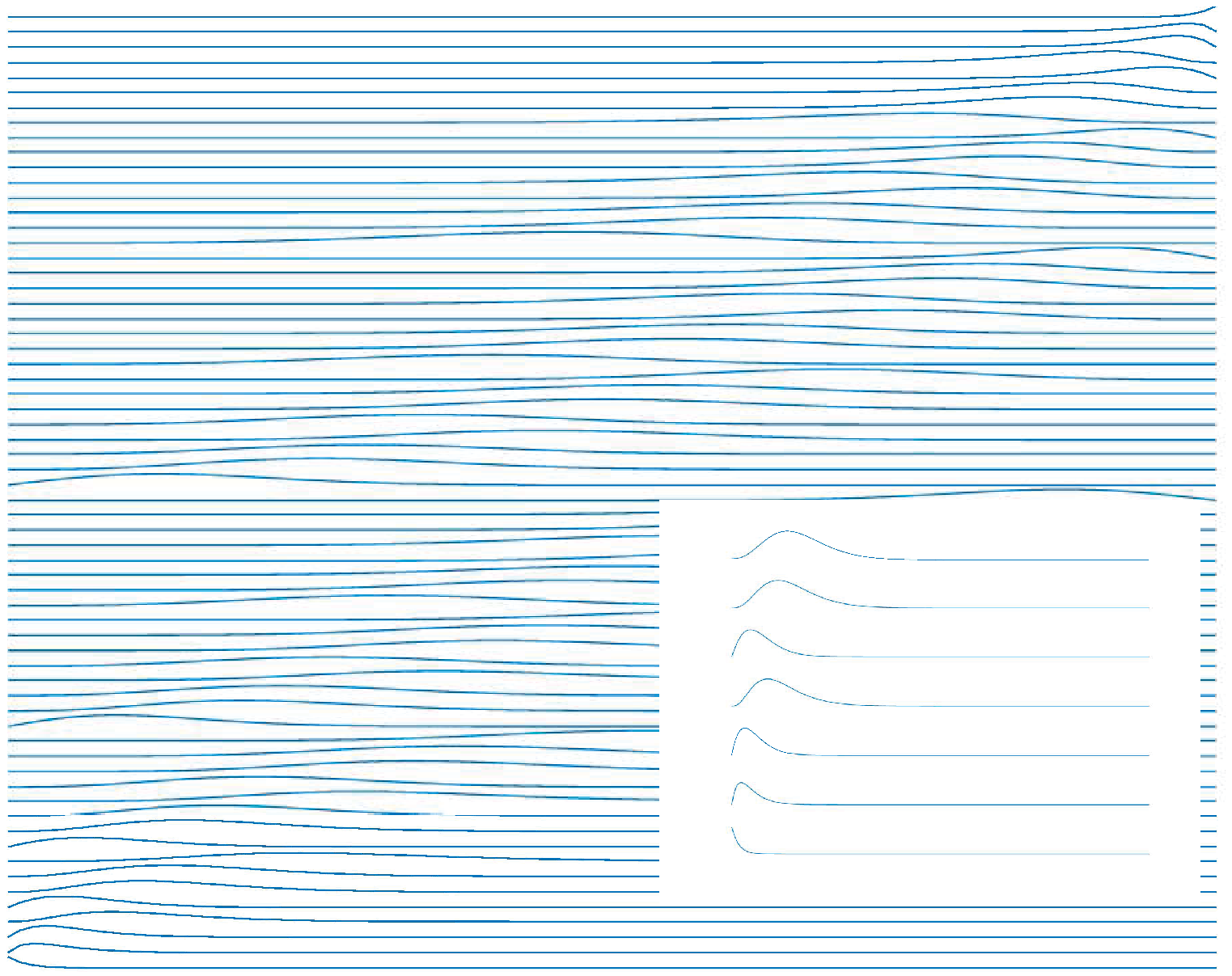}}
  \centerline{(b) $N=64$ (the inset zooms in on the bottom rows)}\medskip
\end{minipage}
\caption{Computational polarization of Uniform$[0,1]$ at different construction sizes.\label{fig:wave32and64}}
\end{figure}

\begin{example}[Asymptotics of the uniform runtime distributions]

Let us consider $F = \mathrm{Uniform}[0,1]$, the continuous uniform distribution on $[0,1]$. It can be shown that the worst machine computation time $T_{000\cdots}=\max_{i=1,...,N}\, T^{(i)}$ admits a probability density function $f(t) = N t^{N-1}$. We also obtain that $T_{000\cdots}$ converges to 1 in probability as follows
\begin{align*}
\mathbb{P}\left[ T_{000\cdots} \le t \right] &= \, \mathbb{P}\left[ \max_{i\in[N]}\,  T^{(i)} \le t \right]=  \,\mathbb{P}\left[ T^{(i)} \le t~~\forall i\in[N] \right]=   \, \prod_{i=1}^N \mathbb{P}\left[ T^{(i)} \le t \right]=   \, t^N
\end{align*}
We then have $\lim_{N\rightarrow \infty} \mathbb{P}\left[T_{000\cdots} < 1 \right] = 0$,
and consequently $T_{000\cdots}$ converges to $1$ in distribution. Moreover, the convergence can be improved to almost sure convergence. We first note that
\begin{align*}
\sum_{N=1}^\infty \prob{ \max_{i\in[N]}\,  T^{(i)} \le t } 
 = \sum_{N=1}^\infty t^{N}
= \frac{1}{1-t} <\infty \mbox{\qquad for $t<1$.}
\end{align*}
By applying Borel-Cantelli lemma we obtain
\begin{align*}
\lim_{M\rightarrow \infty} \prob{ \bigcup_{N=M}^\infty \{ \max_{i\in[N]}\,  T^{(i)} <1 \} }= 1\,,
\end{align*}
and conclude that $T_{000\cdots} $ converges to $1$ almost surely.
 A parallel argument shows that the computation time $T_{111\cdots}$ converges to $0$ almost surely.
\end{example}

\begin{figure}[!t]
\begin{minipage}[b]{0.49\linewidth}
  \centering
  \centerline{\includegraphics[width=7cm]{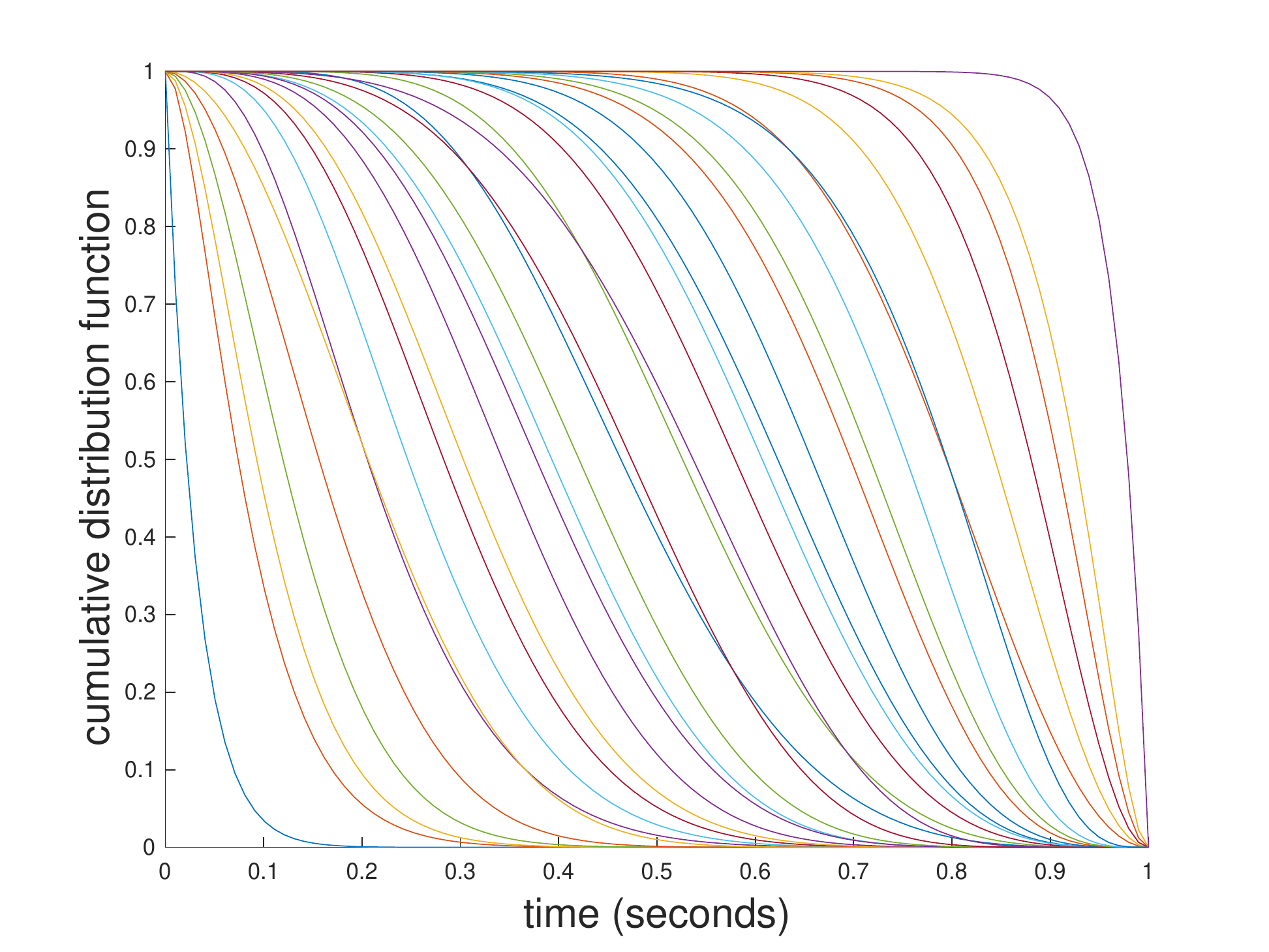}}
  \centerline{(a) Polarized cumulative density functions}\medskip
\end{minipage}
\hfill
\begin{minipage}[b]{0.49\linewidth}
  \centering
  \centerline{\includegraphics[width=7cm]{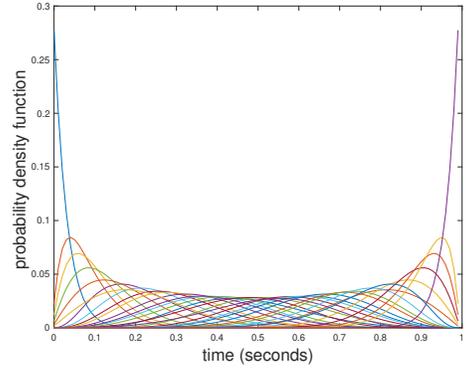}}
  \centerline{(b) Polarized probability density functions}\medskip
\end{minipage}
\caption{Polarized runtime distributions generated by a Uniform$[0,1]$ base distribution for $N=32$. Note that a subset of polarized CDF curves intersect at various points. Therefore, a total ordering in terms of stochastic dominance of the CDFs is not possible. \label{fig:polarizedcdfpdfuniform} }
\end{figure}

\begin{figure}[!t]
\begin{minipage}[b]{0.49\linewidth}
  \centering
  \centerline{\includegraphics[width=7cm]{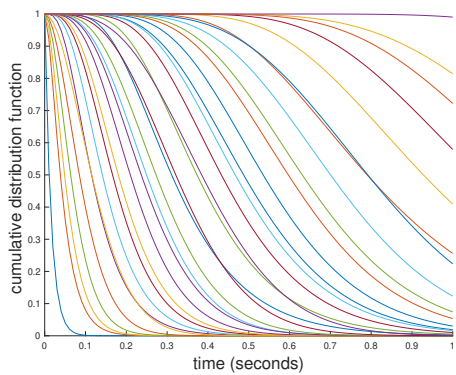}}
  \centerline{(a) Polarized cumulative density functions}\medskip
\end{minipage}
\hfill
\begin{minipage}[b]{0.49\linewidth}
  \centering
  \centerline{\includegraphics[width=7cm]{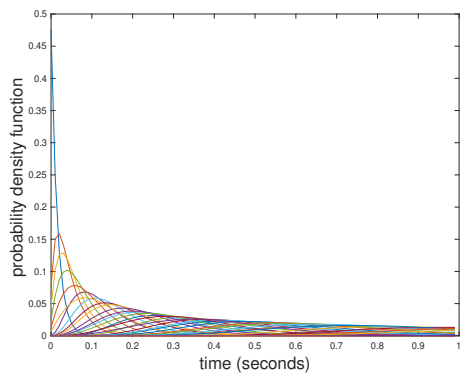}}
  \centerline{(b) Polarized probability density functions}\medskip
\end{minipage}
\caption{Polarized runtime distributions generated by a Exponential$(1)$ base distribution for $N=32$. The base distribution is not symmetric, which generates a non-uniform distribution of polarized probability density functions. \label{fig:polarizedcdfpdfexp}}
\end{figure}
In Figures \ref{fig:polarizedcdfpdfuniform}, \ref{fig:polarizedcdfpdfexp} and \ref{fig:universality} we contrast polarized computation times calculated according the uniform and exponential base distributions. Interestingly, the polarized distributions exhibit similar qualitative behaviors that resemble the Gaussian distribution as the construction size increases. We leave characterizing such universality properties as an interesting future research direction.

\subsection{Banach Space of Radon Measures, its Dual Space and $L_p$ Spaces of Functions}
In this section, we briefly introduce the mathematical background on Banach spaces and Radon measures to state our main results.
Recall that $(S,\Sigma)$ is a measurable space if $S$ is a nonempty set and $\Sigma$ is a $\sigma$-algebra of subsets of $S$. A measure on $(S,\Sigma)$ is defined as any map $\mu\,:\,\Sigma \rightarrow [0,+\infty)$ such that $\mu(\emptyset) =0$ where $\mu$ satisfies $\sigma$-additivity defined by
\begin{align*}
\mu \big( \cup_{n\in\mathbb{N}}\big) = \sum_{n\in\mathbb{N}}\mu(A_n)\,,
\end{align*}
for sets $A_n\in\Sigma$ which are pairwise disjoint.
A Radon measure on $\real$ is a finite measure which is regular, more precisely, for each Borel set $B\subset \real$
\begin{align*}
\mu(B) = \sup\left\{ \mu(C)\,:\, C\subseteq B \,\mbox{ is compact} \right\}\,.
\end{align*}
Given a measure $\mu$ on $\real$, we define its total variation by
\begin{align*}
\|\mu\|_{TV}:=\sup_{\pi}\sum_{A\in\pi}\vert\mu(A)\vert\,,
\end{align*}
where the supremum is taken over all partitions $\pi$ of $S$ into a countable number of disjoint measurable subsets. The set of Radon measures $\mu$ on $\real$ equipped with the total variation norm $\|\mu\|_{TV}$ is a Banach space which we denote by $\mathcal{M}(\real)$.

 We define $\mathcal{C}(\real)$ as the Banach space of continuous, real valued functions on $\real$
\begin{align*}
\mathcal{C}(\real):=\left\{f\,: \real \rightarrow \real\,,\,\mbox{$f$ is continuous}\right\}
\end{align*}
 equipped with the supremum norm
\begin{align*}
\|f\|_{\infty}:=\sup_{t\in \real}\, \vert f(t)\vert\,.
\end{align*}
Riesz–Markov–Kakutani representation theorem
 states that the dual space $\mathcal{C}^*(\real)$ is the set of Radon measures $\mathcal{M}(\real)$ (see e.g., Chapter 7 in \cite{folland1999real}). More precisely, for every positivity preserving linear functional $\Phi(f)\,:\mathcal{C}(\real)\rightarrow \mathcal{C}(\real)$, such that $\Phi(f)\ge 0$ whenever $f\ge 0$, there exists a unique Radon measure $\mu$ such that

\begin{figure}[!t]
\begin{minipage}[b]{0.49\linewidth}
  \centering
  \centerline{\includegraphics[width=8cm]{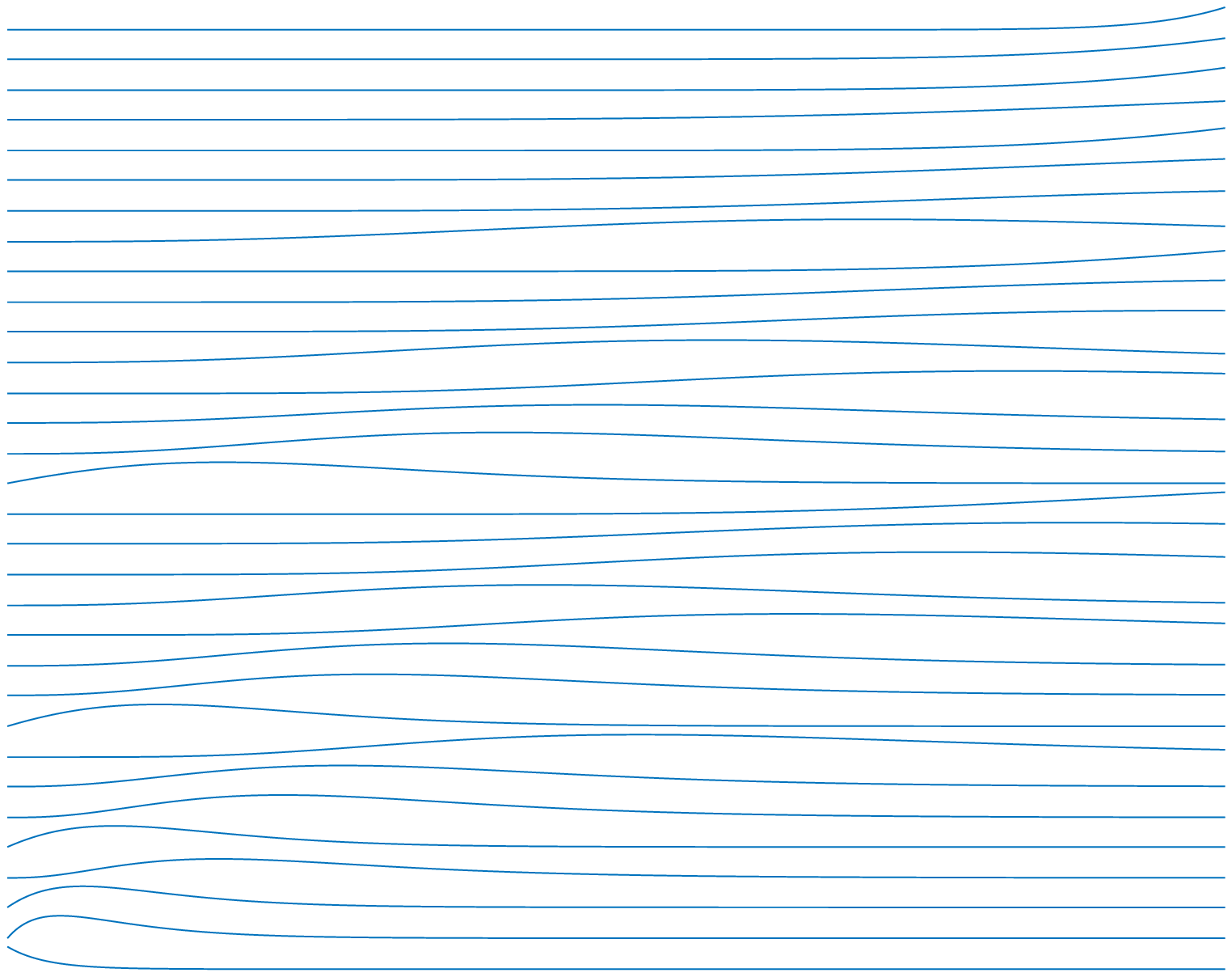}}
  \centerline{(a) Uniform[0,1]}\medskip
\end{minipage}
\hfill
\begin{minipage}[b]{0.49\linewidth}
  \centering
  \centerline{\includegraphics[width=8cm]{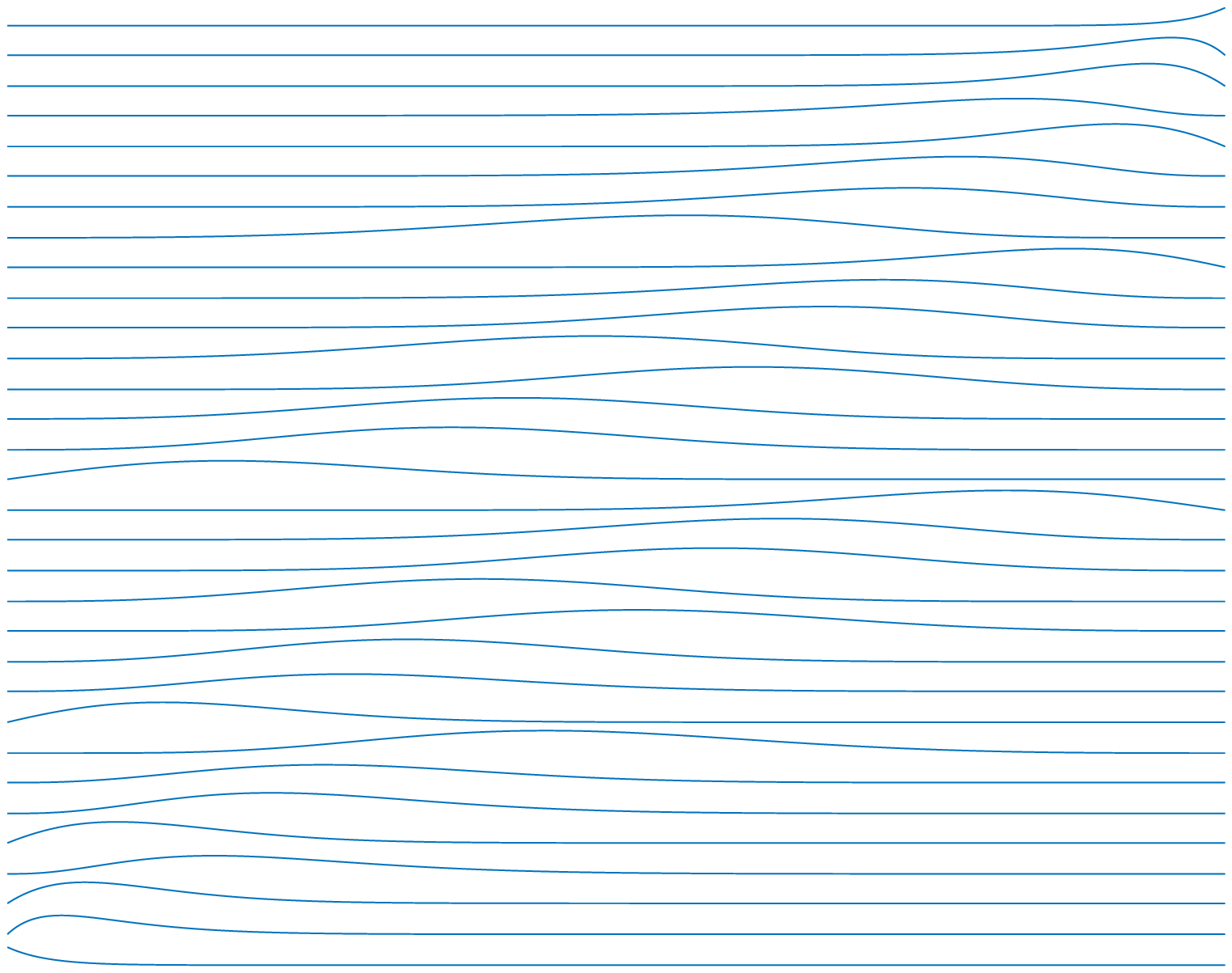}}
  \centerline{(b) Exponential(1)}\medskip
\end{minipage}
\caption{Polarization of Uniform and Exponential return distributions for construction size $N=32$ \label{fig:universality}}
\end{figure}

\begin{align*}
\Phi(f)=\int f d\mu\,.
\end{align*}
An implication of this bijective correspondence is the dual representation of the total variation norm 
\begin{align}
\|\mu\|_{TV} = \sup_{f\in\mathcal{C}(\real),\,\|f\|_{\infty}\le 1} \left \vert \int f d\mu \right\vert\,. \label{TVdualformula}
\end{align}

 Therefore, we can identify Radon measures as the Banach space of bounded linear functionals on $\mathcal{C}(\real)$ with respect to the supremum norm.

We can construct Radon measures directly from cumulative distribution functions. Let $F$ be a right-continuous and non-decreasing function with support $[a,b]\subseteq \real$. We define the generalized inverse distribution function $G(u) = \sup_{t\in \real} t\, {\mbox{ s.t. }\, F(t) \le u}$.
Next, we define the Radon measure $\mu$ as
\begin{align*}
\langle \mu, f\rangle = \int_{a}^b f(G(u))du\,.
\end{align*}
Here, the preceding Lebesgue-Stieltjes integral can also be written as
\begin{align*}
\langle \mu, f\rangle = \int_{-\infty}^{\infty} f(t) dF(t)\,.
\end{align*}
Let us define the Dirac delta measure $\delta_{a}$ as the Radon measure which satisfies
\begin{align*}
\langle \delta_a, f \rangle = f(a)\,.
\end{align*}
From the dual representation of the total variation norm in \eqref{TVdualformula}, we observe that $\|\delta_a\|_{TV}=1$. In fact, it can be shown that Dirac delta measures are the extreme points of the unit total variation ball given by $\{\mu \in \mathcal{M}(\real):\, \|\mu\|_{TV}\le 1\}$. 
It can be verified that $\delta_{a}$ is generated by the cumulative distribution function%
\begin{align*}
F(t) = \begin{cases} 0 \quad \mbox{ for } \quad t<a \\ 
1 \quad \mbox{ for } \quad t\ge a\,.
\end{cases}
\end{align*}
We use the notation $\frac{d}{dt}F(t)=\delta_a$ to indicate this correspondence, which formally means that
\begin{align}
 \int f(t) dF(t)&=\langle f, \delta_a\rangle \nonumber =f(a).\label{eqDiracDefn}
\end{align}
Note that, although $F(t)$ is not differentiable at every point, $\delta_a \in \mathcal{M}(\real)$ and $F(t)\in \mathcal{C}(\real)$ are well-defined according to the above, and our notation $\frac{d}{dt}F(t)=\delta_a$ is justified.

\noindent\textbf{$L_p$ spaces of functions\\}
We define the $L_p$ norm of a function $f(t):\real\rightarrow\real$ with respect to the Lebesgue measure by%
\begin{align*}
\| f\|_{L_p} \triangleq \left( \int \vert f(t)\vert^p dt \right)^{1/p}\quad\mbox{for}\quad 1\le p<\infty\,,  and
\end{align*}
\begin{align*}
\| f\|_{L_\infty} = \mbox{ess sup}f(\cdot)\quad\mbox{for}\quad p=\infty,
\end{align*}
where $\mbox{ess sup}f(\cdot)$ stands for the essential supremum of $f$.
The space $L_p(\real)$ consists of measurable functions $\real\rightarrow \real$ such that $\| f\|_{L_p}<\infty$, where two measurable functions are equivalent if they are equal almost everywhere with respect to the Lebesgue measure.

\subsection{Freezing and Local Assignment of Computations}

In polar codes, the channels with large Bhattacharyya parameters can be \emph{frozen} by fixing their inputs to known quantities, e.g., zero bits. This process ensures reliable decoding of the successive cancellation decoding and enables capacity achieving behavior of the Polar codes. Note that this is possible since the set of real numbers is a totally ordered set. However, this simple fact is no longer true for the elements of a vector space martingale, whose elements consist of continuous functions.

In computational polarization, one can freeze workers by transmitting fixed data such as zero matrices. As a result, undesirable runtimes can be avoided in a straightforward manner. However, the functional nature of the martingale process presents another theoretical challenge in deciding the set of workers to freeze, which correspond to individual density functions. In contrast, in traditional Polar Codes, the freezing set can be determined simply by sorting scalar channel reliability values and freezing unreliable channels. Here, we explore multiple options for the freezing set with different overall runtime design considerations and theoretical guarantees. 
\subsubsection{Quantile rule}
We pick an arbitrary time instant $t^*\in \real$, and consider the set of indices that solve
\begin{align}
\min_{S\,:\,|S|\le RN} \sum_{i\in S} F_{n,i}(t^*)\,,
\end{align}
which is achieved by picking the indices of the smallest $RN$ values of $F_{n,i}(t^*)$ for fixed $n$ and $t^*$.  The main motivation behind this rule is the union bound applied to the recovery time of the computation. Consequently the complement of the set $S$ is frozen. However, the role of the value of $t^*$ is non-obvious and yet to be determined. In Section \ref{sec:runtime}, we show that the quantile value $t^*=F^{-1}(R)$ has desirable properties in the recovery time.

\subsubsection{Laplace transform rule}  
Here we consider the Laplace transform of the runtimes
\begin{align*}
M_{n,i}(\lambda) = \int e^{\lambda t} dF_{n,i}(t)\,.
\end{align*}
Now we consider the set of indices that solve
\begin{align}
\min_{S\,:\,|S|\le RN} \sum_{i\in S} M_{n,i}(\lambda)\,.
\end{align}
Similarly, the above optimal set is achieved by picking the indices of the smallest $RN$ values of $M_{n,i}(\lambda)$ for fixed $n$ and $\lambda$.
This rule significantly differs from the quantile rule. The motivation behind this rule is a tight upper-bound on the expected runtime in terms of the Laplace transform, as we will establish in Theorem \ref{thm:runtimeguarantees} of Section \ref{sec:runtime}.

\subsection{Fixed Points and Cycles of the Functional Process}
Fixed points of the functional process in \eqref{eq:functional_update} are given by the following equation
\begin{align*}
F(t)(1-F(t))=0\,.
\end{align*}
It follows that the fixed points satisfy $F(t)=0$ or $F(t)=1$ for all t. Therefore the fixed points are step functions $F(t)=1[t\le t_0]$ for some 
$t_0 \in \reals$.
When the process is initialized with one of these fixed points, we have $F_n(t)=F(t),\,\forall n\in\mathbb{N}$.
It is critical to note that there exists certain paths that exhibit nonconvergent behavior. Particularly, suppose that the base distribution $F(t)$ is such that $F_{10}=F(t)$. Then, it follows that $F_{1010}=F(t)$ and $F_{10101010\cdots}=F(t)$ ad infinitum. Although in the martingale process this event has probability converging to zero, the fixed points can be characterized in terms of polynomial equations. We now illustrate this for the case of length two cycles in the following example. 
\begin{example}[Length two cycles]
~\\
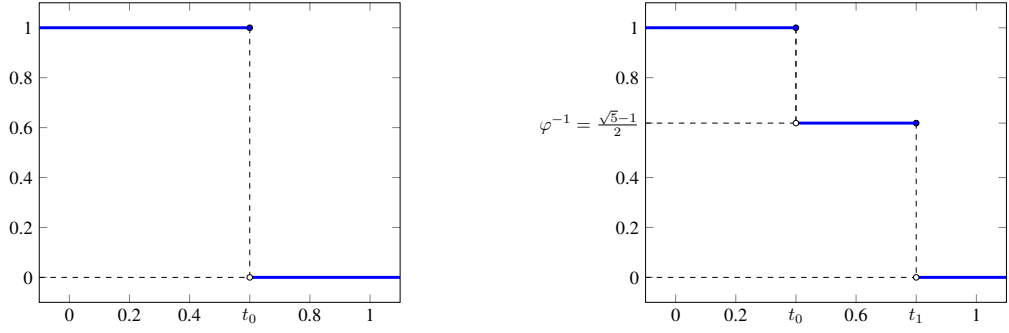
\begin{figure}[t!]
\begin{center}
\begin{minipage}{0.45\textwidth}
\centering
\begin{tikzpicture}[scale=0.7]
\begin{axis}[
				xmin=-0.1,
				xmax=1.1,
				ymin=-0.1,
				ymax=1.1, 
				samples=50, 
				ytick = {0,0.2,0.4,0.6,0.8,1},
  				yticklabels = {0,0.2,0.4,0.6,0.8,1},
				xtick = {0,0.2,0.4,0.6,0.8,1},
  				xticklabels = {0,0.2,0.4,$t_0$,0.8,1},
			]
  \addplot[blue, ultra thick, domain=-0.1:0.6] (x,1);
  \addplot[blue, ultra thick, domain=0.6:1.1] (x,0);
  \draw[dashed] (0.6,0) -- (0.6,1);
  \draw[dashed] (-0.1,0) -- (0.6,0);
  \draw [fill=blue] (0.6,1) circle [radius=1.5pt];
  \draw [fill=white] (0.6,0) circle [radius=1.5pt];
\end{axis}
\end{tikzpicture}
\end{minipage}
\begin{minipage}{0.45\textwidth}
\centering
\begin{tikzpicture}[scale=0.7]
\begin{axis}[
				xmin=-0.1,
				xmax=1.1,
				ymin=-0.1,
				ymax=1.1, 
				samples=50, 
				ytick = {0,0.2,0.4,0.618,0.8,1},
  				yticklabels = {0,0.2,0.4,$\varphi^{-1}=\frac{\sqrt{5}-1}{2}$,0.8,1},
				xtick = {0,0.2,0.4,0.6,0.8,1},
  				xticklabels = {0,0.2,$t_0$,0.6,$t_1$,1},
			]
  \addplot[blue, ultra thick, domain=-0.1:0.4] (x,1);
  \addplot[blue, ultra thick, domain=0.4:0.8] (x,0.618);
  \addplot[blue, ultra thick, domain=0.8:1.1] (x,0);
  \draw[dashed] (-0.1,0.618) -- (0.4,0.618);
  \draw[dashed] (-0.1,0) -- (0.8,0);
  \draw[dashed] (0.4,1) -- (0.4,0.618);
  \draw[dashed] (0.4,1) -- (0.4,0.618);
  \draw[dashed] (0.8,0.618) -- (0.8,0);
  \draw [fill=white] (0.4,0.618) circle [radius=1.5pt];
  \draw [fill=blue] (0.4,1) circle [radius=1.5pt];
  \draw [fill=white] (0.8,0) circle [radius=1.5pt];
  \draw [fill=blue] (0.8,0.618) circle [radius=1.5pt];
\end{axis}
\end{tikzpicture}
\end{minipage}
\end{center}
\caption{Fixed point of the computational polarization process, CDFs with a single jump (left) and fixed point of length two cycles, CDFs with two jumps (right) and inverse golden ratio proportion. In both cases, $t_0,t_1 \in \real$ may assume arbitrary values.\label{fig:golden}}
\end{figure}

We define $k$-step fixed points as $F_{b_0b_1\cdots b_k}=F(t)$ for a bit sequence $b_0b_1\cdots b_k$. The two-step cycles corresponding to the infinite sequence $1010\cdots$ can be found using the following equation
\begin{align*}
F_{10}(t) &= 2(F(t)^2) - (F(t)^2)^2=F(t)\,.
\end{align*}
The above is a depressed quartic equation in $z=F(t)$
\begin{align*}
z^4 - 2z^2+z =0\,.
\end{align*}
The solutions of the preceding equation are given by $\big \{0,1,\varphi^{-1}, -\varphi\big\}$, where $\varphi=\frac{\sqrt{5}+1}{2}$ is the golden ratio. Recall that the golden ratio satisfies
\begin{align*}
 1 + \frac{1}{\varphi} = \varphi\,.
 \end{align*}
  Noting that $0\le F(t)\le 1$ for all $t$, the only valid solutions for $F(t)$ are $\{0,1,\varphi^{-1}\}$. Combined with the fact that $F(t)$ is non-increasing, we obtain that $F(z)$ is the sum of two unit step functions
\begin{align*}
F(t) &= (1-\varphi^{-1}) 1[t \le t_0]+\varphi^{-1}1[t \le t_1]\end{align*}
for some $t_0, t_1 \in \reals$ satisfying $t_0<t_1$. The corresponding probability distribution is given by
\begin{align*}
p(t) &= (1-\varphi^{-1}) \delta_{t_0}+\varphi^{-1} \delta_{t_1}
\end{align*}
for some $t_0,t_1$ such that $t_0<t_1$.
When the base distribution is given by this family of distributions, it holds that $F_{0101\cdots}(t)=F(t)$, which does not converge to the Dirac delta measure.

\end{example}

\section{Limiting Distribution and Achieving Optimal Runtime}
\label{sec:runtime}
In this section we establish the limiting distribution of the function process, and prove that the computational polarization scheme can achieve the optimal runtime.

\subsection{Evolution of the Cumulative and Probability Density Function}
We index the cumulative density functions in the computational polarization process by bit sequences analogous to the analysis in Polar Codes. The CDF resulting from the maximization, i.e., slower worker, is indexed with 1, whereas the lower CDF resulting from the minimization, i.e., faster worker, is indexed with 0 at every layer. 

\noindent\textbf{The Functional Tree Process:}
Consider an infinite binary tree, where each node is associated with an element of a Banach space, e.g., continuous functions in $L_p$. We start with the null sequence at the root node, which corresponds to the base distribution $F(t)$. At the construction level $n$ and a CDF indexed by $b_1b_2\cdots b_n$, the slower worker node constructed from it is indexed by $b_1b_2\cdots b_n 1$. Similarly, the faster worker constructed from the CDF indexed by $b_1b_2\cdots b_n$ is indexed by $b_1b_2\cdots b_n 0$. We denote the CDF indexed by the bit sequence $b_1b_2\cdots b_n$ by $F^{b_1b_2\cdots b_n}(t)$.

In order to analyze the evolution of the polarized CDFs, we define a functional random tree process $\{F_{n}(t),\,n\ge 0\}$. Consider the probability space $(\Omega,\mathcal{B},\mathbb{P})$, where $\Omega$ is the space of all binary sequences $(b_1,b_2,...)\in \{0,1\}^{\infty}$, $\mathcal{B}$ is the Borel field generated by the cylinder sets
\begin{align*}
S(b_1,b_2,...):=\{ w \in \Omega,:\, w_1=b_1, \cdots, w_n=b_n \} \mbox{ for $n\ge 1$ }
\end{align*}
and $\mathcal{P}$ is the probability measure defined on $\mathcal{B}$ such that $\mathbb{P}[S(b_1,b_2,...b_n)]=2^{-n}$. Define $\mathcal{B}_n$ as the Borel field generated by the cylinder sets $S(b_1,b_2,...b_i)$, $1\le i\le n$ and $\mathcal{B}_0$ as the null set. Now we define the random process $F_n(t)=F^{b_1b_2\cdots b_n}(t)$ for $w=(w_1,w_2,...)\in \Omega$ and $n\ge 1$.

The CDF functional process is defined by
\begin{align}
F_{n+1}(t) = \begin{cases} F_n(t)^2 & \mbox{if } b_n=1 \\ 1-(1-F_n(t))^2 & \mbox{if } b_n=0 \\ \end{cases}.
\end{align}
An alternative way to express the same process is the following equation
\begin{align}
\label{eq:cdfalternative}
F_{n+1}(t) = F_n(t) + \epsilon_n F_n\big(1-F_n(t)\big),\,\forall t\,,
\end{align}
where we have introduced $\pm 1$ valued Rademacher random variables, defined as $\epsilon_n:=1-2b_i$ satisfying $\prob{\epsilon_n=+1}=\prob{\epsilon_n=-1}=\frac{1}{2},\,\forall n\in\mathbb{N}$.

\begin{prop}
\label{prop:banachmartingale}
The sequence of random functions $F_n(t)\in L_p$ and Borel fields $\{\mathcal{B}_n,n\ge 0\}$ is a Banach space martingale. More precisely, we have
\begin{align*}
& \mathcal{B}_n \subset \mathcal{B}_{n+1} \mbox{ and } F_n(t) \mbox{ is } \mathcal{B}_{n}\mbox{-measurable}\\
& \Exs[ \| F_n \|_{L_p} ] < \infty\\
& F_n(t) = \Exs[ F_{n+1}(t) | \mathcal{B}_n)]\,.
\end{align*}
\end{prop}
\begin{proof}
Noting that $\Exs [\epsilon_n]=0,\,\forall n$, the representation \eqref{eq:cdfalternative} immediately verifies the martingale property 
\begin{align}
\Exs[F_{n+1}(t) \,|\, \epsilon_1,\cdots,\epsilon_n] = F_n(t),\,\forall t\,,
\end{align}
where we use the conditioning $\{ \epsilon_i\}_{i=1}^n$ and $\mathcal{B}_n$ in an exchangeable manner.
\end{proof}
Next, we provide a convergence result that describes the pointwise asymptotic behavior of $F(t)$, i.e., for fixed values of $t\in\real$.
\begin{prop} \label{prop:pointwisecdf}The process $\eqref{eq:cdfalternative}$ pointwise converges almost surely to a limiting random variable taking values in $\{0,1\}$. More precisely, for fixed $t\in \real $ we have
\begin{align*}
\prob{ \lim_{n\rightarrow \infty} | F_n(t) - F_\infty(t) | =0 } = 1\,, 
\end{align*}
where $\prob{F_{\infty}\in\{0,1\}}=1$.
\end{prop}
\begin{proof}
This result parallels Theorem \ref{thm:arikanbec} on the polarization of the binary erasure channel due to Arikan and follows from Doob's martingale convergence theorem for scalar valued martingales as stated in Theorems \ref{thm:martingalescalarlp} and \ref{thm:martingalescalaralmostsure}. A crucial assumption is the boundedness of $F_n(t)\in[0,1]$, which follows since $F_n(t)=\prob{T_{b_1\cdots b_m} \le t}$ for some bit sequence $b_1\cdots b_m$  by construction.
\end{proof}
The above result implies that the cycles of the process as shown in Figure \ref{fig:golden}, only occur with zero probability.

\subsection{Limiting Distributions}

Now we characterize the limiting distribution of the functional process to establish the asymptotic behavior of polarized runtimes.
\begin{theorem}[Asymptotic runtime distribution]
\label{thm:maindiracdelta}
Consider the functional polarization process
\begin{align}
F_{n+1}(t) = 
\begin{cases} 
F^2_n(t), &\text{$\forall t$ if } u_t = 0 \label{eq:cdfprocessinthm}\\
1-\big(1-F(t)\big)^2 &\text{$\forall t$ if } u_t = 1\,.
\end{cases}
\end{align}
Suppose that the mean of the base runtime distribution $F_0(t):=F(t)$ is finite. Then, the polarized runtime distributions $F_n(t)$ converge almost surely to the Dirac delta measure $\delta_{t^*}$ almost everywhere, where $t^*$ is a random variable distributed according to the base runtime distribution $F(t)$. More precisely, we have
\begin{align*}
\prob{\lim_{n\rightarrow \infty} \| F_n(t) - F_{\infty}(t)\|_{L_p} = 0 } = 1,\,
\end{align*}
for any $p\in(1,\infty)$, where $\frac{d}{dt}F_{\infty}(t)=\delta_{t^*}$ and $t^*\sim F(t)$ is a random variable distributed according to $\prob{t^*>t}=F(t)$.
\end{theorem}
\begin{remark}
It is remarkable that the computational polarization process \eqref{eq:cdfprocessinthm} converges to Dirac distributions that are maximally concentrated in time (see e.g., Figure \ref{fig:polarizedcdfpdfuniform}), which are deterministic functions of the path $b_1\cdots b_n$, while the distribution of their location is completely characterized by the base distribution $F(t)$. This is consistent with the fact that $T_{b_1\cdots b_n}$ is a permutation of $T^{(1)},\ldots,T^{(N)}$, and a uniformly random path in the tree has distribution identical to $T^{(1)}$, i.e., $F(t)$.
\end{remark}
\begin{proof}
First note that the process $F_n(t)$ verifies the boundedness condition $\sup_{n} \Exs \|F_n(t)\|<\infty$ in order to apply the Banach space martingale convergence theorem (see Appendix \ref{sec:appendixbanachmartingale}). Suppose $p\in(1,\infty)$ and observe that
\begin{align*}
\Exs \|F_n(t)\|_{L_p} 
&= \Exs \left(\int_{-\infty}^{\infty} \vert F_n(t) \vert^p dt \right)^{1/p} 
\le \left(\Exs \int_{-\infty}^{\infty} \vert F_n(t) \vert^p dt \right)^{1/p}\le \left(\Exs \int_{-\infty}^{\infty} F_n(t) dt \right)^{1/p}\,,
\end{align*}
where in the first inequality we have applied Jensen's inequality leveraging the concavity of the map $u\in\real_+ \rightarrow (u)^{1/p}$ for $p\in(1,\infty)$. The second inequality follows from $\vert F_n(t) \vert^p \le \vert F_n(t) \vert = F_n(t)$ since $F_n(t)\in[0,1]$ assumes valid probability values. Next, noting that $F_n(t)$ is integrable, we apply Fubini's theorem to reach 
\begin{align*}
\Exs \|F_n(t)\|_{L_p} 
&\le \left(\int_{-\infty}^{\infty} \Exs F_n(t) dt \right)^{1/p} = \left(\int_{-\infty}^{\infty} F(t) dt \right)^{1/p} = \left(\int_{-\infty}^{\infty} t dF(t) \right)^{1/p} = (\Exs[T])^{1/p}\,,
\end{align*}
where $\prob{T\le t}=F(t)$. By our assumption, $\Exs[T]<\infty$ and the above chain of inequalities imply that $\Exs \| F_n(t) \|_{Lp}<\infty$ for all $n\in \mathbb{N}$.

The space $L_p(\real)$ equipped with the norm $\|f\|_{L_p}$ forms a Banach space. We first note that $F_n(t),\, n\in\nat$ is an integrable martingale taking values in this Banach space since
\begin{align*}
\Exs[ F_{n+1} | \mathcal{B}_n ] &= F_n(t),\,\forall t\in\real
\end{align*}
Next, we apply the convergence theorem for Banach space valued martingales as stated in Theorem \ref{thm:Banachmartingale} (see Appendix \ref{sec:appendixbanachmartingale}), which implies that $F_n(t)$ converges almost surely to an integrable random variable $F_{\infty}(t)\in L_p$ taking values in the Banach space $L_p$.
We next identify the distribution of $F_{\infty}(t)$. By Proposition \ref{prop:pointwisecdf}, $F_{\infty}(t)$ takes values in $\{0,1\}$ with probability one, for any $t\in \real$. Since $F_{n}(t) = \prob{ T_{b_1\cdots b_m} \le t }$ for some bit sequence $b_1\cdots b_m$, for all $n\in\mathbb{N}$,  $F_{\infty}(t)$ is a non-decreasing function. Therefore, with probability one, $F_{\infty}(t)$ equals one for $t\in [t^*,\infty)$ for some $t^*$, and equals zero otherwise. Next, we determine the distribution of the random threshold $t^*$.
Note that the martingale property implies
\begin{align}
F(t)&=\Exs [F_{\infty}(t)|\mathbb{B}_{\infty}]\nonumber\\
    &= \prob{F_{\infty}(t)=1}\,, \label{eq:cdfinfvalue}
\end{align}
where the final identity leverages the fact that $\prob{F_\infty(t)\in\{0,1\}}=1$. Finally, we notice that
\begin{align}
\prob{t^* \le t} &= \prob{F_{\infty}(t)=1}\\
& = F(t)\,,
\end{align}
where the final equality follows from the earlier identity \eqref{eq:cdfinfvalue}. We conclude that $F_{\infty}(t)$ is of the form $1[t\ge t^*]$, where $t^*$ is distributed according to CDF $F(t)$.

\end{proof}

Another process of interest is the characteristic functions defined by
\begin{align*}
\varphi_{n}(v) = \int e^{i v t} dF_n(t)\,.
\end{align*}
An interesting observation is that the characteristic functions also form a functional martingale process as the next result illustrates.

\begin{prop}
The characteristic function of the runtime random variables obey the following recursion
\begin{align}
\varphi_{n+1}(v) = 
 \begin{cases}
  -2iv^{-1}\varphi_{n}(v)  * \varphi_{n}(v) & \mbox{if } \epsilon_n=+1 \\ 2 \varphi_{n}(v)+2iv^{-1}\varphi_{n}(v)  * \varphi_{n}(v) & \mbox{if } \epsilon_n=-1 \\ 
  \end{cases}
  \label{eq:charfuncprocess}
\end{align}
where $*$ denotes the continuous convolution operation, where $(f*g)(u)=\int_{-\infty}^\infty f(t)g(u-t)dt$. Consequently the sequence of random functions $\varphi_n(t)\in L_p$ and Borel fields $\{\mathcal{B}_n,n\ge 0\}$ is a Banach space martingale.
\end{prop}
The proof of the above result is similar to the proof of Theorem \ref{thm:Banachmartingale} and omitted.
Note that $F_n$ is a bounded measure and the integral $\int e^{ivt} dF_{\infty}(t)$ is well-defined. The limiting distribution of $\varphi_n(v)$ as $\rightarrow \infty$, denoted by $\varphi_{\infty}(v)$ can be identified by the dominated convergence theorem and the relation $\varphi_{\infty}(v)=\int e^{ivt} dF_{\infty}(t)$ as a random complex exponential
\begin{align*}
\varphi_{\infty}(\lambda) \eqindist e^{i\lambda t^*}\,,
\end{align*}
where $P(t^* \le t ) = F(t)$.

\subsection{Laplace Transform}

We obtain the Laplace transform martingale by defining the random process 
\begin{align}
M_n(\lambda) &=\int e^{\lambda t} dF_n  \\
&= \int e^{\lambda t} p_n(t)dt\,.
\end{align}
It also holds that $M_n(\lambda)=\varphi_n(-i\lambda)$, which relates the characteristic function with the Laplace transform.

\subsection{Runtime of the MDS Scheme}
The overall runtime of the MDS coded computation scheme for absolutely continuous CDFs can be characterized via the following result.
\begin{theorem}
\label{thm:returntimes}
Suppose that $F$ is an absolutely continuous CDF and let $T_D$ be the overall runtime of the MDS coding scheme. Then it holds that
\begin{align}
T_D \xrightarrow{d} \mathcal{N} \Big( F^{-1}(R), \frac{1}{N}\frac{R(1-R)}{p(F^{-1}(R))^2}\Big)\,,
\end{align}
and
\begin{align*}
\lim_{n \rightarrow \infty} \prob{T_D > F^{-1} (R) + t } \le e ^{-\frac{N t^2c_1}{2}}\,,
\end{align*}
where $c_1$ is a constant that only depend on the rate $R=\frac{K}{N}$ and $F$.
\end{theorem}
\begin{proof}
Suppose that $T^{(1)}$, $T^{(2)}$, ..., $T^{(N)}$ are sampled from a random variable distributed according to the cumulative probability density $F(t)$. MDS coded computation provides successful decoding when any $K$ workers finish their tasks (see e.g. \cite{Lee2018}). Let $T_{(1)}$, $T_{(2)}$, ..., $T_{(N)}$ be their sorted values in increasing order. Then the smallest $K$-th value, i.e., $K$-th order statistics where $K=NR$ converges in distribution to $\mathcal{N} \Big( F^{-1}(R), \frac{1}{N}\frac{R(1-R)}{p(F^{-1}(R))^2}\Big)$.  The tail bound follows from elementary bounds on the Gaussian tail probability (see, e.g., \cite{vanderVaart}).%
\end{proof}
We now show that the overall  runtime $F^{-1}(R)$ is optimal.
\begin{theorem}[Information-theoretic lower-bound]
\label{thm:infotheoryoptimality}
For any $\beta\in(0,1)$, any coded computation scheme with rate $R=\frac{K}{N}$ obeys the lower bound
\begin{align*}
\prob{T_D \ge F^{-1}(R(1-\beta))}\ge \beta - \frac{1}{NR} %
\end{align*}
Therefore, the overall runtime $T_D$ is necessarily lower bounded by $F^{-1}(R)$ for vanishing error probability as $N\rightarrow \infty$ with any coding scheme.

\end{theorem}
The proof of this Theorem can be found in Section \ref{sec:proofs}.

\subsection{Runtime of the Computational Polarization Scheme}
Define $T_{D}$ as the earliest time such that the decoder can output the correct result of the computation. In the computational polarization scheme, the total runtime is the maximum of the polarized computation times which are not frozen. This is given by the expression
\begin{align*}
T_D := \max_{i \notin \mathcal{F}}\, T_{n,i}  
\end{align*}
where $\mathcal{F}$ is the set of frozen workers.
\begin{theorem}
\label{thm:runtimeguarantees}
(Asymptotic optimality)
Computational polarization scheme has the following runtime guarantees
\begin{align}
\mbox{Laplace transform rule, $\lambda=\frac{\log NR}{\epsilon}$} & \qquad & \lim_{N\rightarrow \infty} \Exs[ T_D] \le F^{-1}(R) + \epsilon\\
\mbox{Quantile rule, $t=F^{-1}(R)$} & \qquad & \prob{T_D>F^{-1}(R) + \epsilon} \le 2^{-\sqrt{N} c_1}\,, 
\end{align}
where $R=\frac{K}{N}$ is the rate, $\epsilon>0$, and $c_1$ is a constant independent of $N$ .
\end{theorem}
Proof of this theorem can be found in Section \ref{sec:proofs}. The above result establishes that computational polarization scheme achieves asymptotically optimal runtimes in expectation and with high probability under different freezing set rules. In particular, the lower-bound given in Theorem \ref{thm:infotheoryoptimality} is matched asymptotically. Importantly, it verifies that polarization takes places fast enough over non-frozen indices in a functional sense.

\subsection{Practical Determination of the Freezing Set and Error Probability}
\label{sec:practicalcode}
Although Theorem \ref{thm:runtimeguarantees} guarantees the asymptotic optimality of computational polarization under the Laplace transform and quantile rules, here we outline a practical algorithm that implements the quantile rule to determine the freezing set at finite block lengths and upper-bound the error probability. The resulting scheme enjoys the same high probability guarantee as in quantile rule. We assume that the numer of workers, i.e., block size, $N=2^n$ for some integer $n$, which is the level of polarization.

\begin{enumerate}
    \item Set $t^*=F^{-1}(R)$ where $R=\frac{K}{N}$ is the rate of the code, and $F$ is the CDF of the runtime variable.
    \item Calculate the probability values $\{F_{n,i}(t^*)\}_{i=1}^{N}$ using the CDF polarization recursion given in \eqref{eq:cdfgeneralrecursion}.
    \item Pick the subset $S$ such that the indices $i\in S$ are the smallest $RN$ values of the collection $\{F_{n,i}(t^*)\}_{i=1}^N$. Subset $S$ is used for data and the complement set $\mathcal{F}=S^c$ is \emph{frozen} by inputting zero or known fixed matrices at the corresponding indices.
\end{enumerate}
The above strategy provides a linear code with block length $N=2^n$ as in traditional Polar Codes. Here we briefly discuss some alternatives when $N$ is not a power of two. One can decompose $N$ as a sum of powers of two, i.e, $N=\sum_{i=1}^q 2^{n_i}$ and simultaneously employ computational polarization schemes with different construction sizes $n_1,\ldots,n_q$. Furthermore, it is also possible to employ a larger matrix instead of the $2\times 2$ matrix $H_2 =  \begin{bsmallmatrix}1&1\\1&-1\end{bsmallmatrix} $ in the one step computational polarization process and obtain similar recursive constructions, e.g., of size $N=3^n$ . Alternatively, we may compute $f(A_i)$ for $i\in\mathcal{F}$ at reliable worker nodes, instead of freezing virtual workers $\mathcal{F}$.

After the code is constructed, we can obtain an upper-bound on the probability of decoding failure at a deadline $t$ via the union bound
\begin{align}
\label{eq:errorbound}
\prob{\mbox{decoding failure at time }t}=\prob{T_{D}>t} \le \sum_{i\in S} F_{n,i}(t)\,.
\end{align}
Note that the probability values $\{F_{n,i}(t)\}_{i=1}^{N}$ can be calculated using the CDF polarization recursion \eqref{eq:cdfgeneralrecursion} for any value of $t$. On the other hand, Theorem \ref{thm:runtimeguarantees} guarantees that the error probability will vanish asymptotically for $t>t^*=F^{-1}(R)$ under the quantile rule.
In the next subsections, we investigate the CDF martingale process in detail to obtain a recursion for the probability densities.

\subsection{Probability Density Martingale}
In this section, we turn our attention to the probability density of the distribution characterized by the CDF $F(t)$ and the family of distributions that are generated by the computational polarization process. In order to characterize the probability density corresponding to the CDF martingale process
\begin{align}
F_{n+1}(t) = F_n(t) + \epsilon_n F_n(t)\big(1-F_n(t)\big)\,,
\label{eq:functional_update_cdf_for_pdf}
\end{align}
we first define the linear integral operator $\mathcal{L}(\cdot)$
\begin{align}
\mathcal{L}\{ p(t) \} =  \int_{t}^{\infty} p_n(u)du\,.
\end{align}
The next result shows the form of the adjoint of the above integral operator.
\begin{lemma}
\label{lem:adjoint}
The adjoint of the operator $\mathcal{L}(\cdot)$ is given by $\mathcal{L}^*(\cdot)$ 
\begin{align}
\mathcal{L}^*\{ p(t) \} =  \int_{-\infty}^{t} p_n(u)du\,.
\end{align}
\end{lemma}
\noindent Note that we have $\mathcal{L}\{ p(t) \} + \mathcal{L}^*\{ p(t) \} = \int_{-\infty}^{\infty} p(u)du = 1$.
With these definitions, the probability density process is described in terms of the operator $\mathcal{L}$ and its adjoint as follows.
\begin{lemma}
\label{lem:pdfupdate}
Suppose that the distribution characterized by the CDF $F(t)$ admits a probability density given by $p(t)=\frac{\partial}{\partial t} F(t)$. Then $p_n(t)$ satisfies
\begin{align}
p_{n+1}(t) & = 2p_n(t) 
\begin{cases} 
\mathcal{L}\{p_n(u)\} & \mbox{for } b_n=1\\
 \mathcal{L^*}\{p_n(u)\} & \mbox{for } b_n=0\,.
 \end{cases}
\end{align} 
\end{lemma}
The proofs of Lemma \ref{lem:adjoint} and \ref{lem:pdfupdate} are presented in Section \ref{sec:proofs}.
Therefore, $p_n(t), n\ge 0$ is a functional martingale process since the following relation is a consequence of the representation given by Lemma \ref{lem:pdfupdate}
\begin{align*}
\Exs[ p_{n+1}(t) \,|\, \mathcal{B}_n] = p_n(t)\,\forall t\in\real.
\end{align*}

\if FALSE
\begin{theorem}
The probability density process converges in probability. More precisely, for any $\eta>0$ we have
\begin{align}
\lim_{n\rightarrow \infty} \prob{\| p_{n+1}(t)-p_n(t) \|_{L_2}>\eta} = 0\,.
\end{align}
\end{theorem}
\begin{theorem} %
The probability density process converges almost surely in $L_p$, for any $p\in(1,\infty)$. More precisely, we have
\begin{align}
 \prob{\lim_{n\rightarrow \infty} \| p_{n+1}(t)-p_n(t) \|_{L_p}=0} = 1\,.
\end{align}
\end{theorem}
\fi

\if 0
\begin{align*}
p_{n+1}(t) = 
\begin{cases}
 2p_n(t)\int_{0}^t p_n(u)du  &\mbox{ for } \epsilon_n=+1  \\
 2p_n(t)\int_{t}^1 p_n(u)du &\mbox{ for } \epsilon_n=-1
\end{cases}
\end{align*}
\begin{align*}
g_{n+1}(t) = 
\begin{cases}
 1+g_n(t)+\log\int_{0}^t e^{g_n(u)}du  &\mbox{ for } \epsilon_n=+1  \\
 1+g_n(t)+\log\int_{t}^1 e^{g_n(u)}du &\mbox{ for } \epsilon_n=-1
\end{cases}
\end{align*}
\fi
We note that a simpler representation of the probability density martingale follows from Lemma \ref{lem:pdfupdate} and is given by
\begin{align*}
p_{n+1}(t) = 2p_n(t)\int_{tb_n}^{t+b_n(1-t)} p_n(u)du
\end{align*}
Moreover, it is worth noting that the expected ratio of the densities equals one, i.e.,
\begin{align*}
\Exs \, \Big[\frac{p_{n+1}(t)}{p_n(t)} 
\Big]= 1\, \forall t\,.
\end{align*}
In other words, the expected relative density $\frac{\partial F_{n+1}(t)}{\partial F_n(t)}$ is preserved, in a similar manner to the mutual information conservation stated in equation \eqref{eq:polarconservation}.

\begin{figure}[t!]
\centering
\includegraphics[width=16cm]{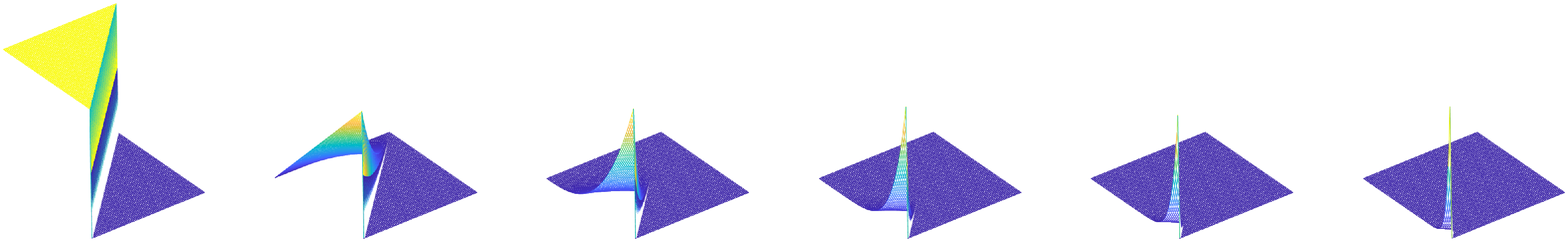}
\includegraphics[width=16cm]{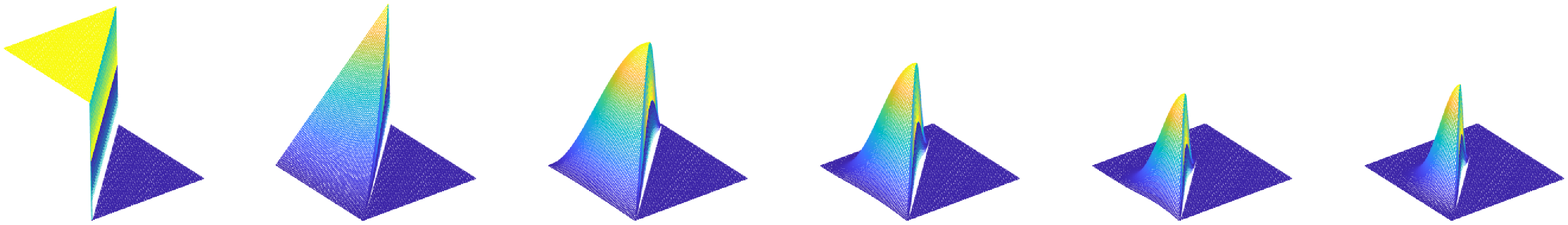}
\includegraphics[width=16cm]{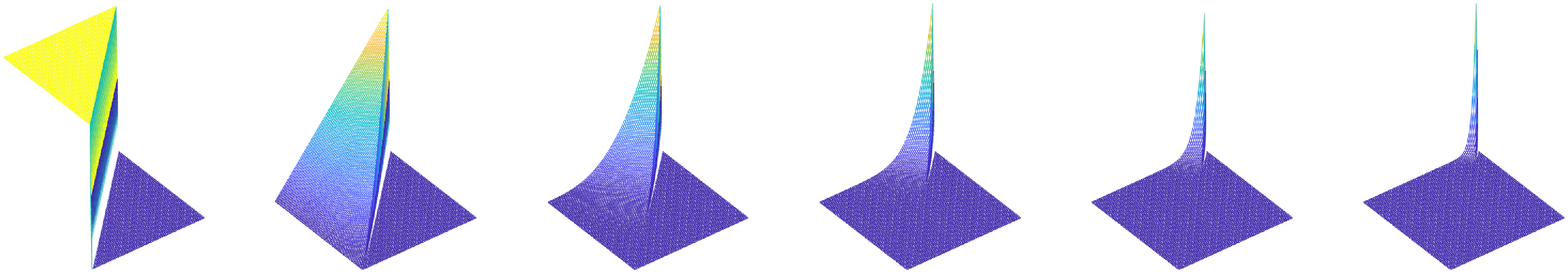}
\caption{Computational polarization of the joint density for the uniform distribution. From left to right, we present sample paths of the joint density $f_n(u,v)$ for $n=0,1,3,4,5$. The top row corresponds to $\epsilon_n=0\,\forall n$, the middle row corresponds to i.i.d. random $\epsilon_n$ uniform on $\{-1,+1\}$, and the bottom row corresponds to $\epsilon_n=1\,\forall n$ \label{fig:jointdensity}.}
\end{figure}
Next, we provide a complementary result on the joint probability densities corresponding to the computational polarization functional process. This result can be seen as a generalization of Lemma \ref{lem:pdfupdate}, where only one-dimensional distributions were considered. As we explore next, it is possible to characterize the evolution of two-dimensional joint densities of the synthesized runtime distributions.
\begin{theorem}
\label{thm:joint}
Suppose that the runtime distribution $F$ admits a continuous probability density $f(t)$. Then the computational polarization functional process obeys the following joint density update equation
\begin{align*}
f_{n+1} (u,v) = 
\begin{cases} 
2f_n^{+}(u) f_n^{+}(v) 1[u\le v] & \mbox{ for } b_n=+1\\
2f_n^{-}(u) f_n^{-}(v) 1[u\le v] & \mbox{ for } b_n=-1\,.\\
\end{cases}
\end{align*}
where $1[u\le v]$ is the zero-one valued indicator of the event $u\le v$ and 
\begin{align*}
f_n^{+}(u):&=\int_{-\infty}^{\infty} f(u,v^\prime)dv^\prime\\
f_n^{-}(v):&=\int_{-\infty}^{\infty} f(u^\prime,v)du^\prime\,,
\end{align*}
are marginalized univariate densities over the first and second arguments respectively.
\end{theorem}
It is noteworthy that the equations take a simpler and symmetric form compared to Lemma \ref{lem:pdfupdate} when the evolution of the joint distributions are considered. Moreover, the above result enables numerical simulation of the joint density evolution and provides informative visualizations. A numerical example is provided in Figure \ref{fig:jointdensity}, where it can be observed that the joint densities approach to two-dimensional Dirac distributions. Another interesting aspect of this simulation is that joint densities visually resemble bivariate normal distributions restricted to the domain $u\le v$ in the middle row of Figure \ref{fig:jointdensity}. The proof of the theorem can be found in Section \ref{sec:proofs}.
\subsection{Random Variables with Discrete Probability Mass}
For discrete probability measures we note that the joint density updates can also be stated in a simple form using finite dimensional matrix algebra as follows
\begin{align*}
P_{n+1} = 
\begin{cases}
L\left(2P_n11^T P_n^T - \diag(P_n11^T P_n^T)\right) & \mbox{ for } b_n = 1\\
L\left(2P_n^T11^T P_n - \diag(P_n^T11^T P_n)\right) & \mbox{ for } b_n = 0
\end{cases}
\end{align*}
where $L(\cdot)$ returns the lower triangular part of a matrix, and $U(\cdot)$ returns the upper triangular part of a matrix. Note that $P_{n+1}$ is always lower triangular.

\subsection{Non-identical Runtime Distributions}
\label{sec:non-identical}
We now revisit the basic computational polarization operation when runtime distributions are not identically distributed.  We now model the computation times of worker nodes as real valued random variables $T^{(1)},\cdots, T^{(N)} \in \mathbb{R}$ which are distributed independently according to cumulative density functions $F^{(1)}(t),\cdots, F^{(N)}(t)$ respectively. More precisely, suppose that
\begin{align*}
\prob{ T^{(k)} \le t} = F^{(k)}(t)\, \mbox{ for $k=1,\cdots,N$}.
\end{align*}
Next, suppose that we apply the one step polarization operation introduced in Section \ref{sec:onesteppolarization}. It is clear that this operation creates two virtual workers whose runtimes are given by $T^{-}$ and $T^{+}$ as follows
\begin{align*}
T^{+}&= \max( T^{(1)}, T^{(2)} ) \\
T^{-}&= \min( T^{(1)}, T^{(2)} )\,.
\end{align*}
However, additional care has to be taken in calculating the distributions of $T^{-}$ and $T^{+}$ due to the violation of the i.i.d. assumption. In particular, the conclusion of the derivations done in Section \ref{sec:onesteppolarization} that leads to the tree process illustrated in Figures \ref{fig:CDFtreetwo} and \ref{fig:CDFtreefour} is no longer valid. Nevertheless, we observe the following identities
\begin{align*}
\prob{\max( T^{(1)}, T^{(2)} ) \le t} &= \prob{ T^{(1)} \le t, T^{(2)}  \le t} \\
& = \prob{ T^{(1)}\le t } \prob{ T^{(2)}\le t}\\
& = F^{(1)}(t)F^{(2)}(t)\,,
\end{align*}
where we have leveraged the independence of $T^{(1)}$ and $T^{(2)}$ in the second equality. Next, we carry out the calculation for the distribution of $T^{-}=\min( T^{(1)}, T^{(2)} )$ in an analogous manner
\begin{align*}
\prob{\min( T^{(1)}, T^{(2)} ) \le t} &= 1-\prob{\min( T^{(1)}, T^{(2)} ) > t} \\
& = 1-\prob{ T^{(1)} > t,\, T^{(2)} > t} \\
& = 1-\prob{ T^{(1)} > t}\prob{ T^{(2)} > t} \\
& = 1-\big(1-F^{(1)}(t)\big)\big(1-F^{(2)}(t)\big)\\
& = F^{(1)}(t)+F^{(2)}(t)-F^{(1)}(t)F^{(2)}(t)\,.%
\end{align*}
where we have leveraged the independence of $T^{(1)}$ and $T^{(2)}$ in the third equality. 

We provide a visual depiction of the above results in Figure \ref{fig:CDFnonuniform}, where it can be seen that the average of two CDFs is conserved. More specifically, from the above identities we verify that
\begin{align*}
 \frac{1}{2}\big(F^{(1)}(t)+F^{(2)}(t)\big)=F^+(t) + F^-(t)\,,
 \end{align*}
 where $F^{+}$ and $F^{-}$ are the CDFs corresponding to $T^+$ and $T^{-}$ respectively. We remark that this conclusion parallels the results on the polarization of non-stationary memoryless channels studied in \cite{alsan2016simple}, where an average mutual information conservation rule holds. We leave establishing convergence results and quantifying the rate of convergence for future work.

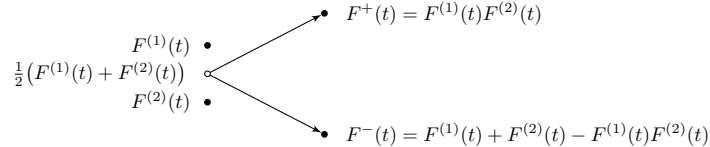
\begin{figure}
\centering
\tikzstyle{POINT} = [draw, circle,fill,scale=0.3]
\tikzstyle{POINTH} = [draw, circle,scale=0.3]

\begin{tikzpicture}[node distance=\nodedistance,auto,>=latex', scale = 0.75, transform shape]
    \tikzstyle{line}=[draw, -latex']
    \node [POINTH] (P1) {};
    \node [POINT, above=0.4 of P1] (P1a) {};
    \node [POINT, below=0.4 of P1] (P1b) {};
    \node [POINT, above right=1 and 2 of P1] (P21) {};
    \node [POINT, below right=1 and 2 of P1] (P22) {};
    \node [left =0.2 of P1] {$\frac{1}{2}\big(F^{(1)}(t)+F^{(2)}(t)\big)$};
    \node [above left=0.2 of P1] {$F^{(1)}(t)$};
    \node [below left=0.2 of P1] {$F^{(2)}(t)$};
    \node [right =0.2 of P21] {$F^{+}(t)=F^{(1)}(t)F^{(2)}(t)$};
    \node [right =0.2 of P22] {$F^{-}(t)=F^{(1)}(t)+F^{(2)}(t)-F^{(1)}(t)F^{(2)}(t)$};
    \path[line] (P1) edge (P21)
    			(P1) edge (P22);
\end{tikzpicture}
\caption{Basic polarization process for non-identical runtime distributions. Note that the average value of the polarized CDFs is conserved as in Figure \ref{fig:CDFtreetwo}. \label{fig:CDFnonuniform}} 
\end{figure}

\section{Convergence of the functional process}
In this section, we present results on the convergence behavior of the computational polarization process as a functional martingale. Let us consider the functional process
\begin{align}
F_{n+1}(t) = F_n(t) + \epsilon_n F_n(t)\big(1-F(t)\big)\,,
\label{eq:functional_update}
\end{align}
where $F(t)$ is the cumulative density function of the random variable $T$, i.e., $\prob{T>t}=F(t)$. 

It is important to note that our results are stronger than existing results in the polarization literature since we prove convergence of the entire function in a vector space sense. In contrast, existing results only provide convergence statements for a single, i.e., fixed value of the time $t$, where the evolution of $F_n(t)$ can be described via a scalar martingale. We also point out interesting and crucial differences between pointwise, uniform and norm convergence behavior of the functional process, which can be dramatically different.

Now we present our first convergence result on the functional process.
\begin{theorem}
Suppose that $0\le F(t)\le 1$ and $\int F(t) dt <\infty$. Then, the functional process \eqref{eq:functional_update} converges in the $L_2$ norm
\begin{align*}
\lim_{n\rightarrow \infty} \, \Exs \, \| F_{n+1}(t) - F_n(t) \|^2_{L_2} = 0\,,
\end{align*}
and for any $\epsilon>0$ it holds that
\begin{align*}
\lim_{n\rightarrow \infty}  \, \prob{\| F_{n+1}(t) - F_n(t) \|^2_{L_2} \ge \epsilon } = 0\,.
\end{align*}
\label{thm:functionalconvergencel2}
\end{theorem}
The proof of Theorem \ref{thm:functionalconvergencel2} follows from a direct analysis of the quadratic variation process, and can be found in Section \ref{sec:proofs}.
\subsection{Pointwise and Uniform Convergence}
We start by analyzing the pointwise convergence of the process
\begin{align}
F_{n+1}(t) = F_n(t) + \epsilon_n F_n(t)\big(1-F_n(t)\big),\, \forall t\in \real\,,
\label{eq:functional_update_pointwise}
\end{align}
for every fixed value of $t\in \real$. We note that the above process for a fixed value of $t$ is identical to the erasure channel process in Polar Codes in $GF(2)$.
\begin{align}
\epsilon_{n+1} = \begin{cases} 2\epsilon_n-\epsilon_n^2 & \mbox{ with probability } \frac{1}{2}  \\ \epsilon_n^2 & \mbox{ with probability } \frac{1}{2} \end{cases}\,, \label{eq:erasureprocess}
\end{align}
where the erasure probability $\epsilon_n=F_n(t)$. Therefore, when the CDF $F(t)$ is invertible, the process defined in \eqref{eq:erasureprocess} is identical to fixing the time variable as
\begin{align}
  t=F^{-1}(\epsilon)\, \label{eq:processinverset} 
\end{align}
and calculating the update according to \eqref{eq:functional_update_pointwise}.

The process \eqref{eq:erasureprocess} plays an important role in the proof of capacity achieving properties of Polar Codes for other discrete symmetric channels. We next leverage the following fundamental result on the convergence of the erasure process which is due to Arikan.
\begin{theorem}
\label{thm:arikanbec}
(\cite{arikan2009channel}) Suppose that $\epsilon_0=\epsilon$, then $\Exs \epsilon_n = \epsilon$ for all $n\in \mathbb{N}$ and $\epsilon_n$ converges almost surely to a random variable $\epsilon_{\infty}$ such that $\Exs [\epsilon_{\infty}]=\epsilon$. Furthermore, the limiting random variable $\epsilon_{\infty}$ equals $0$ or $1$ almost surely, i.e., $\prob{\epsilon_\infty(1-\epsilon_\infty)=0}=1$\,.
\label{thm:arikan}
\end{theorem}
In light of the relation \eqref{eq:processinverset}, Arikan's result applies to a fixed value of $t$, and implies that $F_n(t)$ converges to random variable that equals $0$ or $1$ almost surely. This is a pointwise result on the functional process $F_n(t)$. In the next subsection, we strengthen it to uniform convergence for functions defined over compact metric spaces.

\subsubsection{Uniform convergence}
The convergence result guaranteed by Theorem \eqref{thm:arikan} only guarantees pointwise convergence of the functional process \eqref{eq:functional_update_pointwise} for a fixed value of $t$. We now present a uniform convergence result, which holds in a functional sense, i.e., for all values of $t$, assuming that the domain of the function is compact.
\begin{theorem}
Suppose that the domain of $F(t)$, which denoted by $\mathcal{T}$, is compact. For any $a>0$ and $b<1$, it holds that
\begin{align*}
\lim_{n\rightarrow \infty} \sup_{t\in \mathcal{T}} \prob{F_n(t)\in [a,b])} = 0\,.
\end{align*}
\label{thm:functionaluniform}
\end{theorem}
The proof of Theorem \ref{thm:functionaluniform} can be found in Section \ref{sec:proofs}.

\subsubsection{Counter-example to uniform norm convergence}
An natural question is whether the functional convergence shown in Theorem \ref{thm:functionalconvergencel2} in $L_2$ norm can be extended to uniform norm, in a similar spirit to the uniform convergence of Theorem \ref{thm:functionaluniform}. However, in the next result, we demonstrate that the process \eqref{eq:functional_update_pointwise} does not converge in the uniform norm.
\begin{lemma}
Suppose that the CDF $F(t)$ of the random variable $T$ is continuous. Then we have
\begin{align}
\sup_{t\in \mathcal{T}} \vert F_{n+1}(t) - F_n(t) \vert = \frac{1}{4}\,,
\end{align}
with probability one and the above maximum is achieved at $t=F^{-1}(\frac{1}{2})=\median(T)$.
\end{lemma}
\begin{proof}
We use the recursive definition of the process to obtain
\begin{align*}
\sup_{t\in \mathcal{T}} \vert F_{n+1}(t) - F_n(t) \vert &= \sup_{t\in \mathcal{T}}  \vert F_n(t)\big( 1-F_n(t)\big) \vert\\
&= \max_{f \in \range(F_n)} f(1-f)\\
&= \frac{1}{4} - \min_{f \in \range(F_n)} \Big(f-\frac{1}{2}\Big)^2
\end{align*}
Since the CDF $F$ and hence $F_n$ is continuous, $\range(F)=\range(F_n)=[0,1]$ we have
\begin{align*}
\min_{f \in \range(F_n)} \Big(f-\frac{1}{2}\Big)^2 = 0\,,
\end{align*}
completing the proof. Note that for discontinuous CDFs, we have
\begin{align*}
\lim_{n\rightarrow \infty} \min_{f \in \range(F_n)} \Big(f-\frac{1}{2}\Big)^2 = 0\,,
\end{align*}
which implies that
\begin{align*}
\lim_{n\rightarrow \infty}
\sup_{t\in \mathcal{T}}  \vert F_{n+1}(t) - F_n(t) \vert = \frac{1}{4}\,.
\end{align*}

\end{proof}

\section{Convergence rate of the functional process}
We now consider the aforementioned functional random process
\begin{align}
F_{n+1}(t) = F_n(t) + \epsilon_n F_n(t)\Big(1-F_n(t)\Big)\,,
\label{eq:fncupdate}
\end{align}
and present bounds on the convergence rate of $F_{n+1}(t)-F_n(t)$. Next theorem applies to any function $F(t)$ and shows that $\|F_{n+1}-F_n(t)\|_{L_\beta}$ converges exponentially fast to zero with high probability for any $\beta \in (0,\frac{1}{2}]$.
\begin{theorem}
\label{lem:integratedbound}
Suppose that $F(t)$ is any function, and let $F_n(t)$ be defined via the functional update in \eqref{eq:fncupdate}. For any $\beta \in (0,\frac{1}{2}]$, $\rho \in (\frac{3}{4},1)$ it holds that
\begin{align}
\prob{ \|F_{n+1}(t)-F_n(t)\|_{L_\beta}^\beta > \rho^n} \le \Big(\frac{3}{4\rho}\Big)^{\beta n}\, \int_{-\infty}^{+\infty}
 \big(F(t)(1-F(t))\big)^{\powerp} dt\,.
\end{align}
\label{thm:convergencerate1}
\end{theorem}
\begin{remark}
Note that the term $\int
 \big(F(t)(1-F(t))\big)^{\powerp} dt$ on the right-hand-side depends on the distribution of the random variable $\mathcal{T}$ via its CDF $F(t)$. For instance, for the uniform distribution $U[0,1]$, i.e., $F(t)=t$ for $t\in [0,1]$, we can pick the half norm, $L_{\scriptsize 1/2}$, i.e., $\beta=\frac{1}{2}$, and $\rho = \frac{3}{8}$ to conclude that the process converges exponentially fast, except with exponentially small probability. More precisely,

\begin{align*}
\prob{ \|F_{n+1}(t)-F_n(t)\|_{L_{\scriptsize 1/2}}^{\scriptsize 1/2} > \Big(\frac{3}{8}\Big)^n} &\le \Big(\frac{1}{2}\Big)^{\frac{n}{2}}\, \int_{0}^1 \sqrt{t(1-t)} dt\\
&=\Big(\frac{1}{\sqrt{2}}\Big)^{n}\, \frac{\pi}{8}\,.
\end{align*}
We also remark that, although Theorem \ref{thm:convergencerate1} is applicable to any CDF $F(t)$ generically without any additional assumptions, the exponential convergence rate $\big (\frac{3}{4\rho}\big)^{\beta n}$ can be further improved. Unfortunately, the proof technique does not allow the convergence analysis for any $\beta$ greater than $\frac{1}{2}$, which limits its applicability to the more common $L_p$ norms for $p\in (\frac{1}{2},\infty)$.  
 
\end{remark}
The proof of Theorem \ref{thm:convergencerate1} can be found in Section \ref{sec:proofs}.

\subsection{Non-asymptotic Analysis of the Convergence Rate with Optimal Exponent}
Now we present a non-asymptotic convergence result for bounded distributions that attains the optimal rate of convergence for finite values $n$. Our results parallel the existing results on Polar Codes (see e.g. \cite{arikan2009rate,hassani2014finite}) and extend them to the more general function space setting. In particular, we show that $\|F_{n+1}(t) - F_n(t)\|_{L_p}$ converges to zero with rate $2^{-2^{n/2-O(\log n)}}$ with high probability. We note that the exponent $O(2^{n/2})=O(\sqrt{N})$ matches the  exponent of the erasure process, and the mutual information process in polar codes over discrete memoryless channels, and is known to be optimal for Polar Codes \cite{hassani2014finite}. Since the pointwise convergence behavior is identical to the erasure process, the rate of convergence provided here can not be improved in the functional case.

\begin{theorem}
For every $t \in \reals$, $\eta>0$, $\beta>0$, $\rho\in(\frac{3}{4},1)$ and $n \ge \frac{2}{\log(\frac{4}{3\rho})}$ it holds that
\begin{align}
\min(F_n(t),1-F_n(t)) \le 2^{-2^{(n-\beta\log(n)(\frac{1}{2}-\eta)}}\,.\,\,
\end{align}
with probability at least
\begin{align*}
1-\frac{n^{-\frac{\beta}{2}\log(\frac{1}{\rho})}}{\sqrt{\rho}\big(1- \sqrt{\rho}\big)} - 2^{-(n-\beta\log(n))(1-\mathcal{H}(\frac{1}{2}-\eta))}\,.
\end{align*}
Furthermore, suppose that $F(t)=0$ for $t\le a$ and $F(t)=1$ for $t\ge b$, then we also have
\begin{align}
\| F_{n+1}(t)-F_n(t) \|_{L_p} \le \vert b-a \vert \, 2^{-2^{(n-\beta\log(n)(\frac{1}{2}-\eta)}}\,,
\end{align}
with the same probability.
\label{thm:non_asymptotic}
\end{theorem}

\section{Empirical Illustration of Computational Polarization}
\label{sec:empirical}
\begin{figure}[!t]
\begin{minipage}[b]{0.23\linewidth}
  \centering
  \centerline{\includegraphics[width=4.5cm,trim=0 100 0 100,clip]{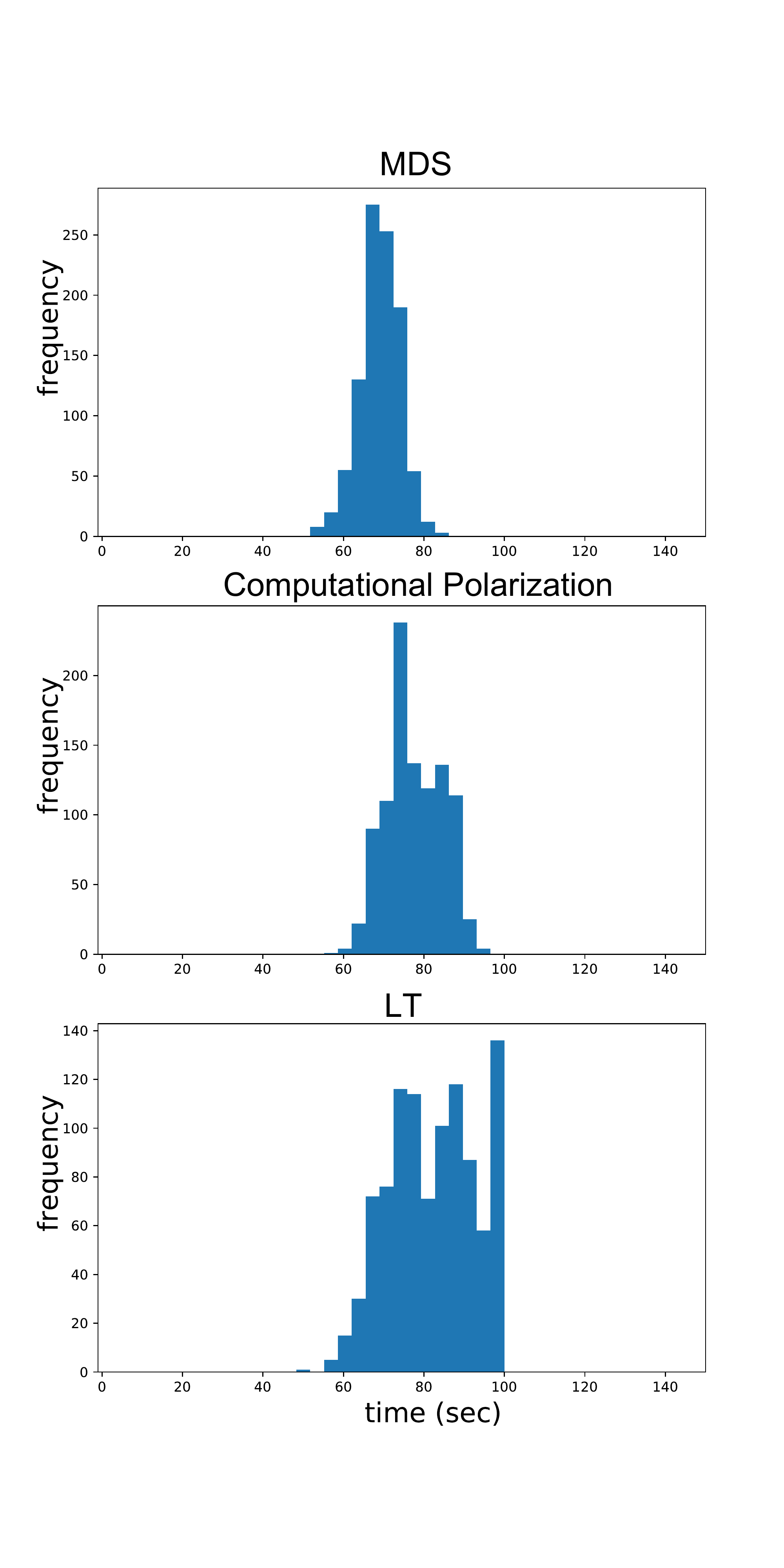}}
  \centerline{(a) Uniform $N=64$}\medskip 
\end{minipage}
\begin{minipage}[b]{0.23\linewidth}
  \centering
  \centerline{\includegraphics[width=4.5cm,trim=0 100 0 100,clip]{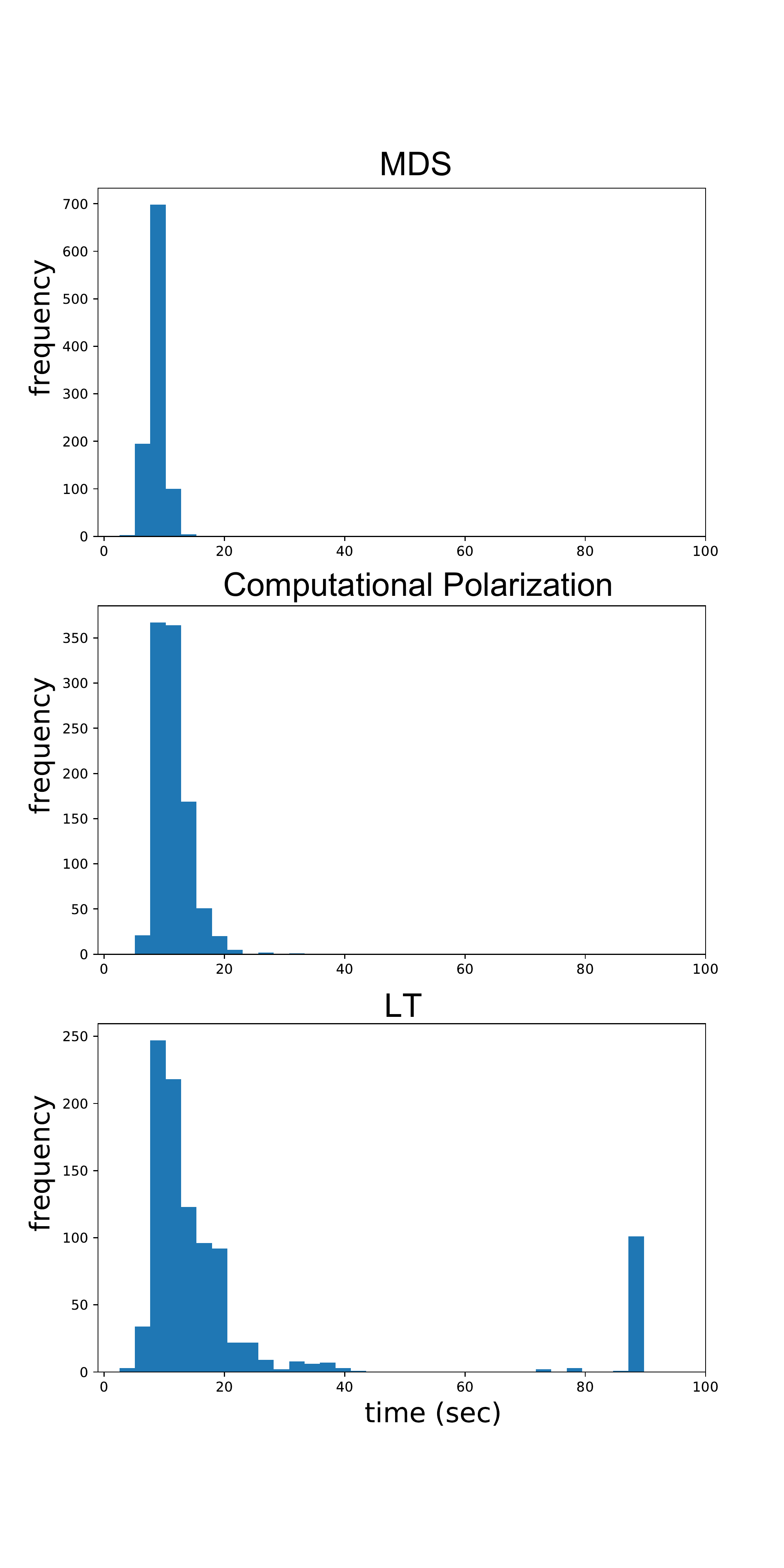}}
  \centerline{(b) Exponential $N=64$}\medskip
\end{minipage}
\hfill
\begin{minipage}[b]{0.23\linewidth}
  \centering
  \centerline{\includegraphics[width=4.5cm,trim=0 100 0 100,clip]{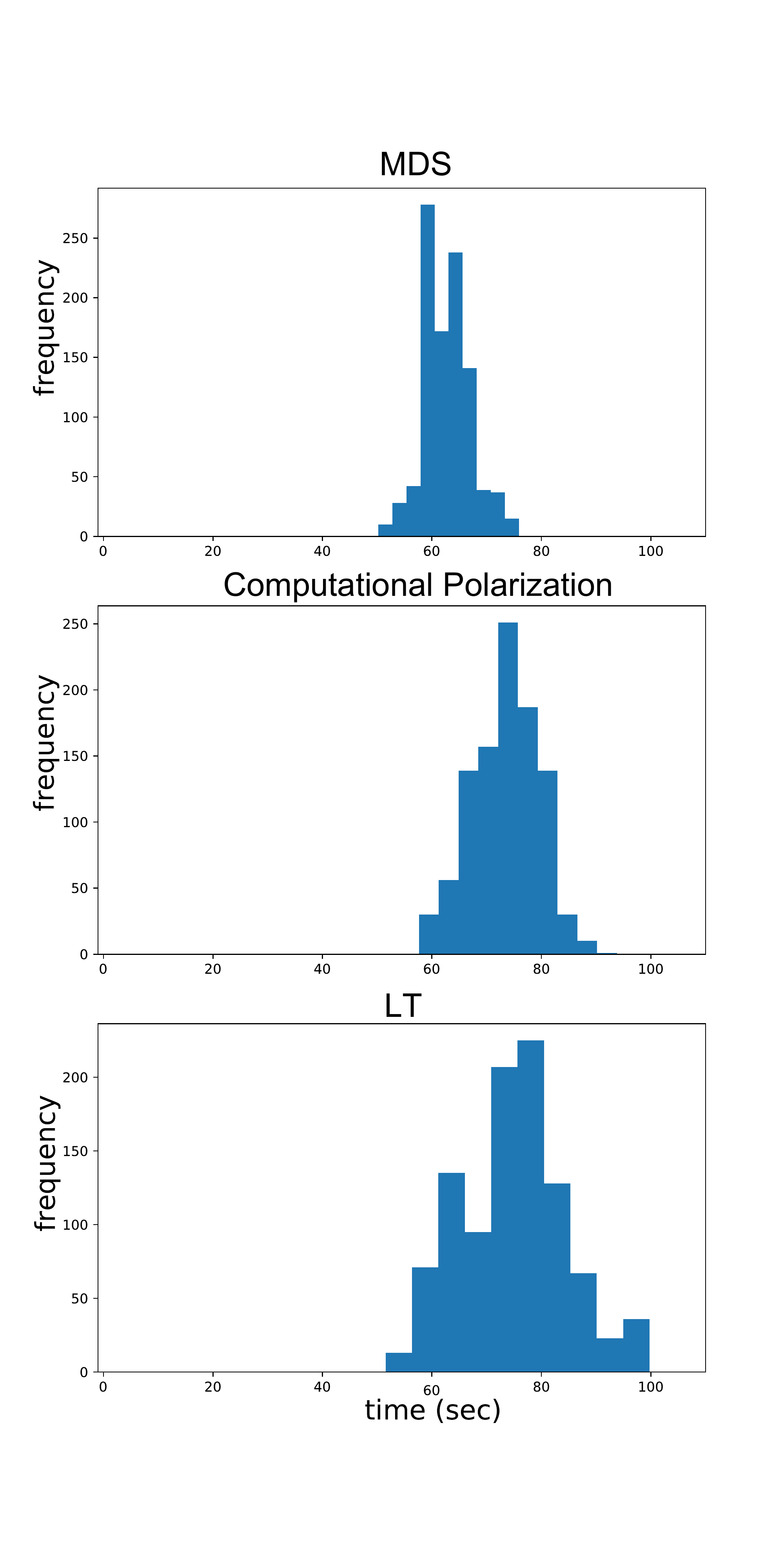}}
  \centerline{(c) Uniform $N=128$}\medskip 
\end{minipage}
\begin{minipage}[b]{0.23\linewidth}
  \centering
  \centerline{\includegraphics[width=4.5cm,trim=0 100 0 100,clip]{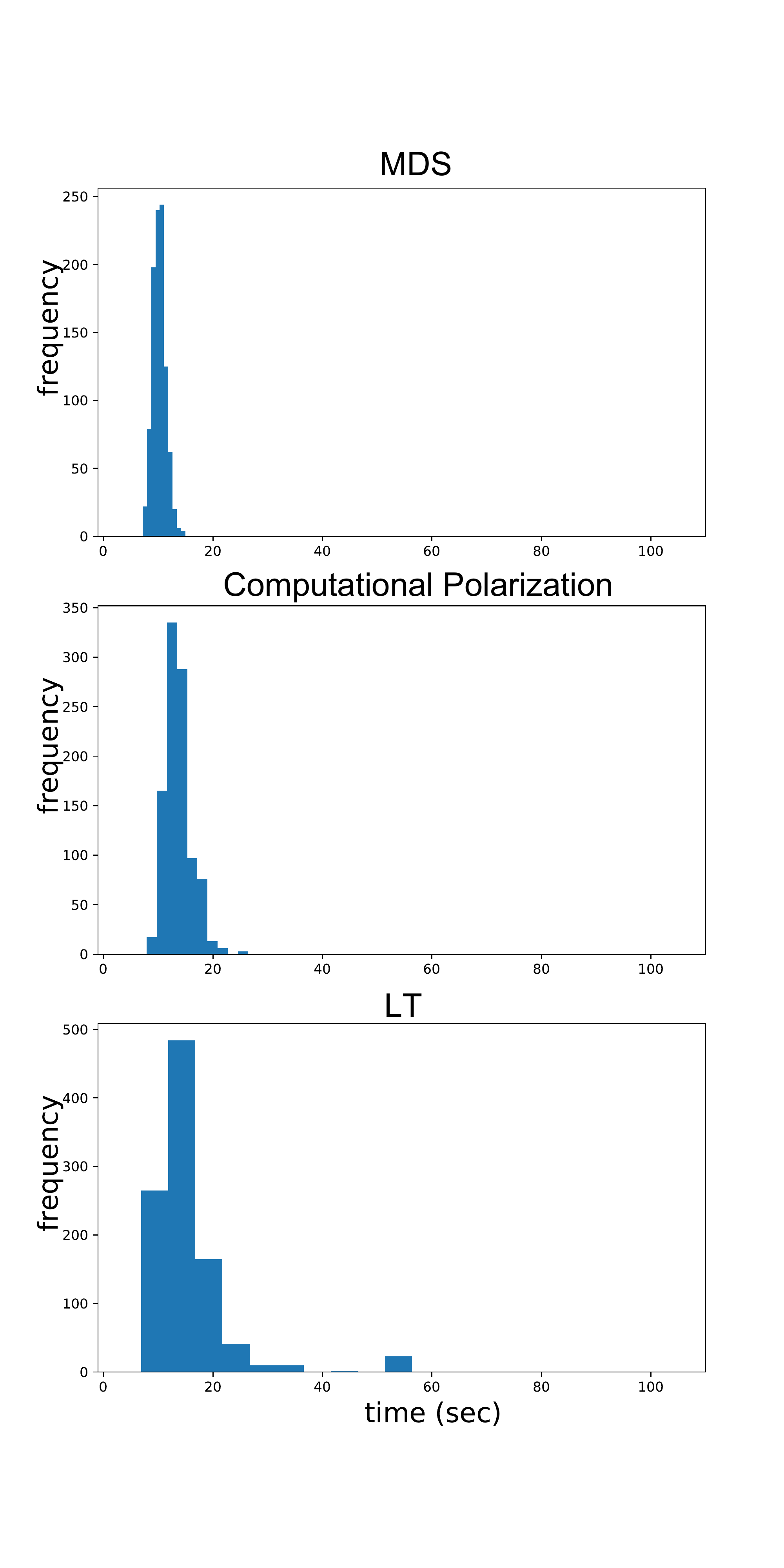}}
  \centerline{(d) Exponential $N=128$}\medskip
\end{minipage}
\caption{Comparison of runtime distributions for MDS codes, computational polarization and LT codes, where the base runtime distribution is Uniform[0,100] in (a) and (c), and exponential with scale parameter 10 in (b) and (d). \label{fig:decodetime}}
\end{figure}

\subsection{Empirical Runtime Distributions}
In this section we use Monte Carlo simulations to obtain the empirical distribution of polarized computation times. Specifically, we repeat the recursive polarization procedure using independently realizations of the input variables, and repeat the procedure using $500$ independent trials. We display the histograms of the random variables generated by the computational polarization process in Figures \ref{fig:4layerpdf} and \ref{fig:4layerpdfexp} for Uniform$[0,1]$ and Exponential$(0.5)$ base runtime distributions.

\begin{figure}[t!]
\centering
\includegraphics[width=13cm]{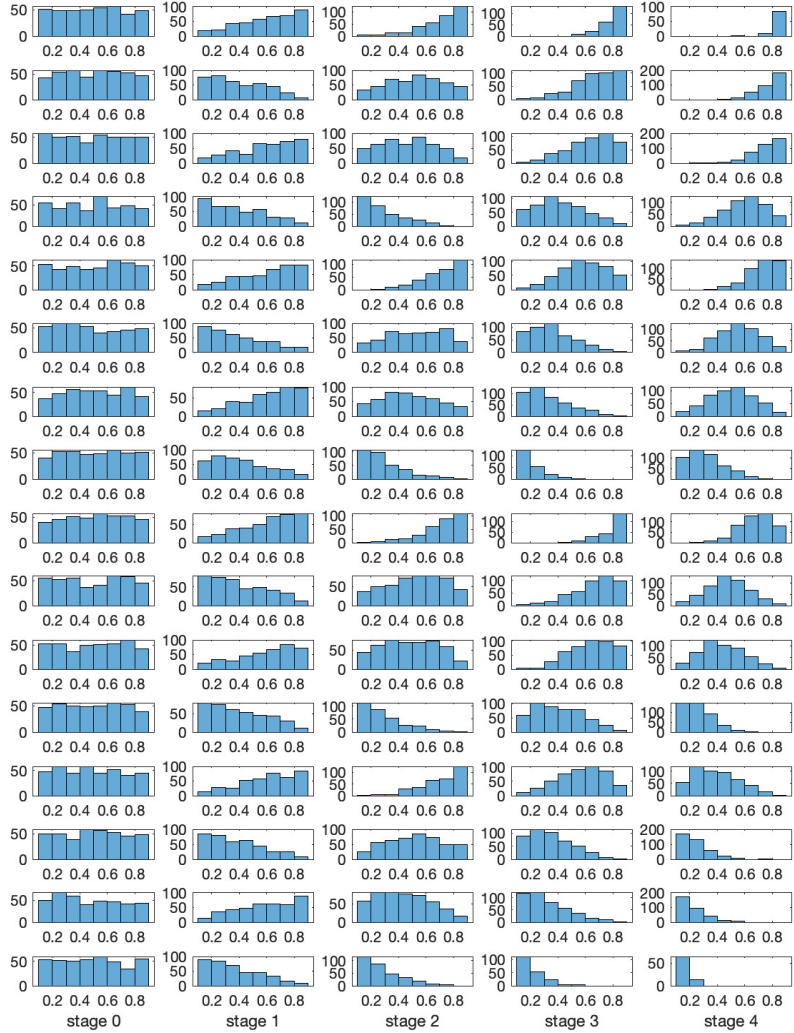}
\caption{Histogram of polarized computation times in 4 layer polarization (i.e., N=$2^4$) with Uniform$[0,1]$ i.i.d. variables. Here the rows correspond to the runtime distributions of the worker nodes, and columns correspond to different stages (layers) in the computational polarization process. Specifically, the first layer is the histogram of i.i.d. input random variables $T^{(1)},\ldots,T^{(N)}$, and the following layers show the recursive application of the one step $\min\max$ transform.\label{fig:4layerpdf}}
\end{figure}

\begin{figure}[t!]
\centering
\includegraphics[width=13cm]{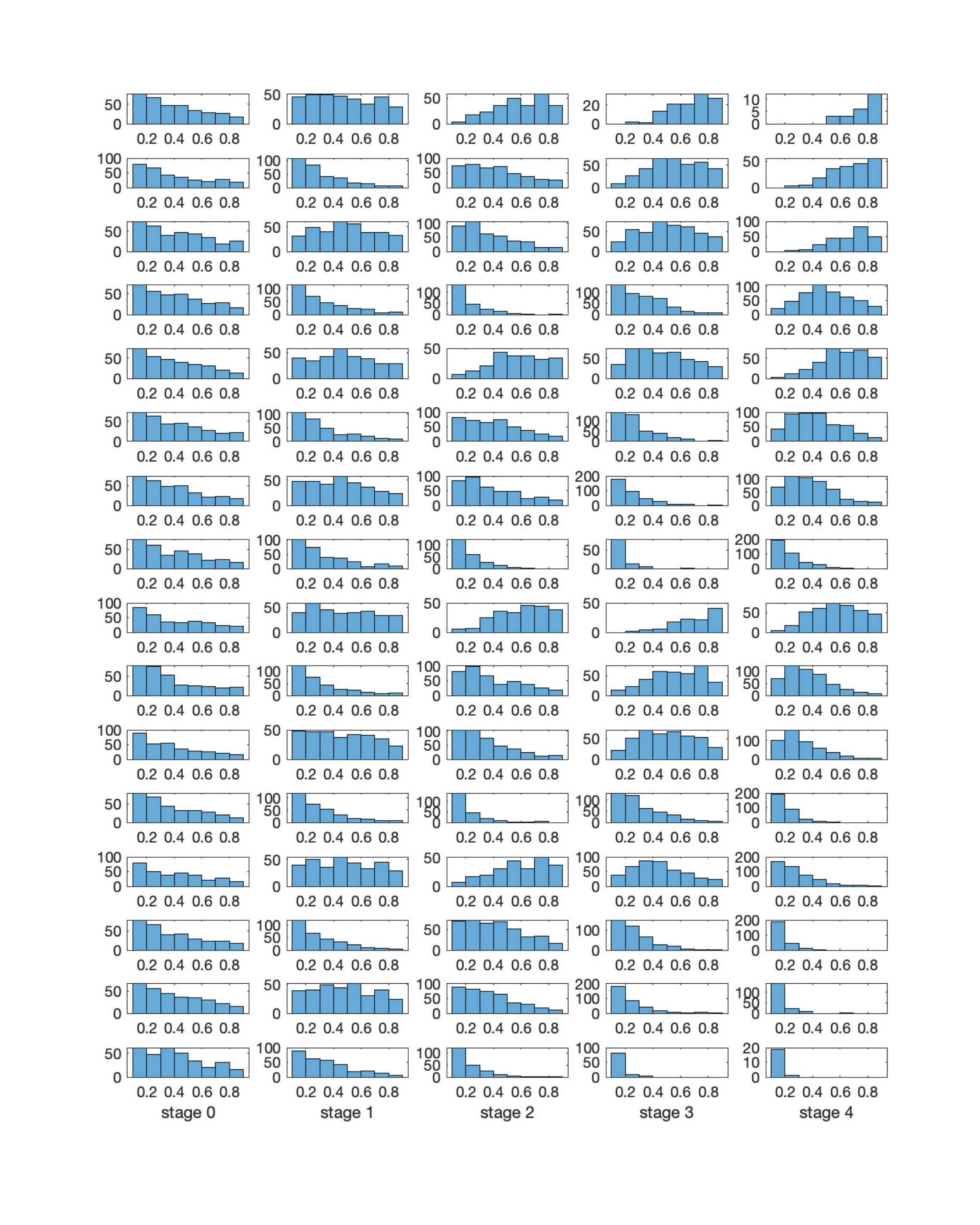}
\caption{Histogram of polarized computation times in 4 layer polarization (i.e., N=$2^4$) with Exponential$(0.5)$ i.i.d. variables. Column/Row Layout is the same as Figure Figure \ref{fig:4layerpdf}.\label{fig:4layerpdfexp}}
\end{figure}

\begin{table}[t!]
  \begin{varwidth}[b]{0.6\linewidth}
    \centering
    \begin{tabular}{ l r r r r }
    \toprule
     N & MDS encoding & MDS decoding& CP encoding& CP decoding\\ 
    \midrule
     64 & 5.3s & 6.3s  & 0.3s & 0.3s\\ 
     128 & 12s & 15s  & 2.3s & 5.3s\\
     256 & 28s & 34s & 6.5s & 1.6s \\
     512 & 73s & 75s & 24s & 3.2s  \\
    \bottomrule
    \,\\
    \,\\
    \,\\
    \end{tabular}
    \caption{Encoding and decoding times for MDS (Reed-Solomon)\\ and CP (Computational Polarization) in seconds. \label{table:encdectimes}}
  \end{varwidth}%
  \hfill
  \begin{minipage}[b]{0.5\linewidth}
    \centering
    \includegraphics[width=1\linewidth]{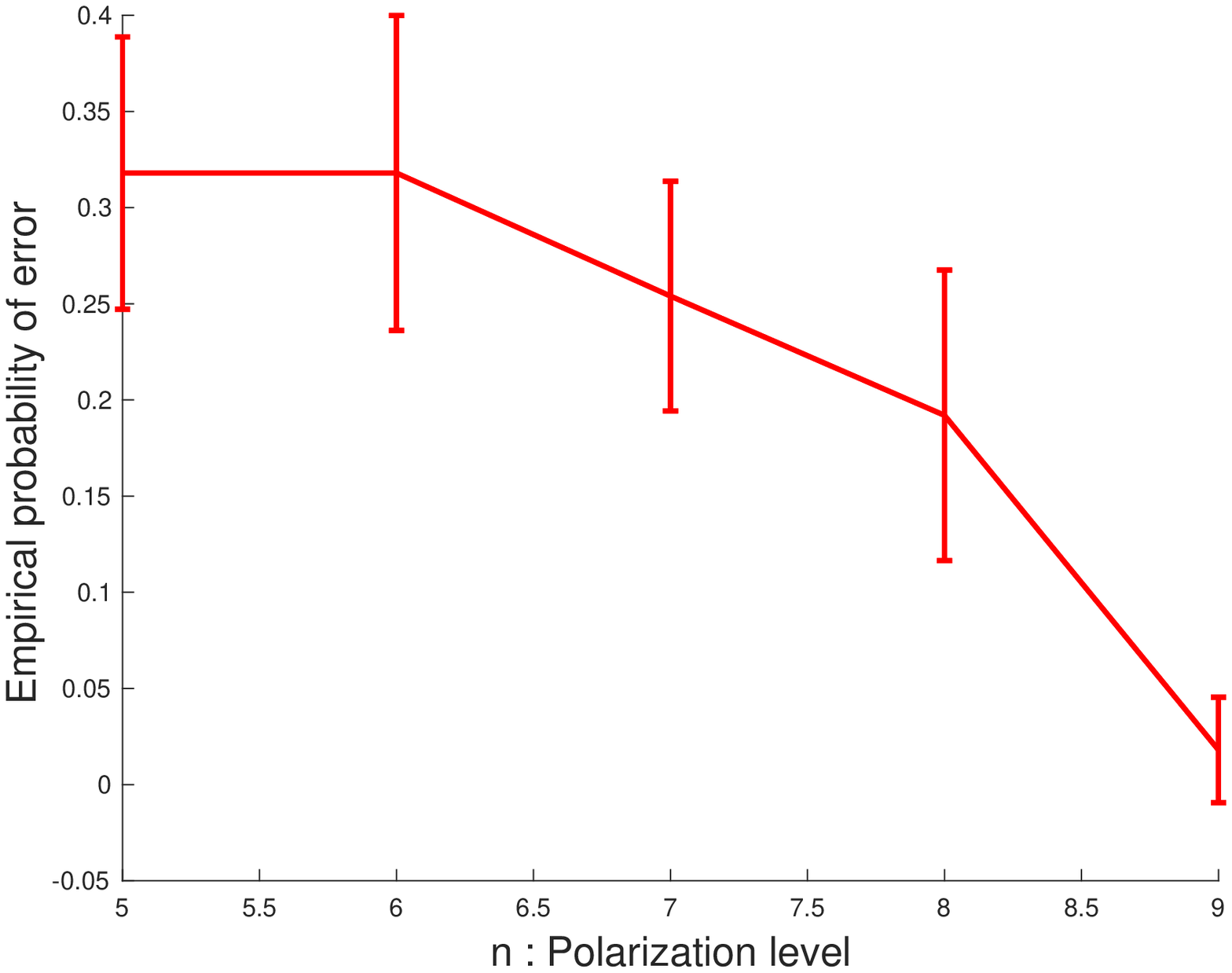}
    \captionof{figure}{Empirical error probability simulated\\ under i.i.d. Uniform[0,1] runtime distributions. \label{fig:errorbar}}
  \end{minipage}
\end{table}

\subsection{Runtime Comparison with MDS and LT Codes}
\label{subsec:runtimecomp}
In Figure \ref{fig:decodetime} we compare the runtimes of the MDS coded computation \cite{Lee2018}, LT codes \cite{mallick2019rateless}, and computational polarization. We employ the quantile freezing rule outlined in Section \ref{sec:practicalcode}. In this experiment we simulate Uniform$[20,100]$ and Exp(10) i.i.d. runtime distributions for all the workers. We set the rate to $0.625$ and consider $N=64$ and $N=128$. It can be seen that MDS coded computations provide the best overall computation time at the expense of solving large linear systems. Computational polarization outperforms LT codes while maintaining near linear time encoding and decoding. Moreover, the proposed computational polarization scheme provides performance relatively close to MDS codes as predicted by Theorem \ref{thm:runtimeguarantees}.

In Table \ref{table:encdectimes}, we compare the encoding and decoding times of Reed-Solomon (MDS) codes and Computational Polarization (CP). Specifically, we generate a random matrix of size $100N\times 5000$ distributed to $N$ workers by partitioning over the rows. We consider the matrix-vector multiplication task $Ax$ where $x$ is a length $5000$ randomly generated vector. We employ Fermat Number Transform (FNT) for faster encoding and decoding of Reed-Solomon codes as decribed in \cite{soro2010fnt}. Although FNT based Reed-Solomon codes are faster than the straightforward implementation, requiring $O(N\log N)$ and $O(N\log^2 N)$ time for encoding and decoding respectively, it can be seen that this method takes significantly larger encoding and decoding times compared to the Computational Polarization approach.

\subsection{Probability of Failure Given a Deadline}
Here we provide a numerical simulations of the error probability given a deadline for the computation. We assume that the runtime distributions are i.i.d. Uniform$[0,1]$ and the rate is $R=1/2$, implying that $t^*=F^{-1}(R)=0.5$. We set the deadline as $t=t^*+\epsilon$ where $\epsilon=0.15$, perform 10 independent trials and calculate the empirical probability of decoding error. Figure \ref{fig:errorbar} shows that the empirical error probability as a function of the polarization level $n$, where the corresponding number of workers is $N=2^n$. The error bars represent one standard deviation. It can be observed that the error probability decreases to zero sharply as predicted by Theorem \ref{thm:runtimeguarantees}.

\subsection{Experiments on Amazon Web Services (AWS)}
Here we illustrate an application of computational polarization to computation times obtained from AWS Lambda serverless computing platform. Serverless computing is an emerging architectural paradigm that presents a compelling option for dynamic data intensive problems. Serverless computing relies on stateless functions that are automatically scheduled by the cloud infrastructure. Thus, they obviate the need for the provider to explicitly configure, deploy, and manage long-term compute units. We consider a distributed matrix multiplication task with random data matrices, where we employ $512$ worker nodes with rate set to $\frac{1}{2}$. In Figure \ref{polarized_cdfs}, we plot the runtime distribution of uncoded computation, computational polarization and MDS coded computation. It can be observed that the overall runtime of the proposed scheme is very close to MDS coded computing. However, the decoding process is significantly faster, especially when the number of worker nodes is very large.

\begin{figure}[!t]
\begin{minipage}[b]{0.3\linewidth}
  \centering
  \centerline{\includegraphics[width=4.2cm]{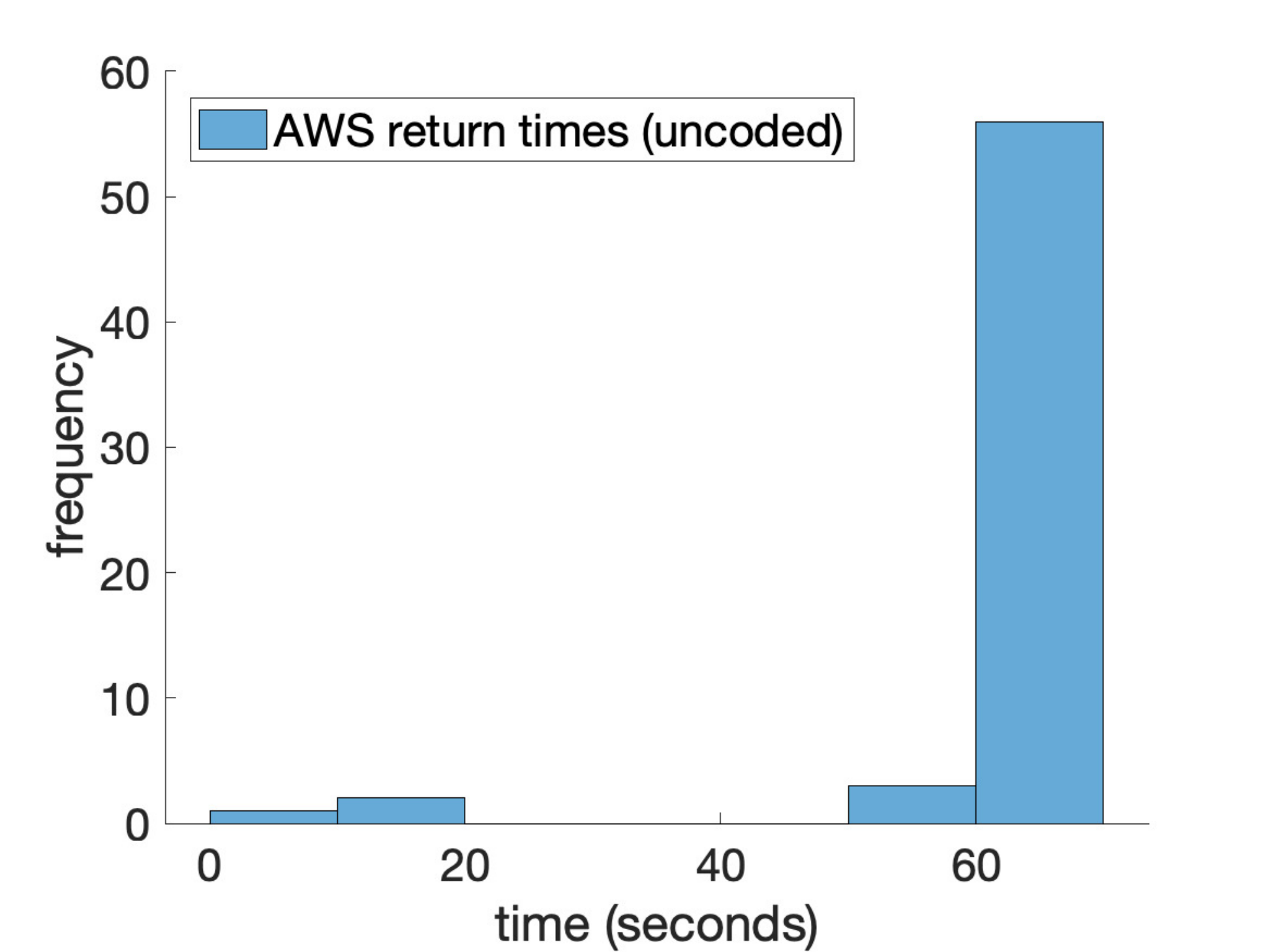}}
  \centerline{(a) Uncoded}\medskip
\end{minipage}
\hfill
\begin{minipage}[b]{0.3\linewidth}
  \centering
  \centerline{\includegraphics[width=4.2cm]{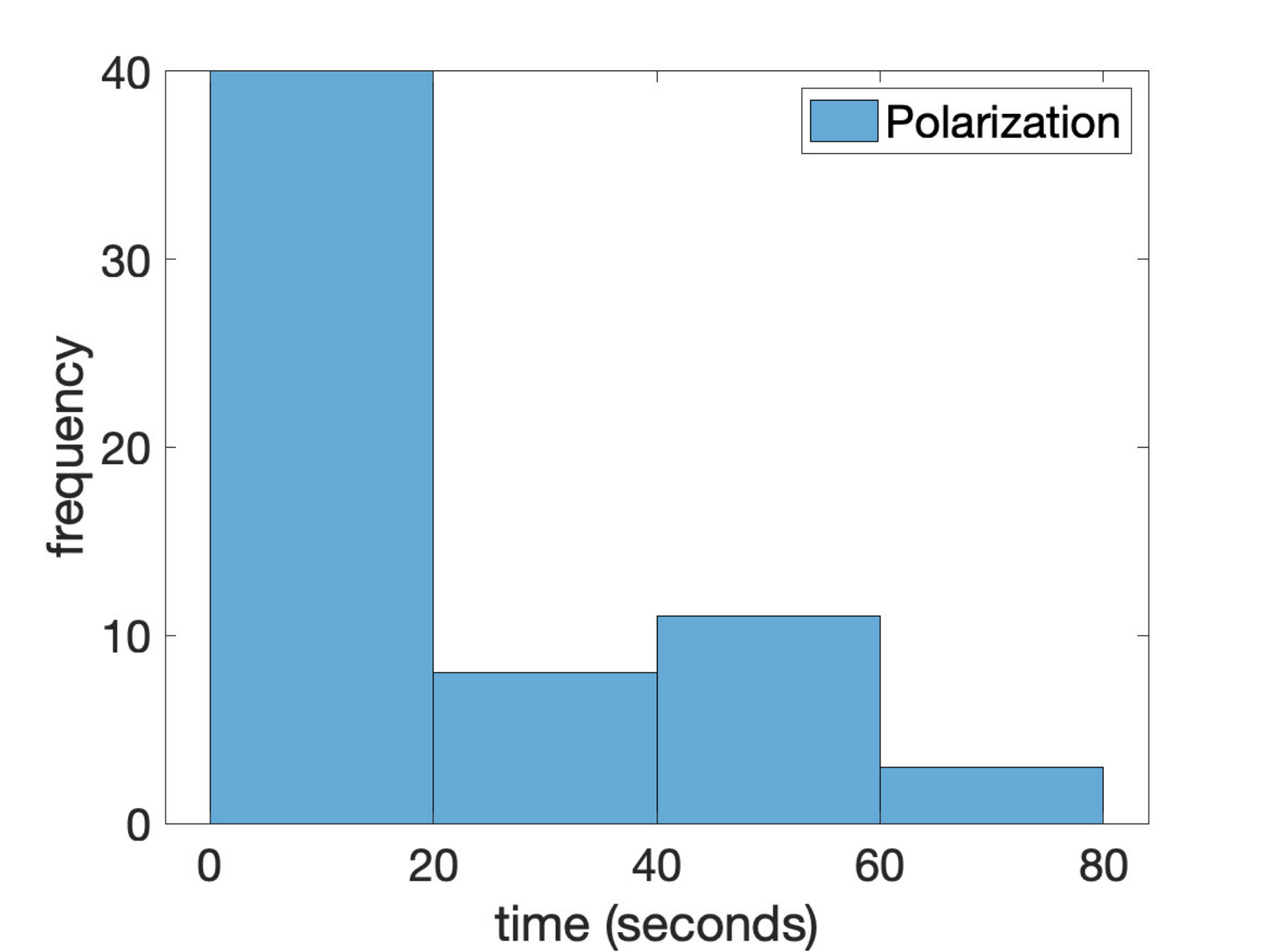}}
  \centerline{(b) Computational Polarization}\medskip
\end{minipage}
\hfill
\begin{minipage}[b]{.3\linewidth}
  \centering
  \centerline{\includegraphics[width=4.2cm]{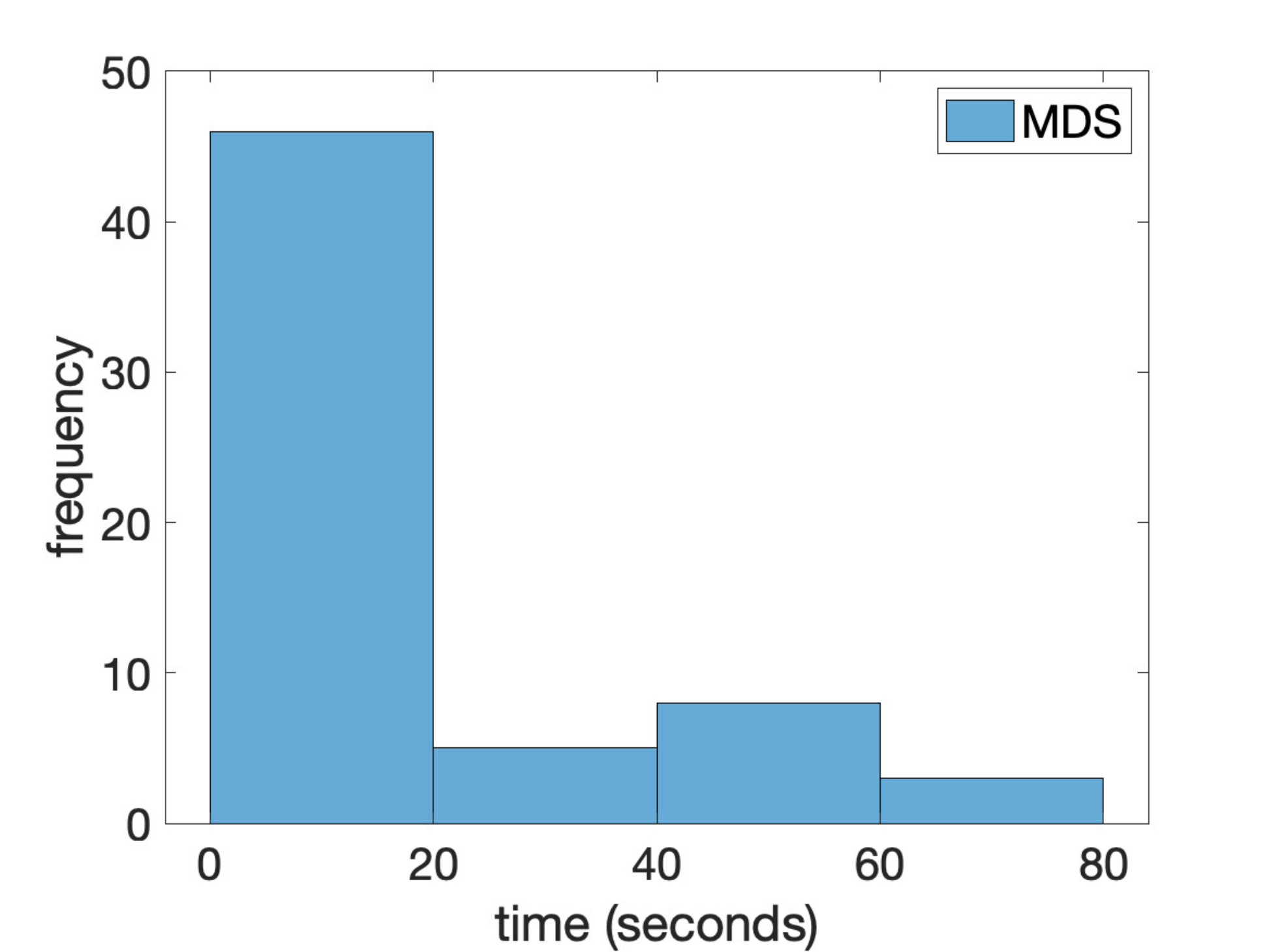}}
  \centerline{(c) MDS coded}\medskip
\end{minipage}
\caption{Histograms of computation times in uncoded, computational polarization coded, and MDS (Reed-Solomon) coded computations performed on AWS Lambda.}
\label{polarized_cdfs}
\end{figure}

\section{Conclusion}

We presented a resilient computing framework that extends the channel polarization phenomenon in an algorithmic sense and leverages a similar recursive construction. The proposed method creates virtual workers with a wide spectrum of runtimes by applying a split operation in a similar spirit to Polar Codes. Unlike polar codes, the polarization takes place over real valued random variables, and do not converge to discrete extremes. Instead, we have shown that the runtime distributions of the workers form a Banach space valued martingale and converge to Dirac delta measure. In this sense, the convergence is towards the extreme points of the total variation ball over the space of Radon measures. Moreover, we have shown an application of Fano's inequality to lower bound the probability of successful recovery of the overall computation. We showed that the proposed computational polarization scheme achieves information theoretically optimal overall runtime. The major advantage of the computational polarization is the negligible decoding complexity, which only involves addition and subtraction operations over the reals. This is in contrast to other coding schemes such as MDS codes. We leave extending the proposed method to computing nonlinear functions as an important future work, where the developed convergence theory can be used. Notable examples of distributed nonlinear computations that can benefit from Computational Polarization include solving linear systems and convex optimization problems \cite{PilWai14a,pilanci2017newton,lacotte2020effective}, and training neural networks \cite{pilanci2020neural}. Another interesting research direction is investigating applications of the Banach space martingales in traditional Polar Codes and improving the existing convergence analysis.

\newpage
\section{Proofs}
\label{sec:proofs}

\begin{proof}[Proof of Lemma \ref{lem:adjoint}]
We now verify the adjoint property. Suppose that $p(t)$ and $q(t)$ are two probability densities and consider the $L_2$ space inner-product
\begin{align*}
\Big \langle \mathcal{L}^*\{ p(t)\}, q(t) \Big\rangle &= \int_{-\infty}^\infty \int_{-\infty}^{t} p(u)du\, q(t)dt \\
&= \int_{-\infty}^\infty \int_{-\infty}^{\infty} p(u)\, q(t)\,1[u\le t]du\,dt\\
&=    \int_{-\infty}^{t} \int_{-\infty}^\infty p(u)\, q(t)\,1[t\ge u]\, dt\,du \\
&= \int_{-\infty}^{t} \int_{u}^\infty p(u)\, q(t)\, dt\,du\,.
\end{align*}
Finally, noting that the preceding expression equals
\begin{align*}
&= \int_{-\infty}^{t} \int_{u}^\infty  q(t)\, dt\,p(u)\,du\\
& = \Big \langle p(t), \mathcal{L}\{ q(t)\}) \Big\rangle\,,
\end{align*}
we prove that $L^*$ is the adjoint of the operator $L$.
\end{proof}
\begin{proof}[Proof of Theorem \ref{thm:functionalconvergencel2}]
Using the definition \eqref{eq:functional_update}, we expand $\Exs F_{n+1}(t)^2 $ as follows
\begin{align*}
\Exs F_{n+1}(t)^2 &= \Exs \big(F_n(t) + \epsilon_n F_n(t)\big(1-F(t)\big) \big)^2\\
&= \Exs F_n(t)^2 + \Exs 2\epsilon_n F_n(t)^2(1-F_n(t)) + \Exs F_n(t)^2(1-F_n(t))^2 \\
&= \Exs F_n(t)^2 + \Exs F_n(t)^2(1-F_n(t))^2\,.
\end{align*}
Integrating both sides, we obtain
\begin{align*}
\int \Exs F_{n+1}(t)^2 dt &= \int \Exs F_n(t)^2 dt + \int \Exs F_n(t)^2(1-F_n(t))^2dt\,.
\end{align*}
Noting the relation 
\begin{align*}
\vert F_{n+1}(t)-F_n(t) \vert^2 &= \vert \epsilon_n F_n(t)\big(1-F(t)\big)\vert\\
&= \vert F_n(t)\big(1-F(t)\big)\vert \,,
\end{align*}
which follows from the functional update in \eqref{eq:functional_update}, and $\vert \epsilon_n\vert =1$, we express the preceding equality as 
\begin{align}
\int \Exs F_{n+1}(t)^2 dt - \int \Exs F_n(t)^2 dt = \int \Exs \vert F_{n+1}(t)-F_n(t) \vert^2 dt \label{eq:squaredintegralconvergence1}\,.
\end{align}
Let us define the deterministic sequence $\{m_n\}_{n=1}^\infty$ as
\begin{align*}
m_n :=  \Exs \int  F_n(t)^2 dt = \int \Exs F_n(t)^2 dt\,,
\end{align*}
where the second equality follows from Fubini's theorem, noting that $F_n(t) \in [0,1]$ and the integrand $F_n(t)^2\le F_n(t) $ is absolutely integrable, i.e.,
\begin{align*}
\Exs \int \vert F_n(t)^2\vert dt \le \Exs \int \vert F_n(t)\vert = \Exs \int  F_n(t)  dt < \infty\,. 
\end{align*}
We may express $\Exs \|F_{n+1}(t)-F_{n}(t)\|_{L_2}$ by another application of Fubini's theorem as
\begin{align*}
\Exs \|F_{n+1}(t)-F_{n}(t)\|_{L_2}^2= \int \Exs  (F_{n+1}(t)-F_{n}(t))^2dt\,,
\end{align*}
which follows from the absolute integrability, i.e.,
\begin{align*}
\int (F_{n+1}(t)-F_{n}(t))^2dt &\le 2 \int F_{n+1}(t)^2dt+2\int F_{n}(t)^2dt\\
& \le  2 \int F_{n+1}(t)dt+2\int F_{n}(t)dt\\
& =  4 \int F(t)dt\\
& < \infty\,,
\end{align*}
by our assumption.

Using the definition of $m_n$, combining the expression for $\Exs \|F_{n+1}(t)-F_{n}(t)\|_{L_2}$ above with the relation \eqref{eq:squaredintegralconvergence1}, we obtain
\begin{align}
\Exs \|F_{n+1}(t)-F_{n}(t)\|^2_{L_2} = m_{n+1}-m_n\,. \label{eq:squaredintegralconvergence2}
\end{align}
Next, we consider the convergence of the sequence $m_n$. Note that
\begin{align*}
  m_{n+1}\ge m_n\,,
\end{align*}
since $\Exs \|F_{n+1}(t)-F_{n}(t)\|_{L_2}\ge 0$ in \eqref{eq:squaredintegralconvergence2}. Furthermore the sequence $\{m_n\}_{n=1}^\infty$ is bounded since we have
\begin{align*}
m_n = \Exs \int F_{n}(t)^2 dt \le \Exs \int F_{n}(t) dt\,.  
\end{align*}
The first inequality follows from $0\le F_n(t) \le 1$ for all $n\in \mathbb{N}$. Moreover, since the functional process \eqref{eq:functional_update} is a martingale, it holds that
\begin{align*}
\int F_n(t)dt = \int F(t) dt = \mean(T)\,,
\end{align*}
which shows that $m_{n+1} \le m_n\le \mean(T)$, where $\prob{T\le t}=F(t)$.
Noting that $\{m_n\}_{n=1}^\infty$ is bounded and monotone, we apply monotone convergence theorem, which shows that the sequence has a finite limit. Therefore we have
\begin{align*}
\lim_{n\rightarrow \infty} \vert m_{n+1}-m_{n} \vert = 0\,.
\end{align*}
Combining the preceding result with \eqref{eq:squaredintegralconvergence2}, we obtain
\begin{align*}
\lim_{n\rightarrow \infty} \Exs \|F_{n+1}(t)-F_{n}(t)\|^2_{L_2} = 0\,.
\end{align*}
This proves the first claimed result
Consequently, an application of Markov's inequality yields that
\begin{align*}
\prob{\|F_{n+1}(t)-F_{n}(t)\|^2_{L_2}\ge \epsilon} \le \frac{\|F_{n+1}(t)-F_{n}(t)\|_{L_2}}{\epsilon}\,,
\end{align*}
for any fixed scalar $\epsilon>0$. Applying the previous result on the right-hand-side we obtain
\begin{align*}
\lim_{n\rightarrow \infty} \prob{\|F_{n+1}(t)-F_{n}(t)\|^2_{L_2}\ge \epsilon} \le \frac{\lim_{n\rightarrow \infty} \|F_{n+1}(t)-F_{n}(t)\|_{L_2}}{\epsilon}=0\,,
\end{align*}
which completes the proof of the theorem. We further note that
\begin{align*}
 \Exs \|F_{n+1}(t)-F_{n}(t)\|_{L_2} \le \sqrt{\Exs \|F_{n+1}(t)-F_{n}(t)\|^2_{L_2}}\,,
 \end{align*}
 by Jensen's inequality, and the preceding probabilistic bounds also apply to $\|F_{n+1}(t)-F_{n}(t)\|_{L_2}$.
\end{proof}
\begin{proof}[Proof of Theorem \ref{thm:functionaluniform}]
We consider with the representation $F_{t+1}(t) = F_n(t) + \epsilon_n F_n(t)(1-F_n(t))$ and expand $\Exs F_{n+1}(t)^2$ as follows
\begin{align*}
\Exs F_{n+1}(t)^2 &= \Exs F_n(t)^2 + \Exs 2\epsilon_n F_n(t)^2(1-F_n(t)) + \Exs F_n(t)^2(1-F_n(t))^2 \\
&= \Exs F_n(t)^2 + \Exs F_n(t)^2(1-F_n(t))^2\,.
\end{align*}
It follows that $\Exs F_{n+1}(t)^2 \ge \Exs F_{n}(t)^2$ for all $n\in \mathbb{N}$ and $t\in \reals$ since we have
\begin{align}
\Exs F_{n+1}(t)^2 - \Exs F_{n}(t)^2 = \Exs F_n(t)^2(1-F_n(t))^2 \ge 0 \quad \forall n\in\mathbb{N} \mbox{ and } t\in\reals\,.\label{eq:pointwisenonneg}
\end{align}
Furthermore, note that $0\le F_n(t)\le 1$, hence for every fixed value of $t\in [0,1]$, $F_n(t)$ is bounded and monotone. Consequently, $\lim_{n\rightarrow 0} F_{n+1}(t)-F_n(t) = 0$ by monotone convergence theorem, and the limit exists. This was noted by \cite{alsan2016simple}, where a simple proof of polarization using the monotone convergence theorem was presented. We then introduce the following function as in \cite{alsan2016simple}
\begin{align*} %
\delta(a,b) := \inf_{t\in \mathbb{R} \,:\, a \le  F(t) \le b}\, F^2(t)(1-F(t))^2\,.
\end{align*}
Note that $\delta(a,b)>0$ for all $a>0$ and $b<1$.
Using the above definition, we obtain the following lower bound
\begin{align*}
\Exs F_n(t)^2(1-F_n(t))^2 &\ge \Exs \inf_{F_n(t)\,:\, a \le  F_n(t) \le b}\, F_n^2(1-F_n)^2\, 1[ a\le F_n(t) \le b]\\
&\ge   \delta(a,b) \, \prob{F_n(t) \in [a,b]}\,.
\end{align*}
Now we plugin the above lower bound in the earlier expression for $\Exs F_{n+1}(t)^2$ and obtain
\begin{align*}
\prob{F_n(t) \in [a,b]} \le \frac{\Exs F_{n+1}(t)^2 - \Exs F_n(t)^2}{\delta(a,b)}\,.
\end{align*}
Taking supremum over $t$ both sides in the above expression, we get
\begin{align}
\sup_{t\in\mathcal{T}} \prob{F_n(t) \in [a,b]} \le \frac{\sup_{t\in\mathcal{T}} \Exs F_{n+1}(t)^2 - \Exs F_n(t)^2}{\delta(a,b)}\,. \label{eq:uniformconvsup}
\end{align}
On the other hand, Theorem \ref{thm:arikan} shows that, for a fixed value of $t$, $F_n(t)$ converges to a limiting random variable $F_{\infty}(t)$ almost surely, which satisfies $F_\infty(t)^2=F_{\infty}(t)$ with probability one. Therefore, we have
\begin{align}
\lim_{n\rightarrow \infty} \Exs F_n(t)^2 &= \Exs F_{\infty}(t)^2 \nonumber \\
&= \Exs F_{\infty}(t) \nonumber\\
&= F(t)\,, \label{eq:pointwisecont}
\end{align}
where the last equality follows from the martingale property of the process $\{F_n(t)\}_{n=1}^\infty$.

Now let us fix $t$, and note that the sequence of functions $\{\Exs F_n(t)^2\}_{n=1}^\infty$ satisfies
\begin{align*}
&(i)  & &\Exs F_{n+1}(t)^2\ge \Exs F_{n}(t)^2\enskip  \mbox {as implied by \eqref{eq:pointwisenonneg}}\\
&(ii) & &\Exs F_n(t)^2 \enskip\mbox{and}\enskip  \lim_{n\rightarrow \infty} \Exs F_n(t)^2 \enskip \mbox{are continuous in $t$ as implied by \eqref{eq:pointwisecont}}\\
&(iii)& &\mbox{the domain of } \Exs F_n(t)^2 \enskip\mbox{and}\enskip  \lim_{n\rightarrow \infty} \Exs F_n(t)^2,~\mathcal{T } \enskip \mbox{is a compact metric space.}
\end{align*}
Consequently, Dini's theorem (see e.g. \cite{Rudin}) can be applied to obtain that
\begin{align*}
 \lim_{n\rightarrow \infty}\sup_{t\in\mathcal{T}} \Exs F_{n+1}(t)^2  - \Exs F_n(t)^2 &= \lim_{n\rightarrow \infty} \sup_{t\in\mathcal{T}} \big(\Exs F_{n+1}(t)^2 -\Exs F_{\infty}(t)^2\big)  - \big(\Exs F_n(t)^2 - \Exs F_{\infty}(t)^2\big)  \\
 &\le 2 \lim_{n\rightarrow \infty} \sup_{t\in\mathcal{T}} \left\vert\Exs F_{n+1}(t)^2 -\Exs F_{\infty}(t)^2\right \vert\\
&=0\,,
\end{align*}
where we have applied triangle inequality in the second line. Combining the above result with the bound in \eqref{eq:uniformconvsup}, we conclude that
\begin{align*}
 \lim_{n\rightarrow \infty} \sup_{t\in\mathcal{T}} \prob{F_n(t) \in [a,b]} \le \lim_{n\rightarrow \infty} \frac{\sup_{t\in\mathcal{T}} \Exs F_{n+1}(t)^2 - \Exs F_n(t)^2}{\delta(a,b)} = 0\,,
\end{align*}
which completes the proof.
\end{proof}
\begin{proof}[Proof of Theorem \ref{thm:convergencerate1}]
We define the following difference function $\Delta_n(t)$ as
\begin{align*}
\Delta_{n} (t) := \big| F_{n+1}(t) - F_{n}(t)\big |\,.
\end{align*}
Using the functional update equation \eqref{eq:fncupdate}, we first establish that the difference function satisfies
\begin{align*}
\Delta_{n}(t) = F_n(t)\big(1-F_n(t)\big),
\end{align*}
where we have used the inequalities $0\le F_n \le 1$ and $0\le 1-F_n(t) \le 1$.
Plugging in the form of the functional update in \eqref{eq:fncupdate} shows that the difference function further satisfies the following relation
\begin{align*}
\Delta_{n}(t) &= F_n(t)\big(1-F_n(t)\big) \\
&= \left(F_{n-1}(t) + \epsilon_n F_{n-1}(t)\big(1-F_{n-1}(t)\big)\right)\left(1-F_{n-1}(t) - \epsilon_n F_{n-1}(t)\big(1-F_{n-1}(t)\big)\right)\\
&=F_{n-1}(t)(1-F_{n-1}(t)) - \big(F_{n-1}(t)\big(1-F_{n-1}(t)\big)\big)^2 + \epsilon_n \big( F_{n-1} (1-F_{n-1})^2 - F_{n-1}^2 (1-F_{n-1})\big) \\
&=\Delta_{n-1}(t)(1- \Delta_{n-1}(t)) + \epsilon_n \Delta_{n-1} (1-2F_{n-1}(t))\\
&=\Delta_{n-1}(t)\big(1- \Delta_{n-1}(t) + \epsilon_n (1-2F_{n-1}(t)) \big )\\
&= \begin{cases}
\Delta_{n-1}(t)(1-F_{n-1}(t))(2-F_{n-1}(t)) &\mbox{ for } \epsilon_n=+1\\
\Delta_{n-1}(t)F_{n-1}(t)(1+F_{n-1}(t)) &\mbox{ for } \epsilon_n=-1\,.\\
\end{cases}
\end{align*}
Next we consider the random variable $\int_{-\infty}^\infty \Delta^{1/2}_{n+1}(t) dt$, which satisfies
\begin{align*}
\int_{-\infty}^\infty \Delta^{1/2}_{n+1}(t) dt
&= \begin{cases}
\int_{-\infty}^\infty \Delta^{1/2}_{n-1}(t)(1-F_n(t))^{1/2}(2-F_n(t))^{1/2}dt &\mbox{ for } \epsilon_n=+1\\
\int_{-\infty}^\infty  \Delta^{1/2}_{n-1}(t)F_n(t)^{1/2}(1+F_n(t))^{1/2}dt &\mbox{ for } \epsilon_n=-1.\\
\end{cases} 
\end{align*}
Calculating the expected value of $\int_{-\infty}^\infty \Delta^{1/2}_{n+1}(t) dt$ we obtain 
\begin{align*}
\Exs \int_{-\infty}^\infty \Delta^{1/2}_{n+1}(t) dt &= \frac{1}{2} \int_{-\infty}^\infty \Delta^{1/2}_{n}(t) \left( 1-F_n(t))^{1/2}(2-F_n(t))^{1/2} + F_n(t)^{1/2}(1+F_n(t))^{1/2} \right) dt\\
&\le  \frac{1}{2} \int_{-\infty}^\infty \Delta^{1/2}_{n}(t) dt \left\{ \max_{t} \left( 1-F_n(t))^{1/2}(2-F_n(t))^{1/2} + F_n(t)^{1/2}(1+F_n(t))^{1/2} \right) \right\}\\
& = \frac{\sqrt{3}}{2} \int_{-\infty}^\infty \Delta^{1/2}_{n}(t) dt\,.
\end{align*}
Recursively applying the above inequality we obtain
\begin{align*}
\Exs \int_{-\infty}^\infty \Delta^{\frac{1}{2}}_{n}(t) dt &\le \big(\frac{3}{4}\big)^{\frac{n}{2}} \Exs \int_{-\infty}^\infty \Delta^{\frac{1}{2}}_{0}(t)dt\\
&=  \big(\frac{3}{4}\big)^{\frac{n}{2}} \Exs \int_{-\infty}^\infty (F(t))^{\frac{1}{2}}(1-F(t))^{\frac{1}{2}}dt\,.
\end{align*}
More generally, for any $\beta$ satisfying $0<\powerp\le \frac{1}{2}$, 
we have 
\begin{align*}
\Exs \int_{-\infty}^\infty \Delta^{\powerp}_{n+1}(t) dt &= \frac{1}{2} \int_{-\infty}^\infty \Delta^{\powerp}_{n}(t) \left( 1-F_n(t))^{\powerp}(2-F_n(t))^{\powerp} + F_n(t)^{\powerp}(1+F_n(t))^{\powerp} \right) dt\\
&\le  \frac{1}{2} \int_{-\infty}^\infty \Delta^{\powerp}_{n}(t) dt \left\{ \max_{t} \left( 1-F_n(t))^{\powerp}(2-F_n(t))^{\powerp} + F_n(t)^{\powerp}(1+F_n(t))^{\powerp} \right) \right\}\\
& = \left(\frac{3}{4}\right)^{\powerp} \int_{-\infty}^\infty \Delta^{\powerp}_{n}(t) dt
\end{align*}
Recursively applying the final inequality above, we obtain
\begin{align*}
\Exs \int_{-\infty}^\infty \Delta^{\powerp}_{n}(t) dt
& \le \left(\frac{3}{4}\right)^{\powerp n} \int_{-\infty}^\infty \Delta^{\powerp}_{0}(t) dt\\
& = \left(\frac{3}{4}\right)^{\powerp n} \int_{-\infty}^\infty \big(F(t)(1-F(t))\big)^{\powerp}(t) dt
\end{align*}
Now we apply Markov's inequality to the random variable $\int (F_n(t) (1-F_n(t)))^{\beta}dt$ and obtain the inequality
\begin{align*}
\prob{ \int (F_n(t) (1-F_n(t)))^{\beta}dt > \tau} \le \frac{1}{\tau}\Exs \int (F_n(t) (1-F_n(t)))^{\beta}dt
\end{align*}
Using the earlier expression for the expectation on the right-hand-side, we obtain
\begin{align*}
\prob{ \int (F_n(t) (1-F_n(t)))^{\beta}dt > \tau} \le \frac{1}{\tau}\left(\frac{3}{4}\right)^{\powerp n} \int_{-\infty}^\infty \big(F(t)(1-F(t))\big)^{\powerp} dt.
\end{align*}
Setting $\tau = \rho^n$, we obtain the bound
\begin{align*}
\prob{ \int (F_n(t) (1-F_n(t)))^{\beta}dt > \rho^n} \le \left(\frac{3}{4\rho}\right)^{\powerp n} \int_{-\infty}^\infty \big(F(t)(1-F(t))\big)^{\powerp} dt.
\end{align*}
This completes the proof. The above bound is effective as $n\rightarrow \infty$, when $\rho\in[0,1)$ and $\frac{3}{4\rho}<1$.
\end{proof}

\begin{proof}[Proof of Theorem \ref{thm:non_asymptotic}]
~\\
We begin by proving the following auxiliary lemmas to establish the claimed result.
\newcommand{\myepsmin}{\delta}
\begin{lemma}
\label{lem:minFbound}
For a fixed $\myepsmin \le \frac{1}{4}$, the condition $F_n(t)(1-F_n(t))\le \myepsmin$ implies that  
\begin{align*}
\min(F_n(t),1-F_n(t)) \le  \frac{1-\sqrt{1-4 \myepsmin }}{2}
\end{align*}
\end{lemma}
\begin{proof}
Note that $(x-x_1)(x-x_2)=x-x^2-\myepsmin$ is a factorization of the polynomial $x-x^2-\myepsmin = x(1-x)-\myepsmin$, where $x_1$ and $x_2$ are the roots given by 
\begin{align*}
x_1 &= \frac{1}{2} \left(1-\sqrt{1-4\myepsmin}\right)\\
x_2 & = \frac{1}{2} \left(1+\sqrt{1-4\myepsmin}\right)\,.
\end{align*}
The above roots are real valued and satisfy $x_1\le x_2$ for $\delta\le \frac{1}{4}$. The condition $x-x^2-\myepsmin\le 0$ implies that $x \in (-\infty,x_1]\cup[x_2,\infty)$\,. Rearranging the preceding condition implies that $\min(x,1-x) \le \frac{1}{2}(1-\sqrt{1-4\myepsmin})\,.$
The condition $x-x^2-y\le 0$ implies that $x$
\end{proof}

\begin{lemma}
\label{lem:geodecaybound}
For every fixed $t \in [0,1]$, and any $\rho \in (\frac{3}{4},1)$ it holds that
\begin{align}
\prob{ F_n(t) (1-F_n(t)) > \rho^n} \le \big(\frac{3}{4\rho}\big)^{\frac{n}{2}}\,.
\end{align}
\end{lemma}

\begin{lemma}
\label{lem:exponentbound}
Suppose that $F_n(t) \le \big(\frac{1}{2}\big)^{\frac{1}{\epsilon}}$\,, then we have
\begin{align*}
F_{n+1}(t) \le F_n(t)^{Z_n (1-\epsilon)} 
\end{align*}
\end{lemma}

\begin{lemma}
\label{lem:negatedexponentbound}
Suppose that $1-F_n(t) \le \big(\frac{1}{2}\big)^{\frac{1}{\epsilon}}$\,, then we have
\begin{align*}
1-F_{n+1}(t) \le (1-F_n(t))^{Z_n (1-\epsilon)} 
\end{align*}

\end{lemma}

\begin{figure}[t!]
\centering
\includegraphics[width=6cm]{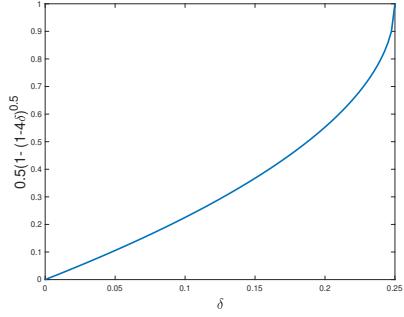}
\caption{Plot of the function $f(\delta):=\frac{1}{2}(1-\sqrt{1-4\delta})$ where $\delta\in[0,\frac{1}{4}]$. \label{fig:funconeminus}}
\end{figure}

Suppose that we have $F_{n_0}(t)(1-F_{n_0}(t))\le{} \delta$ for some $n\in\mathbb{Z}_+$. In particular, Lemma \ref{lem:geodecaybound} guarantees the existence of an $n=n_0$ such that this assumption holds with high probability. Further suppose that there exists $n_1 \in \mathbb{Z}_+$ such that for all $n \in \{n_0,...,n_1\}$ the preceding inequality holds. More precisely, consider that the condition
\begin{align}
\max_{n \in \{n_0,...,n_1\}} F_{n}(t)(1-F_{n}(t))\le \delta\,,
\label{eqn:discretemaxbound}
\end{align}
holds.

Now we consider any $n \in \{n_0,...,n_1\}$. We can apply Lemma \ref{lem:minFbound} to deduce that $\min(F_n(t),1-F_n(t)) \le \frac{1}{2}(1-\sqrt{1-4\delta})$. Observe that the function $\frac{1}{2}(1-\sqrt{1-4\delta})\le $ obeys
\begin{align*}
\frac{1}{2}(1-\sqrt{1-4\delta})\le \frac{1}{2}(1-(1-4\delta))=2\delta\,,
\end{align*}
and can be made arbitrarily small as $\delta \rightarrow 0$\, as shown in Figure \ref{fig:funconeminus}\,. Let us set $\delta$ sufficiently small to satisfy
\begin{align*}
\frac{1}{2}(1-\sqrt{1-4\delta}) \le \big(\frac{1}{2}\big)^{\frac{1}{\epsilon}}\,.
\end{align*}
Note that Lemma \ref{lem:minFbound} implies
\begin{align*}
F_n(t) \le \big(\frac{1}{2}\big)^{\frac{1}{\epsilon}}  \mbox{~~ or~~ } 1-F_n(t) \le \big(\frac{1}{2}\big)^{\frac{1}{\epsilon}}\,.
\end{align*}
Now, we consider the following two cases separately:
\begin{itemize}
\item[case $1$:] $F_n(t)\le \frac{1}{2}$ and $F_n(t) \le 1-F_n(t)$\\
Since $F_n(t)=\min(F_n(t),1-F_n(t))\le \big(\frac{1}{2}\big)^{\frac{1}{\epsilon}} $, Lemma \ref{lem:exponentbound} implies that 
\begin{align*}
F_{n+1}(t) \le F_n(t)^{Z_n(1-\epsilon)}
\end{align*}

\item[case $2$:] $F_n(t) > \frac{1}{2}$ and $F_n(t) > 1-F(t)$\\
Since $1-F_n(t)=\min(F_n(t),1-F_n(t))\le \big(\frac{1}{2}\big)^{\frac{1}{\epsilon}} $, Lemma \ref{lem:negatedexponentbound} implies that 
\begin{align*}
1-F_{n+1}(t) \le (1-F_n(t))^{Z_n(1-\epsilon)}\,.
\end{align*}
\end{itemize}
Combining the above two cases, we obtain the bound
\begin{align*}
\min(F_{n+1}(t),1-F_{n+1}(t)) \le \big[ \min(F_n(t)(1-F_n(t)) \big]^{Z_n(1-\epsilon)}\,.
\end{align*}
\newcommand{\nzero}{{n_0}}
\newcommand{\none}{{n_1}}
Now we apply the preceding inequality recursively from $n=\nzero\in\mathbb{Z}_+$ to $n=\none\in\mathbb{Z}_+$ where $\nzero>\none$ to obtain
\begin{align}
\label{eq:minrecursive}
\min(F_\none(t),1-F_\none(t)) = \min(F_\nzero(t),1-F_\nzero(t))^{\prod_{i=\nzero}^{\none-1} Z_i(1-\epsilon)}\,.
\end{align}
Now we focus on the exponent term $\prod_{i=\nzero}^{\none-1} Z_i(1-\epsilon)$ in the above inequality. Taking logarithms, we get
\begin{align*}
\log  \prod_{i=\nzero}^{\none-1} Z_n(1-\epsilon) = \sum_{i=\nzero}^{\none-1} \log Z_i + (\none-\nzero)\log(1-\epsilon)\,. 
\end{align*}
Noting that $\log Z_\nzero, ... \log Z_{\none-1}$ are i.i.d. Bernoulli variables satisfying
\begin{align}
\prob{\log Z_i=0} = \prob{\log Z_i=1} = \frac{1}{2}\, \mbox{for } i=\nzero,...,\none-1\,. 
\end{align}
We now invoke Chernoff's bound (see, e.g., \cite[p. 531]{gallagertext}) and obtain
\begin{align}
\label{eq:chernoff}
\prob{\frac{1}{\none-\nzero} \sum_{i=\nzero}^{\none-1} \log Z_i < \frac{1}{2} - \eta } \le 2^{-(\none-\nzero)(1-\mathcal{H}(\frac{1}{2}-\eta))}\,,
\end{align}
where we have used the fact that $\Exs \log Z_i = \frac{1}{2}\, \forall i$, and introduced the binary entropy function $\mathcal{H}(p) := -p \log p - (1-p) \log (1-p) $\,.
Using the probability inequality in \eqref{eq:chernoff}, we deduce that
\begin{align*}
\prod_{i=\nzero}^{\none-1} Z_i (1-\epsilon) &> 2^{(\none-\nzero)(\frac{1}{2}-\eta+\log(1-\epsilon))}\\
&=(1-\epsilon) 2^{(\none-\nzero)(\frac{1}{2}-\eta)}\,,
\end{align*}
with probability at least $1- 2^{-(\none-\nzero)(1-\mathcal{H}(\frac{1}{2}-\eta))}$\,.
Combining the above probabilistic bound with the earlier expression in \eqref{eq:minrecursive} we obtain that
\begin{align}
\min(F_\none(t),1-F_\none(t)) = \min(F_\nzero(t),1-F_\nzero(t))^{(1-\epsilon)2^{(\none-\nzero)(\frac{1}{2}-\eta)}}\,,
\label{eqn:minF}
\end{align}
holds with the same probability.
Now we verify our initial assumption in \eqref{eqn:discretemaxbound} that $F_{n}(t)(1-F_{n}(t))\le \delta$ simultaneously for all $n\in \{n_0,...,n_1\}$ using Lemma \ref{lem:geodecaybound}. 
 Let us pick $\log(\frac{1}{\delta})=n_0 \log(\frac{1}{\rho})$, which verifies that $\delta = \rho^{n_0}$. 
 By Lemma \ref{lem:geodecaybound}, the event $F_{n}(t)(1-F_{n}(t))\le \delta$ happens with probability $1-\big(\frac{3}{4\rho}\big)^{\frac{n}{2}}$ for a fixed value of $n\in \mathbb{Z}_+$. Applying union bound over $n\in \{n_0,...,n_1\}$, we obtain that the aforementioned condition \eqref{eqn:discretemaxbound} holds with probability at least
\begin{align*}
\prob{\max_{n \in \{n_0,...,n_1\}} F_{n}(t)(1-F_{n}(t))> \delta} &\le \sum_{n=n_0}^{n_1} \big(\frac{3}{4\rho}\big)^{\frac{n}{2}} \\
&= \frac{\big(\frac{3}{4\rho}\big)^{ \frac{n_0}{2} }-\big(\frac{3}{4\rho}\big)^{ \frac{n_1+1}{2} }}{1- \big(\frac{3}{4\rho}\big)^{\frac{1}{2}}}\\
& \le \frac{\big(\frac{3}{4\rho}\big)^{\frac{n_0}{2}}}{1- \big(\frac{3}{4\rho}\big)^{\frac{1}{2}}}\,.
\end{align*}
Now we recall the following choices
\begin{align}
&(i)\qquad\qquad 2\delta \le \big(\frac{1}{2}\big)^{\frac{1}{\epsilon}} \label{eq:deltacond1}\\
&(ii)\qquad\qquad \delta = \rho^{n_0}\,.
\end{align}
The above conditions can be satisfied by letting
\begin{align*}
\epsilon = \frac{1}{n_0 \log \frac{1}{\rho}-1}\,.
\end{align*}
Then we turn to the inequality in \eqref{eqn:minF}, note that $\min(F_\nzero(t),1-F_\nzero(t))\le \delta=\rho^{n_0}$ by our assumption in \eqref{eqn:discretemaxbound} and obtain
\begin{align}
\min(F_\none(t),1-F_\none(t)) \le \rho^{n_0(1-\epsilon)2^{(\none-\nzero)(\frac{1}{2}-\eta)}}\,,
\end{align}
using \eqref{eqn:minF}, with the stated probability.
Then note that $\rho^{n_0(1-\epsilon)} \le \frac{1}{2}$ whenever the parameter $n_0$ satisfies
\begin{align*}
n_0 \ge \frac{1}{(1-\epsilon)\log\frac{1}{\rho}}\,.
\end{align*}
Finally combining the earlier pieces we obtain that
\begin{align}
\prob{\min(F_\none(t),1-F_\none(t)) > 2^{-2^{(\none-\nzero)(\frac{1}{2}-\eta)}}} \le \frac{\big(\frac{3}{4\rho}\big)^{\frac{n_0}{2}}}{1- \big(\frac{3}{4\rho}\big)^{\frac{1}{2}}} + 2^{-(n_1-n_0)(1-\mathcal{H}(\frac{1}{2}-\eta))}\,.
\end{align}
Now let us pick $\nzero=\floor{\beta\log \none}\le \beta \log\none$. Now we set $n=\none$ and obtain
\begin{align}
\prob{\min(F_n(t),1-F_n(t)) > 2^{-2^{(n-\beta\log(n)(\frac{1}{2}-\eta)}}} &\le 
\frac{2^{-\frac{1}{2}\log(\frac{4\rho}{3})(\beta\log(n)-1)}}{1- \big(\frac{3}{4\rho}\big)^{\frac{1}{2}}} + 2^{-(n-\beta\log(n))(1-\mathcal{H}(\frac{1}{2}-\eta))}\nonumber\\
&=\frac{2^{-\frac{\beta}{2}\log(\frac{4\rho}{3})\log(n)}}{\big(\frac{3}{4\rho}\big)^{\frac{1}{2}}\big(1- \big(\frac{3}{4\rho}\big)^{\frac{1}{2}}\big)} + 2^{-(n-\beta\log(n))(1-\mathcal{H}(\frac{1}{2}-\eta))}\,,
 \label{eq:finalprobbound}
\end{align}
where we have used the fact that $n_0$
By picking $\epsilon=\frac{1}{2}$, we can satisfy $\delta\le \frac{1}{4}$ under the conditions \eqref{eq:deltacond1}, which is required to to apply Lemma \ref{lem:minFbound}. We then conclude that for $n_0=\log n \ge \frac{2}{\log\frac{1}{\rho}}$, the probabilistic bound in \eqref{eq:finalprobbound} is valid.

Finally, conditioned on the event under which \eqref{eq:finalprobbound} holds, we can bound $\| F_{n+1}(t) - F_n(t) \|_{L_p}$ as follows
\begin{align*}
\| F_{n+1}(t) - F_n(t) \|_{L_p} & = \| F_n(t)\big(1-F_n(t)\big) \|_{L_p} \\
& = \left( \int_a^b \vert F_n(t)\big(1-F_n(t)\big)\vert^p dt \right)^{1/p} \\
& \le \left( \int_a^b \vert 2^{-2^{(n-\beta\log(n)(\frac{1}{2}-\eta)}}\vert^p dt \right)^{1/p} \\
& =    2^{-2^{(n-\beta\log(n)(\frac{1}{2}-\eta)}}\, \int_a^b dt \\
& = \vert b-a \vert 2^{-2^{(n-\beta\log(n)(\frac{1}{2}-\eta)}}\,,
\end{align*}
where the inequality in the third line follows since the inequalities $F_n(t)\big( 1- F_n(t) \big) \le F_n(t)$ and  $F_n(t)\big( 1- F_n(t) \big) \le 1-F_n(t)$ simultaneously hold, and consequently we have the upper-bound
\begin{align*}
  F_n(t)\big( 1- F_n(t) \big) = \vert F_n(t)\big( 1- F_n(t) \big) \vert  \le \min(F_n(t),1-F_n(t)) \le 2^{-2^{(n-\beta\log(n)(\frac{1}{2}-\eta)}}\,,
\end{align*}
completing the proof of the theorem.
\end{proof}

\begin{proof}[Proof of Theorem \ref{thm:runtimeguarantees}]
\,\\
We start with the Laplace transform rule. Let us denote the number of indices $i$ that satisfy $M_{n,i}\le e^{\lambda t^*}$ for some $t^* \in \real$ by
\begin{align*}
\left \vert\left \{ i\,:\,M_{n,i}(\lambda)\le e^{\lambda t^*} \right\}\right\vert.
\end{align*}
Consequently, taking limits we obtain
\begin{align*}
\lim_{N\rightarrow \infty} \frac{1}{N}\left \vert\left \{ i\,:\,M_{n,i}(\lambda)\le e^{\lambda t^*} \right\}\right\vert &= \lim_{n\rightarrow \infty} \prob{M_{n} (\lambda) \le e^{\lambda t^*}}\\
&= \prob{ M_{\infty}(\lambda) \le e^{\lambda t^*}}\\
& = \prob{ e^{\lambda T} \le e^{\lambda t^*} }\\
& = \prob{ T \le t^* }\\
& = F(t^*)\\
& = R \numberthis \label{eq:laplaceratio}\,,
\end{align*}
where we have set $t^* = F^{-1}(R)$.

We apply the log-sum-exp upper bound on the expected computational run-time as follows
\begin{align*}
\Exs T_D &= \Exs \max_{i \in \mathcal{F}}\, T_{n,i}\\
    &= \Exs \max_{i \in \mathcal{F}}\, \frac{1}{\lambda} \log e^{\lambda T_{n,i}}\\
    &= \Exs \log \max_{i \in \mathcal{F}}\, \frac{1}{\lambda} e^{\lambda T_{n,i}}\\
& \le \frac{1}{\lambda} \log \Exs \sum_{i \in \mathcal{F}} e^{\lambda T_{n,i}}\\
& = \frac{1}{\lambda} \log  \sum_{i \in \mathcal{F}} M_{n,i}(\lambda) \,,
\end{align*}
where we have applied Jensen's inequality in the last inequality.

Next, we combine the bound \eqref{eq:laplaceratio} with the preceding log-sum-exp upper bound on the expected computational run-time in the limit where $N\rightarrow \infty$. We obtain
\begin{align*}
	\lim_{N\rightarrow \infty} \Exs T_D &\le \lim_{N\rightarrow \infty}  \frac{1}{\lambda} \log  \sum_{i \in \mathcal{F}} M_{n,i}(\lambda) \\
	& \le \lim_{N\rightarrow \infty}  \frac{1}{\lambda} \log  \Big( NR e^{\lambda t^*} \Big)\\
	& \le \lim_{N\rightarrow \infty}  \frac{1}{\lambda} \log(NR)+ t^*\,.
\end{align*}
Setting $\lambda = \frac{\log(NR)}{\epsilon}$ in the final inequality, we obtain
\begin{align*}
\lim_{N\rightarrow \infty} \Exs \, T_D \le t^* + \epsilon\,.
\end{align*}

We next consider the quantile selection rule. Let us denote the number of indices $i$ that satisfy $F_{n,i}\le t^*$ for some $t^* \in \real$ by 
$\left \vert\left \{ i\,:\, F_{n,i} \le  t^* \right\}\right\vert$.

First we obtain the asymptotic ratio
\begin{align*}
\lim_{N\rightarrow \infty} \frac{1}{N}\left \vert\left \{ i\,:\, F_{n,i} \le  t^* \right\}\right\vert  &= \lim_{N\rightarrow \infty} \prob{F_{n}\le t^*}\\
& = \prob{F_{\infty}\le t^*}\\
& = F(t^*)\\
&  = R\,,
\end{align*}
where we substituted $t^*=F^{-1}(R)$.

Now, we note that
\begin{align*}
\prob{T_D > t^* + \epsilon} &= \prob{ \max_{i\in\mathcal{F}} T_{n,i} > t^* + \epsilon} \\
& \le \sum_{i\in \mathcal{F}} \prob{T_{n,i} > t^* + \epsilon}\\
 & \le RN \max_{i \in \mathcal{F}} \prob{T_{n,i} > t^* + \epsilon}\,.
\end{align*}
Next we apply Theorem \ref{thm:non_asymptotic}, and obtain that
\begin{align*}
\prob{T_D > t^* + \epsilon} \le RN 2^{-c_1 \sqrt{N}}\,,
\end{align*}
where $c_1$ is a constant independent of $N$ to complete the proof.
\end{proof}

\begin{proof}[Proof of Theorem \ref{thm:infotheoryoptimality}]
~\\
Consider a uniformly generated source random variable $J \sim \mathrm{Uniform}\{1,\cdots, 2^K\}$, and let $b_1\cdots b_K$ be the corresponding binary expansion. We let $A_1,...,A_K\in \real^{m\times d}$ be equal to  $b_1 1_m 1_d^T,\cdots, b_K 1_m 1_d^T$. Consider the matrix vector product function $f(A_k)=A_k1_d\,\forall k$. We note that the choice of the $1_d$ vector is arbitrary and other choices are equally applicable. Clearly, $f(A_1),\cdots f(A_k)$ is sufficient to exactly reconstruct the source $J$. Suppose that $A_1^\prime,\cdots,A_N^\prime$ is the encoded data for any coded computation scheme, where the rate is $R=\frac{K}{N}$ and the computational tasks $f(A_1)\cdots f(A_N)$ are distributed to $N$ independent and identical worker nodes.
At any fixed time instant $t$, let $Y_1(t),\cdots Y_N(t)$ denote the output of the workers, and define $Y_{k}=\{e\}$, e.g., an erasure event whenever the $k$-th worker is not finished the task. Then we observe that
\begin{align*}
\prob{Y_k(t)=\{e\}}&=\prob{T^{(k)}\ge t}\\
&= 1-F(t)\,.
\end{align*}
Therefore, the mutual information between the source and available information at time $t$ obeys
\begin{align*}
I( f(A_k) ;\, Y_k)\le 1-(1-F(t))=F(t)\,\forall k\,,
\end{align*}
where the right-hand-side is the Shannon capacity of a memoryless binary erasure channel with erasure probability $1-F(t)$ (see, e.g., \cite{cover1999elements}). Since the worker nodes are i.i.d., we immediately have
\begin{align*}
I( f(A_1),\ldots,f(A_N) ;\, Y_1,\ldots Y_N)\le F(t)N\,,
\end{align*}
Furthermore, since $f(A_1)\ldots f(A_N)$ is a function of $J$, we have
\begin{align*}
I(J;\,Y_1\ldots Y_N)
&\le I( f(A_1),\ldots,f(A_N) ;\, Y_1,\ldots Y_N)\\
&\le  F(t) N\,,
\end{align*}
Let $\hat J(Y_1,\ldots,Y_N)$ be any estimate of $J$ based on $Y_1\ldots Y_N$. Next, we apply Fano's inequality to obtain
\begin{align*}
\prob{\hat J(Y_1,\ldots, Y_N) \neq J} &\ge \frac{H(J|Y_1\ldots Y_N)-1}{K}\\
& = \frac{K-I(J;\,Y_1\ldots Y_N)-1}{K}\\
&\ge 1- \frac{F(t)}{R}-\frac{1}{NR}\,.\,
\end{align*}
where we have used the preceding bound on $I(J;\,Y_1\ldots Y_N)$ in the final inequality.
Observe that $T_D\le t$ whenever decoding $f(A_1)\ldots f(A_K)$ succeeds at time $t$ and hence there exists an estimator based on $Y_1\ldots Y_N$ such that $\hat J(Y_1\ldots Y_N)=J$. Therefore we have
\begin{align*}
\prob{T_D \ge  t} \ge 1- \frac{F(t)}{R}-\frac{1}{NR}\,.
\end{align*}
Plugging in $t=F^{-1}(R(1-\beta))$ we obtain
\begin{align*}
\prob{T_D \ge  t} &\ge 1 - \frac{F( F^{-1}(R(1-\beta)))}{R}-\frac{1}{NR}\\
& \ge 1 - \frac{R(1-\beta)}{R}-\frac{1}{NR}\\
&\ge \beta - \frac{1}{NR}\,.
\end{align*}
\end{proof}

\begin{proof}[Proof of Lemma \ref{lem:pdfupdate}]
Taking time derivatives $\frac{\partial }{\partial t}$ on both sides we obtain the probability density martingale
\begin{align}
\frac{\partial F_{n+1}(t)}{\partial t}  = \frac{\partial F_n(t)}{\partial t} + \epsilon_n \frac{\partial F_n(t)}{\partial t}\big(1-F_n(t)\big) - \epsilon_n  F_n(t)\frac{\partial F_n(t)}{\partial t}\,.
\label{eq:functional_update_cdf_for_pdf}
\end{align}
In terms of the probability density $p_{n+1}(t):= \frac{\partial F_{n+1}(t)}{\partial t}$, the above update equation reduces to
\begin{align*}
p_{n+1}(t)  &= p_n(t)+ \epsilon_n p_n(t)\big(1-F_n(t)\big)+ \epsilon_n  F_n(t)\Big(1-p_n(t)\Big)\\
& = p_n(t)+ \epsilon_n p_n(t)\Big(1-\int_{-\infty}^t p_n(u)du\Big) - \epsilon_n p_n(t) \int_{-\infty}^t p_n(u)du \,.
\label{eq:functional_update_pdf}
\end{align*}
Next, we plug-in the relation%
\begin{align*}
1-\int_{-\infty}^t p_n(u)du = \int_{t}^\infty p_n(u)du\,,
\end{align*}
which follows from $\int_{-\infty}^{\infty} p_n(t)=1$, and enables us to further simply the probability density process as
\begin{align}
p_{n+1}(t) & = 
\begin{cases} 
2p_n(t)\int_{t}^{\infty} p_n(u)du & \mbox{for } \epsilon_n=+1\\
2p_n(t)\int_{-\infty}^{t} p_n(u)du & \mbox{for } \epsilon_n=-1
 \end{cases}
\end{align}
and reach the claimed identity.
\end{proof}

\begin{proof}[Proof of Theorem \ref{thm:joint}]
As stated, we assume that the random variables $T_1$ and $T_2$ admit a continuous density. Let $v$ and $u$ be real numbers satisfying $u\le v$. Consider the joint cumulative density function $\prob{\max(T_1,T_2) \le v,\, \min(T_1,T_2) \le u}$ and note that
\begin{align*}
\prob{\max(T_1,T_2) \le v,\, \min(T_1,T_2) \le u} &= \prob{\max(T_1,T_2) \le v} - \prob{\max(T_1,T_2)\le v,\, \min(T_1,T_2) > u}\} \\ 
& = \prob{T_1\le v,\, T_2\le v} - \prob{T_1\le v,\, T_2\le v,\,T_1>u,\, T_2>u}\\
& = \prob{T_1\le v}\prob{T_2\le v} - \prob{u < T_1\le v}\prob{u < T_2 \le v}\\
& = \prob{T_1\le v}^2 - \prob{u < T_1 \le v}^2\\
& = F(v)^2 - (F(v)-F(u))^2\\
& = 2F(u)F(v) - F(u)^2\,,
\end{align*}
for $u\le v$. 
Consequently, we can obtain the joint probability density of $X=\min(T_1,T_2)$ and $Y=\max(T_1,T_2)$ by differentiating the above joint cumulative probability.
\begin{align*}
p_{XY}(u,v) &= \frac{\partial^2 }{\partial u \partial v} \prob{\max(T_1,T_2) \le v,\, \min(T_1,T_2) \le u}\\
&=\frac{\partial^2 }{\partial u \partial v} 2F(u)F(v) - F(u)^2\\
&=\frac{\partial}{\partial u} 2 F(u) f(v)\\
&= 2 f(u) f(v)\,,
\end{align*}
which holds for $u\le v$. Conversely, for $u>v$ we have $p_{XY}(u,v) = 0$.
Therefore, in the case of a maximum of two independent variables, a joint distribution $f(u,v)$ is transformed to the product distribution of its marginal $f(u)=\int_{-\infty}^{\infty} f(u,v^\prime)dv^\prime$, which is given by $2f(u)f(v)$ times the indicator function $1[u\le v]$. The indicator enforces the constraint that the minimum is upper-bounded by the maximum.
\begin{align*}
f_{n+1} (u,v) = 
\begin{cases} 
2\int_{-\infty}^{\infty} f_n(u,v^\prime)dv^\prime \int_{-\infty}^{\infty} f_n(v,v^\prime)dv^\prime 1[u\le v] & \mbox{ for } b_n=+1\\
2\int_{-\infty}^{\infty} f_n(u^\prime,u)du^\prime \int_{-\infty}^{\infty} f_n(u^\prime,v)dv^\prime 1[u\le v] & \mbox{ for } b_n=-1\,.\\
\end{cases}
\end{align*}
We can simplify the above notation as follows. Where we defined the marginalization operator
\begin{align*}
M_{f}^{+}(u):=\int_{-\infty}^{\infty} f(u,v^\prime)dv^\prime.
\end{align*}
Here the superscript $+$ indicates that we are integrating over the first variable $u$. Similarly, define
\begin{align*}
M_{f}^{-}(v):=\int_{-\infty}^{\infty} f(u^\prime,v)du^\prime\,.
\end{align*}
then the update can be written equivalently as
\begin{align*}
f_{n+1} (u,v) = 
\begin{cases} 
2M^{+}_{f_n}(u) M^{+}_{f_n}(v) 1[u\le v] & \mbox{ for } b_n=+1\\
2M^{-}_{f_n}(u) M^{-}_{f_n}(v) 1[u\le v] & \mbox{ for } b_n=-1\,.\\
\end{cases}
\end{align*}
\end{proof}
\appendices
\section{}
In this section, we restate classical convergence results for scalar martingales. We refer the reader to \cite{neveu1975discrete} for a detailed exposition.

Let $(\Omega,\mathcal{A},\mathbb{P}$ be a probability space. We will denote by $\mathcal{B}_n,\,n\in \mathbb{N}$ an increasing sequence of of sub-$\sigma$-fields of $\mathcal{A}$. A sequence of random variables $X_n, \,n \in \mathcal{N}$ is called adapted if for all $n\in\mathbb{N}$, the random variable $X_n$ is $\mathcal{B}_n$-measurable.

An adapted sequence of integrable real valued random variables $X_n,\, n\in\mathbb{N}$ is called an integrable submartingale if the almost sure inequality
\begin{align*}
X_n \le \Exs^{\mathcal{B}_n}[X_{n+1}]
\end{align*}
holds for all $n\in\mathbb{N}$. Such a sequence of random variables is called an integrable martingale if the preceding inequality holds with equality. A collection of random variables $X_n, n\in I$ is called uniformly integrable if
\begin{align*}
\lim_{M\rightarrow \infty} \sup_{n\in I} \,\Exs[|X_i|\,|\, |X_i|>M] = 0\,.
\end{align*}
In this case $X_n$
As a simple corollary, if there exists $c>0$ such that $|X_n|\le c$ for all $n\in I$, then the collection $X_n,\, n\in I$ is uniformly integrable. Moreover, if $\Exs[|X_n|^k]$ is bounded for some $k>1$ for all $n\in I$, then the collection $X_n,\,n\in I$ is uniformly integrable. A martingale $X_n,\, n\in \mathbb{N}$ is called a uniformly integrable martingale if the collection of random variables $X_n,\, n\in \mathbb{N}$ is uniformly integrable.

\begin{theorem}{(Almost sure convergence of scalar martingales \cite{neveu1975discrete})}
\label{thm:martingalescalaralmostsure}
Every integrable submartingale $X_n,\,n\in\mathbb{N}$ satisfying the condition $\sup_{n}\, \Exs[\max(X_n,0)]<\infty$ converges almost surely to a limit $X_\infty$ which is an integrable random variable. In the case of an integrable martingale, this condition is given by $\sup_n \Exs[X_n]<\infty$.
\end{theorem}

\begin{theorem}{($L^p$ convergence of scalar martingales)}
\label{thm:martingalescalarlp}
Every martingale satisfying $\sup_n\, \Exs [|X_n|^p]<\infty$ for some $p>1$, converges almost surely to a limit $X_\infty=\limsup \limits_{n\rightarrow \infty} X_n$, and also in $L^p$, i.e., $\lim\limits_{n\rightarrow \infty} \Exs[|X_n-X_{\infty}|^p=0]$. If $X_n$ is a uniformly integrable martingale, then the convergence is in $L^1$, i.e., $\lim\limits_{n\rightarrow \infty} \Exs[|X_n-X_{\infty}|=0]$ and $\Exs[X_\infty\,|\, \mathcal{B}_n]=X_n$.
\end{theorem}

\section{}
In this section, we present some auxiliary technical results regarding the convergence of Banach space valued martingales.
\label{sec:appendixbanachmartingale}
The following result is an extension of Doob's martingale convergence theorem to Banach space valued martingales.
\begin{theorem}[Proposition V-2-8, page 107 \cite{neveu1975discrete}]
\label{thm:Banachmartingale}
Let $E$ be a separable Banach space which is the dual of a separable Banach space $F$. Then, every integrable martingale $\{ F_n, n \in \mathbb{N} \}$ with values in $E$ which satisfies 
\begin{align*}
\sup_{n \in \mathbb{N}}\, \Exs \,\| F_n \| < \infty\,,
\end{align*}
known as \emph{Doob's condition} converges almost surely to an integrable random variable $F_{\infty}$ with values in $E$.
\end{theorem}
In particular, every separable reflexive Banach space and every separable Hilbert space satisfies the requirements of the theorem. As an example $L_{p}$ spaces are reflexive provided that $1<p<\infty$. However, $L_1$ and $L_{\infty}$ spaces are not reflexive.

\section*{Acknowledgment}
This work was partially supported by the National Science Foundation
    under grants IIS-1838179, ECCS-2037304, DMS-2134248, the Army Research Office and Adobe Research. We would like to thank Erdal Arikan, Burak Bartan and Orhan Arikan for insightful discussions and their valuable feedback on this work. We are grateful to Burak Bartan for his assistance in numerical simulations and Neophytos Charalambides for proofreading the manuscript. We also thank anonymous reviewers for their constructive comments.

\ifCLASSOPTIONcaptionsoff 
  \newpage
\fi

\newpage
\bibliographystyle{ieeetr} 

\bibliography{mert}
\typeout{get arXiv to do 4 passes: Label(s) may have changed. Rerun}

\end{document}